\documentclass[10pt]{article}

\usepackage[utf8]{inputenc}

\usepackage[T1]{fontenc} 
\usepackage{mathtools,amsmath,amssymb,amsfonts,amsthm,siunitx,bm,bbm}
\usepackage{moreverb,rotating,graphics,esint}
\usepackage{mathrsfs}      
\usepackage{enumitem} 
\usepackage{titlesec}
\usepackage[a4paper, margin=2.7cm]{geometry}

\usepackage{dsfont}

\numberwithin{equation}{section}

\allowdisplaybreaks

\usepackage[colorlinks=true,linkcolor=blue,citecolor=blue,urlcolor=blue,breaklinks]{hyperref}

\usepackage{graphicx} 
\usepackage{mathtools}
\usepackage{mathabx}
\usepackage{bbm}
\usepackage{bigints}
\usepackage[dvipsnames]{xcolor}
\usepackage[toc,page]{appendix}
\usepackage{comment} 

\usepackage{multirow}

\usepackage{tikz}
\usetikzlibrary{decorations.pathreplacing}

\numberwithin{equation}{section}

\newtheorem{theorem}{Theorem}[section]
\newtheorem{proposition}[theorem]{Proposition}
\newtheorem{remark}[theorem]{Remark}
\newtheorem{lemma}[theorem]{Lemma}
\newtheorem{corollary}[theorem]{Corollary}
\newtheorem{definition}[theorem]{Definition}

\newtheorem{assumption}[theorem]{Assumption}

\usepackage{bigints}
\DeclareMathOperator*{\fiint}{\ensuremath{\iint\text{\kern-1.36em{\raisebox{5.87pt}{\rotatebox{-93}{$\setminus$}}}}}}

\def\N{{\mathbb N}}

\def\R{{\mathbb R}}
\def\C{{\mathbb C}}

\def\rL{{\mathrm L}}
\def\rS{{\mathrm S}}
\def\HL{{\cH_\rL}}

\def\1{{\mathds{1}}}

\def\cD{\mathcal{D}^c}
\def\cE{{\mathcal E}}
\def\cH{{\mathcal H}}
\def\cL{{\mathcal L}}
\def\cK{{\mathcal K}}

\def\cS{{\mathcal S}_c}

\def\cB{{\mathcal B}}
\def\cL{{\mathcal L}}
\def\cU{{\mathcal U}}
\def\cHn{{\mathcal{H}_{\rm L}}}

\def\cO{{\mathcal O}}

\def\Tr{{\rm Tr}}
\newcommand{\lc}{\left<}
\newcommand{\rc}{\right>}

   \DeclarePairedDelimiterX\Set[1]\{\}{%
      
      #1
}

\usepackage[style=numeric-comp,maxbibnames=9,natbib=true]{biblatex}
\AtEveryBibitem{%
  \clearfield{issn} % Remove issn
  \clearfield{doi} % Remove doi

  \ifentrytype{online}{}{% Remove url except for @online
    \clearfield{url}
  }
}

\usepackage{hyperref}
\title{On the relativistic effect in Dirac--Fock theory}

\author{ Long Meng\footnote{\textsc{%
      Long Meng,
      Center for Interdisciplinary Applied Mathematics \& Institute of Fundamental and Transdiciplinary Research,
      Zhejiang University,
      China
    }%
    (\texttt{\href{mailto:longmeng@zju.edu.cn}{longmeng@zju.edu.cn}})}
}

\date{}

\begin{document}

\maketitle

\begin{abstract}
    In this paper, we study the error bound between the Dirac--Fock ground-state energy and the Hartree--Fock ground-state energy, a quantity known as the relativistic effect in quantum mechanics. We confirm that the relativistic effect in the Dirac--Fock ground-state energy is of the order $\cO(c^{-2})$ with $c$ being the speed of light. Furthermore, if the potential between electrons and nuclei is regular, we get the well-known leading order relativistic correction -- the Breit--Pauli term, which is the sum of the mass-velocity term, the Darwin term, and the spin-orbit term. As a consequence, we also show that the same relativistic effects and leading order  relativistic correction also hold in a QED model introduced by Mittleman when the vacuum polarization -- a term of the order $\cO(c^{-3})$ -- is neglected.  To our knowledge, 
this is the first time in mathematics that the leading-order relativistic correction has been obtained from nonlinear Dirac ground-state energy problems.
\end{abstract}

\section{Introduction}
This work is part of a series of papers by the author on Dirac--Fock (DF) theory \cite{meng2024rigorous,crystals,catto2023properties,meng2024ground}. Here, we focus on the relativistic effect in the DF ground-state energy, that is the error bound between the DF ground-state energy and the Hartree--Fock (HF) ground-state energy. Even though we focus on two models in quantum physics, the main contribution is devoted to the study of the relativistic effects of nonlinear Dirac ground-state energy problems which are defined by min-max theory rather than by first principle.

\medskip

The emergence of relativistic quantum mechanics has been one of the most remarkable developments in quantum physics over the past century. Since Dirac's pioneering work, relativity has been a part of the quantum physical picture. While non-relativistic theories successfully describe quantum systems with particle velocities much smaller than the speed of light $c$, relativistic effects play a crucial role in high-precision calculations. Now it became clear that relativistic effects had an essential influence on a number of physical and chemical properties. 

Understanding how relativistic effects influence non-relativistic theories is crucial for bridging the gap between non-relativistic quantum mechanics and the more comprehensive relativistic framework. This is of particular importance for Fermions with spin $\pm \frac{1}{2}$ due to the complexity of the Dirac operator. 

\subsection{Linear Dirac eigenvalue problem}

As a small quantity -- especially for light atoms and molecules built from them -- the relativistic effect is a perturbation of the non-relativistic energies. Consequently, perturbation methods are developed in quantum mechanics, and the direct perturbation theory (see e.g., \cite{kutzelnigg2002perturbation,dyall2007introduction} for a review) is frequently used and has shown significant success in quantum chemistry. 

From a rigorous mathematical point of view, the perturbation theory used in quantum mechanics only works for linear Dirac eigenvalue problems. Its mathematical justification is based on the holomorphy of the resolvent of the Dirac operators with respect to $c^{-2}$ \cite{bulla1992holomorphy}. We refer to \cite[Chp. 6]{thaller2013dirac} as a mathematical review of this perturbation argument.

Using perturbation theory, the eigenvalues of the Dirac operators can be expressed as a Taylor expansion of $c^{-2}$. Among these terms, the leading relativistic term (i.e., the Brei--Pauli term which is of the order $\cO(c^{-2})$) plays an important role in quantum physics. When the nuclei are not very heavy, the leading order correction term is the main contribution of the relativistic effect. It is shown in quantum chemistry that for rare gases (i.e., the number of nuclear charges $z=2,10,18,\cdots,86$),  the leading order correction term accounts for at least $80\%$ of the total relativistic effects \cite{ottschofski1995relativistic}.

\subsection{Nonlinear Dirac ground-state problem}

Concerning the nonlinear Dirac bound-state energies, the study of the relativistic effect is more delicate. A fundamental difficulty stems from the fact that the Dirac operator is unbounded both from above and from below. In physics, Quantum Electrodynamics (QED) provides the theoretical framework for describing such systems. However, this theory lacks a rigorous mathematical foundation due to divergence issues. Thus in physics and chemistry, the ground-state energy is mainly defined by the min-max theory and it is indeed an energy of a special excited state from the point of view of quantum chemistry, see e.g., \cite{saue2011relativistic} in physics.

%Thus it is impossible to define the ground-state energy as a direct minimum problem without ultra-violet cutoff due to the problem of divergence. Thus in physics and mathematics, the ground-state energy of nonlinear Dirac problem is frequently defined by the min-max theory. We refer to \cite{EstebanSere-overview,EstebanLewinSere-variational} and references therein as a review.

The DF model, first introduced in \cite{swirles1935relativistic}, is widely used in computational chemistry. It is a variant of the HF model in which the kinetic energy operator $-\frac{1}{2}\Delta$ is replaced by the free Dirac operator $\cD$. Even though in principle it is not physically meaningful, this approach gives remarkably accurate results that are in excellent agreement with experience data (see, e.g., \cite{desclaux1973relativistic,gorceix1987multiconfiguration}). 

Mathematically the DF ground-state energy can be defined through min-max theory \cite{barbaroux2005some,esteban2001nonrelativistic,Esteban-Sere-minmax-DF}:
\begin{align}\label{eq:DF-min-max}
 e_{c,q}=  \min_{\substack{ \Phi\in G_q(H^{1/2})\\ \Phi\textrm{ solution of DF equations}}}\cE_c(\gamma_\Phi)= \inf_{\substack{V\subset \Lambda^+_c H^{1/2}\\ \dim(V)=q}}\;\;\sup_{\substack{\Phi\in (\Lambda^-_c H^{1/2}\oplus V)^N \\ 0\leq {\rm Gram}_{L^2}\Phi\leq \1_q}}\cE_c(\gamma_\Phi).
\end{align}
Solutions of DF equation have been studied in \cite{esteban1999solutions,paturel2000solutions}. Here $\gamma_\Phi$ is the density matrix associated with $\Phi:=(u_1,\cdots,u_q)\in G_q$ defined by its kernel
\begin{align*}
  \gamma_{\Phi}(x,y)=\sum_{j=1}^q u_j(x)\otimes u_j^*(y).
\end{align*}
The space $G_q$ is the functional space presenting the wavefunctions of $q$ electrons and is a Grassmannian manifold defined by
\[
G_q(H^{1/2}):=\{G\textrm{ subspace of } H^{1/2}(\mathbb{R}^3;\mathbb{C}^4);\,\dim(G)=q\}
\]
where $q$ is the number of electrons, and the $q\times q$ matrix ${\rm Gram}_{L^2}\Phi$ is defined by 
\begin{align*}
    ({\rm Gram}_{L^2}\Phi)_{j,k}=\int_{\R^3}u_k^*(x)u_j(x)dx.
\end{align*}
The operator $\Lambda^+_c$ (resp. $\Lambda^-_c$), defined by \eqref{eq:2.3} is the positive (resp. negative) projector of the free Dirac operator.

Eq. \eqref{eq:DF-min-max} means that minimizers are critical points of nonlinear Dirac functional. As a critical point, DF minimizers also satisfy the following nonlinear constraint:
\begin{align}\label{eq:non-constraint}
    \gamma =P^+_{c,\gamma} \gamma P^+_{c,\gamma}
\end{align}
where $P^+_{c,\gamma}$,  defined by \eqref{eq:P-pm}, is the projector of the positive subspace of the DF operator. From QED point of view,
\begin{align*}
    P^-_{c,\gamma}:=1-P^+_{c,\gamma}
\end{align*}
is an approximation of Dirac sea, which corresponds to the state of infinite virtual electrons that occupies the whole negative spectrum of some Dirac operators. Eq. \eqref{eq:non-constraint} implies that electrons associated with the state $\gamma$ in DF theory are away from the Dirac sea generated by $\gamma$ itself when the vacuum polarization is ignored.

Based on \eqref{eq:non-constraint}, recently S\'er\'e \cite{sere2023new} introduced a new definition of the DF ground-state energy:
\begin{align}\label{eq:1.5}
    E_{c,q}:=\min_{\substack{\gamma\in \Gamma_{q}\\ \gamma =P^+_{c,\gamma} \gamma P^+_{c,\gamma}}}\mathcal{E}_c(\gamma).
\end{align}
This corresponds to \eqref{eq:min-DF} and $\Gamma_q$ is the state of electrons defined by \eqref{eq:state-electrons}. We show in Corollary \ref{cor:definition-equiv} below and \cite{meng2024rigorous} that these two definitions of DF ground-state energy are equivalent. More precisely, in Corollary \ref{cor:definition-equiv}, we will show that for $c$ large enough and $q\leq z$,
\begin{align}\label{eq:1.6}
    E_{c,q}=e_{c,q}.
\end{align}

In this paper, we study the relativistic effect in the DF ground-state energy by using the definition \eqref{eq:1.5}. Then, Eq. \eqref{eq:1.6} implies that the result is equivalent to the study of the relativistic effect in the DF ground-state energy defined by the min-max theory \eqref{eq:DF-min-max}.

\subsection{Mittleman's definition of QED}

As explained above, the DF model is not bounded from below. This deficiency also raises questions about its physical derivation: one would like to show that the DF model or its refined variants can be interpreted as an approximation of QED, see e.g., \cite{mittleman1981theory,sucher1980foundations,chaix1989quantum} and the references therein. Steps in the direction of a rigorous justification are undertaken by considering the Mittleman's definition of QED, see \cite{mittleman1981theory}. In addition, Mittleman's definition of QED gives a way to understand the physical meaningful definition of Dirac sea in relativistic quantum many-body system, see e.g., \cite{meng2024ground}.

Mittleman's QED model is defined as follows (see, e.g., \cite[Eq. (12)]{meng2024rigorous} or \cite[Eq. (12)]{barbaroux2005some}):
\begin{equation*}
    E^{\rm M}_{c,q}:=\sup_{P_c^-\in \mathcal{P}_c}E^{(P_c^-)}_{c,q}
\end{equation*}
and
\begin{align*}
    E^{(P_c^-)}_{c,q}=\min_{\gamma\in \Gamma^{(P_c^-)}_q}\mathcal{E}_c(\gamma).
\end{align*}
Here $\Gamma_q^{(P^-_c)}$, defined by \eqref{eq:electron-position-field} is the set of admissible state of electrons for a fixed Dirac sea $P^-_c$; and $\mathcal{P}_c$, defined by \eqref{eq:dirac-sea}, is the set of Dirac sea in electron-positron Hartree--Fock theory \cite{barbaroux2005hartree}.

In physics, this model states that when the vacuum polarization is neglected, the real physical ground state energy should be obtained by maximizing the ground state energy of HF-type models over all allowed one-particle electron subspace. We refer to \cite[Section 4.5]{EstebanLewinSere-variational} for more details on the formal derivation of this model from Bogoliubov-Dirac-Fock model, another QED model with vacuum polarization. In particular, as the vacuum polarization is neglected,  no divergence problem is encountered. Thus we do not need to introduce the ultraviolet cutoff.

It is shown in \cite{barbaroux2005hartree} that, under some assumptions, $E^{(P_c^-)}_{c,q}$ admits a minimizer for a given $P^-_c$. However, as a maximum problem for a set of non-compact operators, $ E^{\rm M}_{c,q}$ may not have a maximizer. Nevertheless, the relationship between Mittleman's definition and the DF model has been studied in \cite{barbaroux2005hartree,huber2007solutions,barbaroux2005remarks,meng2024rigorous}. Together with \cite{meng2024rigorous}, we show in Corollary \ref{cor:Mittleman} below that for $q\leq z$ and $c$ large enough,
\begin{align*}
    |E^{\rm M}_{c,q}-E_{c,q}|=\cO(c^{-4}).
\end{align*}

We also point out that in the full QED theory, typical QED effects such as vacuum polarization is of the size $\mathcal{O}(c^{-3})$ (see, e.g., \cite[Chap. 5.5]{dyall2007introduction}). Thus the ground state energy $E^{\rm M}_{c,q}$ due to Mittleman and the DF ground state energy describes the full QED model up to an error bound of the order $\mathcal{O}(c^{-3})$.

\subsection{Relativistic effect in nonlinear Dirac ground-state energy}

The relationship between nonlinear Dirac problems and their non-relativistic counterparts has been an active area of mathematical research, see, e.g., \cite{esteban2001nonrelativistic,Pan2024,borrelli2019nonlinear}. However, they mainly focus on the non-relativistic limit of the bound-state solutions, that is the solutions of Dirac equations will converge to the solutions of some non-relativistic Schr\"odinger/Pauli equations. Concerning the bound-state energy, this implies that
\begin{align*}
     \boxed{\mbox{Dirac bound-state energy}}=  \boxed{\mbox{Schr\"odinger/Pauli bound-state energy}}+o_{c\to \infty}(1).
\end{align*} 
Such results only provide asymptotic convergence of energies without quantifying the relativistic effects.  

For the charge of nuclei $z$ large enough and some special number of electrons $q$, it is shown in \cite{huber2007solutions} that the Dirac--Fock ground-state is unique if $c$ is large enough by using fixed point method. A unique minimizer of Hartree--Fock ground-state is constructed in \cite{griesemer2012unique} in the same manner. In these special cases, the fixed point implies that one can expand the Dirac--Fock ground-state energy in power of $c^{-1}$. However, the fixed point method is limited and it can not be used to study the relativistic effects of general $q$ and $z$.

In this paper, we focus on the relativistic effect in the DF ground-state energy for general $q$ and $z$, confirming predictions from relativistic quantum chemistry (i.e., Theorem \ref{th:rela-effect}), that is
\begin{align}
   E_{c,q}=E_q^{\rm HF}+\cO(c^{-2})
\end{align}
with $E_q^{\rm HF}$ being the HF ground-state energy defined by \eqref{eq:min-HF}. Furthermore, if the potential between electrons and nuclei is regular (that is the potential satisfies Assumption \ref{ass:V}), then we obtain the leading order relativistic correction term (see Theorem \ref{th:rela-correction}): there exists a term $E^{(2)}_c$ of the order $\cO(c^{-2})$ such that
\begin{align}
  E_{c,q}=E_q^{\rm HF}+ E^{(2)}_c+o(c^{-2}).
\end{align}
Here the leading order correction term $E^{(2)}_c$ can be decomposed into three terms: the mass-velocity term, the Darwin term and the spin-orbital term (see Proposition \ref{prop:rela-decomp}).

Concerning the QED model defined by Mittleman, the study of the relativistic effect in DF model also imply that
\begin{align*}
    |E^{\rm M}_{c,q}-E_{q}^{\rm HF}|=\cO(c^{-2})
\end{align*}
and under Assumption \ref{ass:V},
\begin{align*}
    E^{\rm M}_{c,q}= E_{q}^{\rm HF}+ E^{(2)}_c+o(c^{-2}).
\end{align*}
This is indeed Corollary \ref{cor:mittleman-rela}.

As mentioned above, typical QED effects such as vacuum polarization is of the size $\mathcal{O}(c^{-3})$ (see, e.g., \cite[Chap. 5.5]{dyall2007introduction}). Thus the same relativistic effect and the corresponding leading order relativistic correction are expected to hold in full QED models from a physical point of view. 

\medskip

To our knowledge, 
this is the first time in mathematics that the leading-order relativistic correction has been obtained from nonlinear Dirac ground-state energy problems. In addition, if the HF minimizers are not unique, the leading-order correction term $E^{(2)}_c$ is defined as the minimum energy of the Breit–Pauli term over all HF minimizers (see \eqref{eq:E2-minimum}). It should be noted that in quantum chemistry, the formula for $E^{(2)}_c$ has only been derived formally for the non-degenerate case of linear Dirac operators ( see e.g., \cite[Section 5]{kutzelnigg2002perturbation}), the correct expression for $E^{(2)}_c$ has not been obtained in non-degenerate cases \cite[Section 6]{kutzelnigg2002perturbation}. This failure is due to an orthogonalization issue. In this paper, we solved this issue and obtained the correct term $E^{(2)}_c$ for general cases by introducing some overlap matrices, see Section \ref{sec:3.4} for more details.

Moreover, even though we only consider the DF ground-state energy problem, min-max definitions analogous to \eqref{eq:DF-min-max} are also used in other nonlinear Dirac problem, see e.g., \cite{CotiZelati,Pan2024}. We can also define the ground-state energy of these nonlinear Dirac problem analogous to \eqref{eq:1.5}. Once the equivalence between the ground-state energy defined by the min-max method \eqref{eq:DF-min-max} and by the projection method \eqref{eq:1.5} are proven, we expect that the method of this paper can also be used to study the relativistic effects of these nonlinear Dirac ground-state energies that are defined by min-max theory.

\subsection{Other relevant works}\label{sec:1.4}
Finally, we end the introduction by summarizing some relevant problems of the DF model.

Concerning the Dirac-Fock type theory, there is another QED model called Bogoliubov-Dirac-Fock model that has been studied in mathematics, see e.g., \cite{Chaix_1989,chaix1989quantum,hainzl2005existence,hainzl2007mean,hainzl2009existence,Gravejat-Lewin-Sere-groundstate}. In this model, the vacuum polarization is considered. Thus our method on relativistic effects does not work for this problem. However, as the vacuum polarization is of the order $\cO(c^{-3})$, results similar to Theorem \ref{th:rela-effect} and Theorem \ref{th:rela-correction} are expected in Bogoliubov-Dirac-Fock model. This will be studied later.

We also mention the relativistic Scott conjecture which concerns the difference between relativistic many-body problem (including Dirac--Fock model) and the Thomas--Fermi model. In the regime that $q=z\to \infty$ and $0<z/c$ remains fixed and small, the relativistic many-body ground-state energy can be approximated by the Thomas--Fermi energy of order $\mathcal{O}(z^{7/3})$ with a relativistic Scott correction of order $\mathcal{O}(z^2)$. Depending on the choice of the Dirac sea, different relativistic effects can be observed in Scott corrections. We refer to \cite{frank2023scott,fournais2020scott,MR2588386} and references therein for the mathematical study of the Scott correction, and we also point out that a ``wrong'' physical projection will lower the relativistic effect in Scott correction by using Mittleman's definition of QED~\cite{meng2024ground} .

\medskip

{\bf Organization of the paper.} In Section \ref{sec:2}, we recall the definition of the Dirac operator and Schr\"odinger operator, and introduce some useful functional spaces and other notations. In Section \ref{sec:3}, we recall some mathematical results on the Dirac-Fock theory and Hartree-Fock theory, and then we state our mains results (i.e., Theorem \ref{th:rela-effect}, Theorem \ref{th:rela-effect} and Corollary \ref{cor:mittleman-rela}). In Section \ref{sec:3.4}, we give a sketch of proof of Theorem \ref{th:rela-correction}, the proof of Theorem \ref{th:rela-effect} follows a similar manner. Sections \ref{sec:5}-\ref{sec:8} are denoted to the proof of Theorem \ref{th:rela-effect} and Theorem \ref{th:rela-correction}. The organization these sections is explained in Section \ref{sec:3.4}.

\section{Dirac and Schr\"odinger operator, functional space and some notations}\label{sec:2}

\subsection{Free Dirac and Schr\"odinger operator}

In this paper, for simplicity we assume the mass of particle $m=1$.

In non-relativistic quantum mechanics with spin $\pm\frac{1}{2}$, the free Schr\"odinger operator is defined by 
\begin{align}
    H_0:=-\frac{1}{2}\Delta.
\end{align} 
The operator $H_0$ acts on $2-$spinors; that is, on functions from $\R^3$ to $\C^2$. It is self-adjoint in $L^2(\R^3;\C^2)$ (which is equivalent to $\cH_{\rm L}$ defined in \eqref{eq:H-fun-non}  below), with domain $H^2(\mathbb{R}^3;
\mathbb{C}^2)$ and form domain $H^{1}(\mathbb{R}^3;\mathbb{C}^2)$. Its spectrum is $\sigma(H_0)=[0,+\infty)$.

\medskip

In the framework of relativistic quantum mechanics, the Schr\"odinger operator should be replaced by the free Dirac operator which is defined by 
\[
\cD=-ic\sum_{k=1}^3\alpha_k\partial_k+c^2\beta\]
with the speed of light $c$ and $4\times4$ complex matrices $\alpha_1,\alpha_2,\alpha_3$ and $\beta$, whose standard forms are:
\[
\beta=\begin{pmatrix} 
\1_2 & 0 \\
0 & -\1_2 
\end{pmatrix},\,
\alpha_k=\begin{pmatrix}
0&\sigma_k\\
\sigma_k&0
\end{pmatrix},
\]
where $\1_2$ is the $2\times 2$ identity matrix and the $\sigma_k$'s, for $k\in\{1,2,3\}$, are the well-known $2\times2$ Pauli matrices
\[
\sigma_1=\begin{pmatrix}
0&1\\
1&0
\end{pmatrix},\,
\sigma_2=\begin{pmatrix}
0&-i\\
i&0
\end{pmatrix},\,
\sigma_3=\begin{pmatrix}
1&0\\
0&-1
\end{pmatrix}.
\]
These algebraic conditions are here to ensure that $\cD$ is a symmetric operator, such that
\begin{align}
    (\cD)^2=c^4-c^2\Delta.
\end{align}

The operator $\cD$ acts on $4-$spinors; that is, on functions from $\R^3$ to $\C^4$. It is self-adjoint in $L^2(\R^3;\C^4)$, with domain $H^1(\mathbb{R}^3;
\mathbb{C}^4)$ and form domain $H^{1/2}(\mathbb{R}^3;\mathbb{C}^4)$. Its spectrum is $\sigma(\cD)=(-\infty,-c^2]\cup[+c^2,+\infty)$. Following the notation in \cite{esteban1999solutions,paturel2000solutions}, we denote by $\Lambda^+_c$ and $\Lambda^-_c=\1_{\mathcal{H}}-\Lambda^+_c$ respectively the two orthogonal projectors on $\mathcal{H}$ corresponding to the positive and negative eigenspaces of $\cD$; that is
\[
\begin{cases}
\cD\Lambda^+_c=\Lambda^+_c\cD=\Lambda^+_c\sqrt{c^4-c^2\Delta}=\sqrt{c^4-c^2\Delta}\,\Lambda^+_c;\\
\cD\Lambda^-_c=\Lambda^-_c\cD=-\Lambda^-_c\sqrt{c^4-c^2\Delta}=-\sqrt{c^4-c^2\Delta}\,\Lambda^-_c.
\end{cases}
\]
More precisely, $ \Lambda^\pm_c$ can be defined by
\begin{align}\label{eq:2.3}
    \Lambda^\pm_c:=\1_{\R^\pm}(\cD)= \frac{1}{2}\pm \frac{\cD}{2|\cD|}.
\end{align}

For any relativistic one-body wavefunction $u$ in $L^2(\R^3;\C^4)$, we set\footnote{In relativistic quantum chemistry, $u^{\rm L}$ is the large component and $u^{\rm S}$ is the small component.}
\begin{align}
    u(x):=\begin{pmatrix}
        u^\rL(x)\\ u^\rS(x)
    \end{pmatrix},\qquad u^\rL,u^\rS:\;\R^3\to \C^2.
\end{align}
Let
\begin{align}\label{eq:cL}
  \mathcal{L}=-i{\pmb \sigma}\cdot \nabla=-i\sum_{k=1}^3\sigma_k\partial_k,\qquad   \pmb \sigma=(\sigma_1,\sigma_2,\sigma_3);
\end{align}
let
\begin{align}\label{eq:S-transform}
    \cS u:=\begin{pmatrix}
        u^L(x)\\ \frac{1}{2c}\cL u^L(x)
    \end{pmatrix}
\end{align}
and
\begin{align}\label{eq:K-transform}
    \cK_\rL u:=\begin{pmatrix}
        u^\rL(x)\\ 0
    \end{pmatrix},\qquad \cK_\rS u:=\begin{pmatrix}
         0\\ u^\rS(x)
    \end{pmatrix}.
\end{align}
Then the Dirac operator and Schr\"odinger operator have the following relationship:
\begin{align}\label{eq:D-Delta}
    (\cD-c^2) \cS u= \begin{pmatrix}
        -\frac{1}{2}\Delta u^L\\  0
    \end{pmatrix} =H_0 \cK_\rL u
\end{align}
since $\cL^2=-\Delta$ on $H^2(\R^3;\C^2)$. This formula plays the essential role in the paper.

\subsection{Functional spaces and density matrices}\label{sec:2.2}

Throughout the paper, we restrict ourselves to $4$-spinors, and we denote $\cH:=L^2(\R^3;\C^4)$ and $H^{s}:=H^{s}(\R^3;\C^4)$ with any $s\in \R$. In non-relativistic quantum mechanics, the state of electrons is described by $u^{\rm L}$, and we have $u^{\rm S}= 0$. Thus for the one-body wavefunctions in non-relativistic quantum mechanics, we replace the functional space $L^2(\R^3;\C^2)$ by the following subspace of $\cH$:
\begin{align}\label{eq:H-fun-non}
    \cHn:=\cK_{\rm L} \cH=\{u\in \cH; u^{\rm S}=0\}\cong L^2(\R^3;\C^2).
\end{align}

Let $\mathcal{B}(W_1,W_2)$ be the space of bounded linear maps from a Banach space $W_1$ to a Banach space $W_2$, equipped with the norm 
\[
\|A\|_{\mathcal{B}(W_1,W_2)}:=\sup_{u\in W_1,\,\|u\|_{W_1}=1}\|Au\|_{W_2}.
\]
We denote $\mathcal{B}(W):=\mathcal{B}(W,W)$. The functional space $\mathfrak{S}_p:=\mathfrak{S}_p(\mathcal{H})$ for $p\in [1,\infty)$ is defined by
\begin{align*}
    \mathfrak{S}_p:=\{\gamma\in\mathcal{B}(\mathcal{H});\Tr[|\gamma|^p]<+\infty\},
\end{align*}
endowed with the norm
\begin{align*}
    \|\gamma\|_{\mathfrak{S}_p}^p:=\Tr[|\gamma|^p].
\end{align*}
We also define
\[
X^s:=\{\gamma\in\mathcal{B}({\mathcal{H}}); \gamma=\gamma^*, (1-\Delta)^{s/4}\gamma(1-\Delta)^{s/4}\in\mathfrak{S}_1\},
\]
endowed  with the norm
\[
 \|\gamma\|_{X^s}:=\|(1-\Delta)^{s/4}\gamma(1-\Delta)^{s/4}\|_{\mathfrak{S}_1}.
\]
In particular, we denote $X:=X^1$.  For any $\gamma\in X$, we also introduce the following $c$-dependent norm:
\[
\|\gamma\|_{X_c}:=\||\cD|^{1/2}\gamma|\cD|^{1/2}\|_{\mathfrak{S}_1}=\|(c^4-c^2\Delta)^{1/4}\gamma(c^4-c^2\Delta)^{1/4}\|_{\mathfrak{S}_1}.
\]

For every density matrix $\gamma\in X$, there exists a complete set of eigenfunctions $(u_n)_{n\geq 1}$ of $\gamma$ in $\mathcal{H}$, corresponding to the non-increasing sequence of eigenvalues $(\lambda_n)_{n\geq 1}$ (counted with their multiplicity) such that $\gamma$ can be rewritten as
\begin{align}\label{eq:gamma-dec}
    \gamma=\sum_{n\geq 1}\lambda_n |u_n \rangle\,\langle u_n|
\end{align}
where $|u \rangle\,\langle u|$ denotes an operator onto the vector space spanned by the function $u$: for any $\psi\in \cH$,
\begin{align*}
    |u \rangle\,\langle u| \psi:= \left<u,\psi\right>_\cH u.
\end{align*}
The kernel $\gamma(x,y)$ of $\gamma$ reads as
\begin{align*}
    \gamma(x,y)=\sum_{n\geq 1}\lambda_n u_n(x)\otimes  u_n^*(y)
\end{align*}
The one-particle density associated with $\gamma$ is
\[
\rho_{\gamma}(x):=\mathrm{Tr}_{\mathbb{C}^4}[\gamma(x,x)]=\sum_{n\geq 1}\lambda_n |u_n(x)|^2,
\]
where the notation $\mathrm{Tr}_{\mathbb{C}^4}$ stands for the trace of a $4\times4$ matrix. 

In the non-relativistic setting, the density matrix $\gamma=\gamma^*$ is situated in $\mathcal{B}(\cH,\cHn)\cap X^2$ and we rather write
\[
\gamma=\sum_{n\geq 1}\lambda_n |u_n \rangle\,\langle u_n|
\]
with $(u_n)_{n\geq 1}$ in $\cHn$. Here we use the notation $\mathcal{B}(\cH,\cHn)\cap X^2$ to represent the set of non-relativistic self-adjoint density matrix, since any self-adjoint operator $A\in \mathcal{B}(\cH,\cHn)$ implies that its kernel $A(x,y)$ is of the form
\begin{align*}
    A(x,y)=\begin{pmatrix}
        A^{\rL,\rL}(x,y)& 0_{2\times 2}\\
        0_{2\times 2}& 0_{2\times 2}
    \end{pmatrix},\qquad  A^{\rL,\rL}=(A^{\rL,\rL})^*:\R^3\times \R^3\to {\rm Mat}_{2\times 2}(\C).
\end{align*}
where ${\rm Mat}_{q\times q}(\C)$ represents the set of $n\times n$ matrices with complex elements with $n\in \N^+$.
\subsection{Abuse of notations}
The following notations will also be used in the whole paper. Let $A=(a_{m,n})_{1\leq m,n\leq n}\in {\rm Mat}_{n\times n}(\C)$ be any $n\times n$ matrix, and $\{u_1,\cdots,u_n\}$ be some functions in $\cH$. Then by using a formal matrix operation, we define
\begin{align}\label{eq:matrix-bra-ket}
\Big(\left|u_1\right>,\cdots,\left|u_n\right>\Big)A \begin{pmatrix}
        \left<u_1\right|\\
        \vdots\\
        \left<u_n\right|
    \end{pmatrix}:=\sum_{1\leq j,k\leq n}a
    _{j,k}\left|u_j\right>\left<u_k\right|
\end{align}
and
\begin{align}\label{eq:matrix-inner}
\Big(\left<u_1\right|,\cdots,\left<u_n\right|\Big)A \begin{pmatrix}
        \left|u_1\right>\\
        \vdots\\
        \left|u_n\right>
    \end{pmatrix}:=\sum_{1\leq j,k\leq n}a_{j,k}\left<u_j,u_k\right>_\cH.
\end{align}
In addition, we will abuse the notation  $\cL u$ for any $u\in \cH$:
\begin{align*}
    \cL u:= \begin{pmatrix}
        \cL u^\rL\\
        \cL u^\rS
    \end{pmatrix}.
\end{align*}
In this paper, above notation $\cL u$ will only be used for non-relativistic $2$-spinors wavefunction $u$ (i.e., $u=\cK_\rL u$ or $u=\cK_\rS u$), thus
\begin{align*}
     \cL u:= \begin{pmatrix}
        \cL u^\rL\\
        0
    \end{pmatrix},\qquad \mbox{or}\qquad  \cL u:= \begin{pmatrix}
       0\\
        \cL u^\rS
    \end{pmatrix}.
\end{align*}
Analogously, we define $\cL \gamma$ in the same manner for any density matrix $\gamma\in X^2$: using \eqref{eq:gamma-dec}, we can write
\begin{align*}
    \cL \gamma \cL= \sum_{n\geq 1}\lambda_n \left|\cL u_n\right>\left<\cL^* u_n\right|= \sum_{n\geq 1}\lambda_n \left|\cL u_n\right>\left<\cL u_n\right|
\end{align*}
with $u_n\in H^1$. The notation $\cL \gamma$ will only be used for non-relativistic density matrices from $\cH_\rL$ to $\cH_\rL$.

\section{Models and main results}\label{sec:3}

\subsection{The DF and HF models for atoms and molecules}
We now recall the DF theory and the HF theory.
\subsubsection{The DF and HF operators and functionals}
For any $\gamma\in X$, the DF functional is defined by
\begin{align}\label{eq:DFfunctional}
    \mathcal{E}_c(\gamma):=\Tr_{\cH}[(\cD -c^2)\gamma]-\Tr_{\cH}[V\gamma]+\frac{1}{2}\Tr_{\cH} [W_{\gamma}\gamma],
\end{align}
while for any $\gamma\in  \cB(\cH,\cH_\rL)\cap X^2$, the HF functional is defined by
\begin{align}\label{eq:HFfunctional}
    \mathcal{E}^{\rm HF}(\gamma):=\Tr_{\cH}[H_0 \gamma]-\Tr_{\cH}[V\gamma]+\frac{1}{2}\Tr_{\cH} [W_{\gamma}\gamma]
\end{align}
where for any $\psi\in H^{1/2}$,
\begin{align}\label{eq:W-dec}
    W_{\gamma}\psi(x)=W_{1,\gamma}\psi(x)-W_{2,\gamma}\psi(x)
\end{align}
with 
\begin{align}
 W_{1,\gamma}\psi(x):=(\rho_{\gamma}*W)\psi(x), \qquad  W_{2,\gamma}\psi(x):= \int_{\mathbb{R}^3}W(x-y)\gamma(x,y)\psi(y)dy.
\end{align}
Here $V$ is the attractive potential between nuclei and electrons, and $W$ is the repulsive potential between electrons. We  consider the electrostatic case $W= \frac{1}{|x|}$ and $V=\mu*\frac{1}{|x|}$ with a nonnegative nuclear charge distribution $\mu\in \mathcal{M}_+(\mathbb{R}^3)$ satisfying $\int_{\mathbb{R}^3} d\mu=z$.

\medskip

Based on DF and HF functionals, the corresponding DF operator is defined by
\begin{align}\label{eq:DF}
    \cD_\gamma:=\cD -V+W_\gamma
\end{align} 
while the corresponding HF operator is defined by 
\begin{align}\label{eq:HF}
    H_{0,\gamma}:=\cK_\rL\Big(H_0- V+W_\gamma\Big)\cK_\rL.
\end{align} 

Before going further, we study the operator $W_{2,\bullet}$. Let $u,v\in H^1$, then
\begin{align}\label{eq:Wpotential-decom}
 \lc W_{2,\left| v\rc\lc v\right|} u, u\rc&= \int_{\R^3}\int_{\R^3}\frac{\big[u^*(x)v(x)\big]\big[v^*(y)u(y)\big]}{|x-y|}dxdy\notag\\
    &=\int_{\R^3}\int_{\R^3}\sum_{j,j'\in\{{\rm S, L}\} }\frac{\big[(\cK_j u)^*(x)(\cK_j v)(x)\big]\big[(\cK_{j'}v)^*(y)(\cK_{j'}u)(y)\big]}{|x-y|}dxdy\notag\\
    &=\sum_{j,j'\in\{{\rm S, L}\} } \lc W_{2,\cK_j\left| v\rc\lc v\right|\cK_{j'} }\cK_{j'} u, \cK_j u\rc.
\end{align}
In particular, for $j,j'\in\{{\rm S, L}\}$,
\begin{align}\label{eq:W-K1-K2}
    \Tr_{\cH}[W_{2,\cK_{j}\gamma_1 \cK_{j'}}\gamma]=\Tr_{\cH}[W_{2,\cK_{j}\gamma_1 \cK_{j'}}\cK_{j'}\gamma\cK_j].
\end{align}
Indeed, we also have
\begin{align}\label{eq:V-K1-K2}
    \Tr_\cH[V \gamma]=\sum_{j\in\{{\rm S, L}\} } \Tr_\cH[V \cK_j \gamma \cK_j],\qquad  \Tr_\cH[W_{1,\gamma'} \gamma]=\sum_{j\in\{{\rm S, L}\} } \Tr_\cH[W_{1,\gamma'} \cK_j \gamma \cK_j].
\end{align}

\subsubsection{Dirac--Fock ground state energy}

Let $q\in \N^+$ be the number of electrons. Let 
\begin{align}\label{eq:state-electrons}
    \Gamma:=\{\gamma\in X; 0\leq \gamma\leq  \1_{\mathcal{H}}\},\quad \Gamma_{q}=:\{\gamma\in \Gamma; \Tr[\gamma]\leq q\},
\end{align}
and let 
\begin{align}\label{eq:P-pm}
    P^{+}_{c,\gamma}=\1_{(0,+\infty)}(\mathcal{D}^c_{\gamma}),\quad  P^{-}_{c,\gamma}=\1_{(-\infty,0]}(\mathcal{D}^c_{\gamma}).
\end{align}
In the DF theory, the relevant set of electronic states is defined by
\[
\Gamma_{q}^+:=\{\gamma\in\Gamma_{q}; P^{+}_{c,\gamma}\gamma P^{+}_{c,\gamma}=\gamma\}.
\]
According to \cite{sere2023new}, the ground state energy of DF model can be redefined by 
\begin{align}\label{eq:min-DF}
    E_{c,q}:=\min_{\gamma\in \Gamma_{q}^+}\mathcal{E}_c(\gamma).
\end{align}

Before going further, we need the following assumption.
\begin{assumption}\cite[Theorem 1.2 and Remark 1.3]{sere2023new}\label{ass:1}
Let $\kappa_c:=2c^{-1}(q+z)$ and $R_c^{\rm DF}:= (1-\kappa_c-\frac{\pi}{4}c^{-1}q)^{-1/2}q+1$. Assume that 
\begin{align*}
    \kappa_c< 1-\frac{\pi}{4}c^{-1} q,\qquad R_c^{\rm DF} <  \frac{1}{2a_c}
\end{align*}
with $a_c:= \frac{\pi}{4c\sqrt{(1-\kappa_c)\lambda_{0,c}}}$ and $\lambda_{0,c}:= (1-c^{-1}\max( q,Z))$.
\end{assumption}
The existence of a ground state is guaranteed by the following.
\begin{theorem}[Existence of minimizers in the DF theory {\cite{sere2023new}}]\label{th:min-DF} Let $q\in \R^+$ and $z\in \R^+$ be fixed such that $q\leq z$. Then under Assumption \ref{ass:1} on $c$, the minimum problem \eqref{eq:min-DF} admits a minimizer $\gamma_*^c \in \Gamma_q^+$. In addition, $\Tr[\gamma_*^c]=q$, and any such minimizer can be written as
\begin{align}\label{eq:gamma-DF}
    \gamma_*^c=\1_{(0,\nu_c)}(\mathcal{D}^c_{\gamma_*^c})+\delta_c
\end{align}
with $
0< \delta_c\leq \1_{\{\nu_c\}}(\mathcal{D}^c_{\gamma_*^c})$ for some $\nu_c\in (0,c^2]$. When $ q<z$, $\nu_c\in (0,c^2)$.
\end{theorem}

\begin{remark}\label{rem:R}
In \cite{sere2023new} the second condition in Assumption \ref{ass:1} is expressed by
    \begin{align*}
        (1-\kappa_c-\frac{\pi}{4}c^{-1}q)^{-1/2}q<R_c^{\rm DF} <\frac{1}{2a_c}.
    \end{align*}
The condition $(1-\kappa_c-\frac{\pi}{4}c^{-1}q)^{-1/2}q<R_c^{\rm DF}$ is obtained from \cite[Corollary 2.12]{sere2023new} due to the fact that any DF minimizer $\gamma_*^c$ of \eqref{eq:min-DF} satisfies
\begin{align*}
    \|\gamma_*^c |\cD|\|_{\mathfrak{S}_1}< cR_c^{\rm DF}.
\end{align*}
\end{remark}

\medskip

\subsubsection{Hartree--Fock ground state energy}

Compared with the definition of the DF ground-state energy, the definition of the HF ground-state energy is much simpler:
\begin{align}\label{eq:min-HF}
    E_{q}^{\rm HF}:=\min_{\gamma\in \Gamma_q^{\rm HF}}\mathcal{E}^{\rm HF}(\gamma)
\end{align}
where the set of states of electrons in the HF theory is defined by $\Gamma_q^{\rm HF}:=\Gamma_q\cap \cB(\cH,\cHn)\cap X^2$. We refer to \cite{lieb1977hartree,bach1992error,bachunfill,ionization1,ionization2,lions} for its mathematical results. Concerning the existence of minimizers, we have the following.
\begin{theorem}[Existence of minimizers in the HF theory {\cite{lieb1977hartree,bach1992error,bachunfill}}]\label{th:min-HF} Let $q\in \R^+$ and $z\in \R^+$ be fixed such that $q< z+1$. Then the minimum problem \eqref{eq:min-HF} admits a minimizer $\gamma_*^{\rm HF} \in \Gamma_q^{\rm HF}$. In addition, $\Tr(\gamma_*^{\rm HF})=q$, and any such minimizer can be written as
\begin{align}\label{eq:gamma-HF}
    \gamma_*^{\rm HF}=\1_{(-\infty,\nu]}(H_{0,\gamma_*^{\rm HF}})
\end{align}
for some $\nu\in (-\infty,0)$.
\end{theorem}
\begin{remark}\label{rem:non-unfill}
Here we use the right-closed interval $(-\infty,\nu]$ to express the non-unfilled shell property in the HF theory \cite{bachunfill}: $\gamma_*^{\rm HF}$ can be rewritten as
\begin{align}\label{eq:gammaHF-uHF}
    \gamma_*^{\rm HF}:= \sum_{n= 1}^q |u_n^{\rm HF} \rangle\,\langle u_n^{\rm HF}|,
\end{align}
where $u_1^{\rm HF},\cdots,u_q^{\rm HF}$ are the orthonormal eigenfunctions of $H_{0,\gamma_*^{\rm HF}}$ in $\HL$ satisfying
\begin{align*}
   H_{0,\gamma_*^{\rm HF}} u_n^{\rm HF}=\lambda_{n}^{\rm HF} u_n^{\rm HF}
\end{align*}
and $\lambda_1^{\rm HF}\leq\cdots\leq\lambda_q^{\rm HF}<0$ are the first $q$ eigenvalues of $H_{0,\gamma_*^{\rm HF}}$ on $\HL$; for any other eigenfunction $u\in \HL$ of $H_{0,\gamma_*^{\rm HF}}$ with eigenvalue $\lambda$ and $\lc u,u_n^{\rm HF} \rc_\HL=0$ for $n=1,\cdots,q$, we have $\lambda>\lambda_q^{\rm HF}$.
\end{remark}

From \eqref{eq:gamma-HF}, it is easy to see that $\gamma_{*}^{\rm HF}\in X^4$. This will be used in the paper.

\subsection{Main result, Part I}\label{sec:main1}
Recall that $\kappa_c,\;\lambda_{0,c},\; a_c$ and $R^{\rm DF}_c$ are given in Assumption \ref{ass:1}. In this paper, we mainly focus on the non-relativistic regime, that is $c\gg1$ and $q,z\in \R^+$ fixed such that $q\leq z$. Under this regime, it is easy to see that for $c$ large enough,
\begin{align*}
    \kappa_c\leq \frac{1}{2}.
\end{align*}
Then under Assumption \ref{ass:1},
\begin{align*}
    c^{-1}q\leq \frac{1}{2}c^{-1}(q+z)=\frac{1}{4}\kappa_c\leq \frac{1}{8},  \qquad \qquad \lambda_{0,c}\geq 1-\kappa_c\geq \frac{1}{2}
\end{align*}
and
\begin{align}\label{eq:a-c==RcDF}
     a_c:=\frac{\pi}{4c\sqrt{(1-\kappa_c)\lambda_{0,c}}}\leq \frac{\pi}{2}c^{-1},\qquad R_c^{\rm DF}=1+(\frac{1}{2}-\frac{\pi}{4}c^{-1}q)q \leq 1+4q.
\end{align}

Thus, for future convenience, we restrict ourselves to the condition $\kappa_c\leq \frac{1}{2}$, and we further assume that
\begin{assumption}\label{ass:c}
Assume that $q\leq z$ and $\gamma_*^{\rm HF}$ is an HF minimizer of \eqref{eq:min-HF}. Let
\begin{align}\label{def:R_0}
    R_0:=\sup_{\gamma_*^{\rm HF}\in \mathcal{G}_{\rm HF}}\max\{2+4q, 1+\|\gamma_*^{\rm HF}\|_{X^2} + 4(\pi+2\sqrt{2}z)(1+\|\gamma_*^{\rm HF}\|_{X^2})^2\}.
\end{align} 
where  
\[\mathcal{G}_{\rm HF}:=\{\gamma_*^{\rm HF}\in \Gamma_{q}\cap \cB(\cH,\HL)\cap X^2;\; \gamma_*^{\rm HF} \mbox{is a HF minimizer of }E_q^{\rm HF}\}
\]
is the set of HF minimizers. We now assume that the speed of light $c\in \R$ satisfies
    \begin{align*}
        c\geq \max\{1,4q+4z,4\pi R_0\}.
    \end{align*}
\end{assumption}
Here $R_0$ is used in Section \ref{sec:6} to construct test density matrix $\gamma\in \Gamma_q^+$. Then,
\begin{theorem}[Relativistic effect in the DF theory]\label{th:rela-effect} Let $q,z\in \R^+$ be fixed such that $q\leq z$. For any $c$ satisfying Assumption \ref{ass:c}, we have 
    \begin{align}\label{eq:rela-effect}
      |E_{c,q}-E_{q}^{\rm HF}|=\cO(c^{-2}).
    \end{align}
\end{theorem}
The proof is postponed until Section \ref{sec:8.1}, and we will give the precise estimates of $\cO(c^{-2})$ later in the paper.

\medskip

As a consequence of Theorem \ref{th:rela-effect}, we can improve \cite[Theorem 2.9]{meng2024rigorous} to the case $q=z$:
\begin{theorem}[No unfilled-shell property]\label{th:non-unfill}
  For $c$ large enough and $Z\in \R^+,q\in \N^+$ satisfy $q\leq Z$, we know  $\delta_c=0$ with $\delta_c$ being given by \eqref{eq:gamma-DF}. 
  
  Then as $c\to \infty$, we have
  \begin{align*}
      \gamma_*^c\to \gamma_*^{\rm HF}
  \end{align*}
 in $X^2$ for some HF minimizers $\gamma_*^{\rm HF}$ of \eqref{eq:HF}. If we write $\gamma_*^c=\sum_{n=1}^q \left|u_{c,n}\right>\left<u_{c,n}\right|$ (for $c$ large enough) with $\{u_{c,n}\}_{1\leq n\leq q}$ an orthonormal eigenfunctions of $\cD_{\gamma_*^c}$ associated with eigenvalue $\lambda_n^c$, then as $c\to \infty$,
  \begin{align*}
   \lambda_n^c\to \lambda_n,\qquad   u_{c,n}\rightarrow u_{n}^{\rm HF}
  \end{align*}
  in $H^1$ and $\{u_{n}\}_{1\leq n\leq q}$ are orthonormal HF eigenfunctions of $H_{0,\gamma_*^{\rm HF}}$ associated with eigenvalue $\lambda_n$.
\end{theorem}
The proof is postponed until Section \ref{sec:8.2}.

 As a result of Theorem \ref{th:non-unfill}, we can also show that \cite[Corollary 2.10]{meng2024rigorous} hold for the case $q=z$. Then,
\begin{corollary}\label{cor:definition-equiv}
    For $c$ large enough and $q\leq z$, we have
 \begin{align*}
        E_{c,q}=e_{c,q}
    \end{align*}
where we recall that $e_{c,q}$ is the DF ground-state energy defined by the min-max theory \eqref{eq:DF-min-max}.
\end{corollary}
\begin{proof}
Theorem \ref{th:non-unfill} shows that $\delta_c=0$ when $q\leq z$ and $c$ large enough. Then \cite[Theorem 2.9]{meng2024rigorous} hold for the case $q=z$, and this corollary is obtained by repeating the proof of \cite[Corollary 2.10]{meng2024rigorous}.
\end{proof}
Thus \eqref{eq:rela-effect} also holds for the min-max definition of the DF ground-state energy:
\begin{corollary}\label{cor:relativistic-effect}
    For $c$ large enough and $q\leq z$, we have
 \begin{align*}
       |e_{c,q}-E_{q}^{\rm HF}|=\cO(c^{-2}).
    \end{align*}
\end{corollary}

\subsection{Main result, Part II}
From Theorem \ref{th:rela-effect}, we know the relativistic effect is of the order $\cO(c^{-2})$. In this part, we show that $\cO(c^{-2})$ is sharp, and we further give the explicit formula of the leading order relativistic correction term under the following additional assumption on $V$:
\begin{assumption}\label{ass:V}
    We assume further that the potential $V$ satisfies
    \begin{align*}
        \|(\nabla V)u\|_{\cH}\lesssim \|(1-\Delta)u\|_{\cH}.
    \end{align*}
Here and below, for two constants $a,b\in \R$, the notation $a \lesssim b$ means that there exists a constant $C$ independent of the speed of light $c$ such that $a\leq Cb$.
\end{assumption}

For a variational problem, this type of regular assumption for potential $V$ seems to be unavoidable. Concerning the relativistic hydrogen problem $\cD -\frac{1}{|x|}$, by Proposition \ref{prop:rela-decomp} below, the Darwin term is $\frac{1}{8c^2}\Delta |\cdot|^{-1}=\frac{\pi}{2c^2}\delta(x)$, and the error term $o(c^{-2})$ in Theorem \ref{th:rela-correction} below is also related to the derivatives of the Dirac distribution $\delta(x)$. However, Dirac distribution and its derivatives might be problematic for a DF minimizer in a variational framework.

\begin{remark}
In relativistic quantum chemistry, the nuclear charge density distribution can not be ignored \cite[Section 4]{andrae2002nuclear}, which means the nuclear charge distribution $\mu\in \mathcal{M}_+(\R^3)$ (we recall that $V$ is defined by $V=\mu*\frac{1}{|x|}$) is not a Dirac delta distribution. Assumption \ref{ass:V} is satisfied for commonly used nuclear charge distribution models such as  Gauss-type charge density distribution and Fermi-type charge distribution.
\end{remark}

Recall that $\mathcal{G}_{\rm HF}:=\{\gamma_*^{\rm HF}\in \Gamma_{q}\cap \cB(\cH,\HL)\cap X^2;\; \gamma_*^{\rm HF} \mbox{is a HF minimizer of }E_q^{\rm HF}\}$. Then using Corollary \ref{cor:definition-equiv},
\begin{theorem}\label{th:rela-correction}
Let $q\leq z$ with $q\in \N^+$. Under Assumption \ref{ass:V}, for $c$ large enough, we have
    \begin{align*}
       e_{c,q}= E_{c,q}=E_q^{\rm HF}+ E^{(2)}_c+o(c^{-2})
    \end{align*}
    where 
    \begin{align}\label{eq:E2-minimum}
       E^{(2)}_c:= \inf_{\gamma_*^{\rm HF}\in \mathcal{G}_{\rm HF}}\cE^{(2)}_c(\gamma_*^{\rm HF})
    \end{align}
and using \eqref{eq:gammaHF-uHF}, $\cE^{(2)}_c(\gamma_*^{\rm HF})$ is defined by
    \begin{align*}
      \cE^{(2)}_c(\gamma_*^{\rm HF}):&= -\frac1{4c^2}\Tr_{\cH}[H_{0,\gamma_*^{\rm HF}}\gamma_*^{\rm HF} \cL^2] \\
     &\quad+\frac{1}{4c^2}\left(\Tr_\cH[ (-V+W_{1,\gamma_*^{\rm HF}})\cL\gamma_*^{\rm HF}\cL]-\Tr_\cH[W_{2, \gamma_*^{\rm HF}\cL} \cL \gamma_*^{\rm HF}]\right)\\
      &= -\frac{1}{4c^2}\sum_{n=1}^q \lambda_{n}^{\rm HF}\left<\cL u_{n}^{\rm HF},\cL  u_n^{\rm HF}\right>_\HL\\
     &\quad+\frac{1}{4c^2}\sum_{n=1}^q\left<\cL u_n^{\rm HF}, (-V+W_{1,\gamma_*^{\rm HF}})\cL u_n^{\rm HF}\right>_{\HL}\\
     &\quad-\frac{1}{4c^2} \sum_{m,n=1}^q\int_{\R^3}  (u_n^{\rm HF})^*(x) u_m^{\rm HF}(x) \left<\cL u_m^{\rm HF}, |x-\cdot|^{-1}\cL u_n^{\rm HF}\right>_{\HL} dx.
    \end{align*}
\end{theorem}
The proof is postponed until Section \ref{sec:8.3}, and we will not give the precise lower bound of $c$ as for Theorem \ref{th:rela-effect} since it is more complicated.

\begin{remark}[Higher order relativistic correction]
It is possible to get higher order relativistic correction terms if the DF minimizer is unique for $c$ large enough. However, if DF minimizers are not unique for $c$ large enough, this might be problematic, since the continuity of the nonlinear term w.r.t. $c^{-1}$ is not clear. The error bound in Theorem \ref{th:rela-correction} is $o(c^{-2})$ rather than $\cO(c^{-4})$ is due to this continuity problem (see estimate \eqref{eq:6.2}).
\end{remark}

\medskip

Finally, we claim that $\cE^{(2)}_c(\gamma_*^{\rm HF})$ is the well-known leading order relativistic correction which can be decomposed into the following $3$ terms.
\begin{proposition}\label{prop:rela-decomp}
    We have 
    \begin{align*}
       4c^2 \cE^{(2)}_c(\gamma_*^{\rm HF}):= E_{\rm mv}+ E_{\rm D}+ E_{\rm so}
    \end{align*}
    where 
    \begin{itemize}
        \item $\frac{1}{4c^2} E_{\rm mv}$ is the mass-velocity term with
    \begin{align*}
        E_{\rm mv}:=-\frac{1}{2} \sum_{n=1}^q \left<u_{n}^{\rm HF}, (-\Delta)^2 u_{n}^{\rm HF}\right>_\HL;
    \end{align*}
    \item $\frac{1}{4c^2} E_{\rm D}$ is the Darwin term with
\begin{align*}
    E_{\rm D}:&=\frac{1}{2}\sum_{n=1}^q\left<u_n^{\rm HF}, \Big[\Delta (-V+W_{1,\gamma_*^{\rm HF}})\Big]  u_n^{\rm HF}\right>_\HL\\
    &\quad-\frac{1}{2}\sum_{m,n=1}^q \int_{\R^3} (u_n^{\rm HF})^*(x) u_m^{\rm HF}(x)\left< u_m^{\rm HF}, [\Delta_y W(x-\cdot)] u_n^{\rm HF}\right>_{\HL} dx;
\end{align*}
    \item $\frac{1}{4c^2} E_{\rm so}$ is the spin-orbital term with
\begin{align*}
    E_{\rm so}:&=\frac{1}{2}\sum_{n=1}^q\left<u_n^{\rm HF},\pmb\sigma \cdot [ (-\nabla V+\nabla W_{1,\gamma_*^{\rm HF}})\times (-i\nabla)]  u_n^{\rm HF}\right>_\HL\\
    &\quad-\frac{1}{2}\sum_{m,n=1}^q\int_{\R^3} (u_n^{\rm HF})^*(x) u_m^{\rm HF}(x)\left<  u_m^{\rm HF}, \pmb \sigma\cdot [(\nabla_y   W(x-\cdot))\times (-i\nabla)] u_n^{\rm HF}\right>_{\HL} dx.
\end{align*}
    \end{itemize}
\end{proposition}
The proof is postponed until Section \ref{sec:8.4}.

\subsection{Relativistic effects in Mittleman's definition of QED}
In this subsection, we will use Theorem \ref{th:rela-effect} and Theorem \ref{th:rela-correction} to show that our estimates on the relativistic effects also hold in a QED model defined by Mittleman. 

Before going further, we recall the definition the ground-state energy of QED due to Mittleman. According to Mittleman \cite{mittleman1981theory}, if the vacuum polarization is neglected, the physical ground state energy is defined by (see e.g., \cite{meng2024rigorous,barbaroux2005some,EstebanLewinSere-variational})
\begin{equation}\label{min-QED}
    E^{\rm M}_{c,q}:=\sup_{P_c^-\in \mathcal{P}_c}E^{(P_c^-)}_{c,q}.
\end{equation}
 Here 
\begin{align*}
    E^{(P_c^-)}_{c,q}=\min_{\gamma\in \Gamma^{(P_c^-)}_q}\mathcal{E}_c(\gamma)
\end{align*}
is a Hartree-Fock type ground-state energy under a fixed Dirac sea $P^-_c$ with the state of electron situated in
\begin{align}\label{eq:electron-position-field}
    \Gamma^{(P^-)}_q:=\{\gamma\in X;\;-P^-\leq \gamma\leq (1-P^-), \; (1-P^-) \gamma P^-=0,\; 0\leq \Tr (\gamma)\leq q\}.
\end{align}
Then when the vacuum polarization is neglected, the real physical ground state energy $ E^{\rm M}_{c,q}$ should be obtained by maximizing the ground state energy of the Hartree-Fock type ground-state over all allowed one-particle electron subspace $\mathcal{P}_c$ in electron-positron Hartree--Fock theory, with
\begin{align}\label{eq:dirac-sea}
    \mathcal{P}_c=\{P^-_{c,g};\; g\in X\}
\end{align}
where $P^-_{c,g}$ is defined by \eqref{eq:P-pm}.

Now we have
\begin{corollary}\label{cor:Mittleman}
      For $c$ large enough and $q\leq z$, we have
 \begin{align*}
        |E_{c,q}-E_{c,q}^{\rm M}|=\cO(c^{-4}).
    \end{align*}
\end{corollary}
\begin{proof}
   The case $q<z$ has been proved in \cite[Theorem 2.11]{meng2024rigorous}. Concerning the case $q=z$, this corollary follows from the same proof as for \cite[Theorem 2.11]{meng2024rigorous} and follows from replacing \cite[Corollary 2.10]{meng2024rigorous} by Corollary \ref{cor:definition-equiv}. This ends the proof.
\end{proof}
By Corollary \ref{cor:Mittleman}, Theorem \ref{th:rela-effect} and Theorem \ref{th:rela-correction}, we can now conclude that
\begin{corollary}\label{cor:mittleman-rela}
      For $c$ large enough and $q\leq z$, we have
 \begin{align*}
        |E_{c,q}^{\rm M}-E_q^{\rm HF}|=\cO(c^{-2})
    \end{align*}
    and under Assumption \ref{ass:V},
    \begin{align*}
        E_{c,q}^{\rm M}=E_q^{\rm HF}+E^{(2)}_c+o(c^{-2}).
    \end{align*}
\end{corollary}

\section{Sketch of proof}\label{sec:3.4}
In this section, we roughly explain the main ideas of the proof of Theorem \ref{th:rela-correction}. The proof of Theorem \ref{th:rela-effect} is simpler and follows a similar manner.

We first consider the estimate 
\begin{align}\label{eq:Ec<=Ehf}
    E_{c,q}\leq E_q^{\rm HF}+E^{(2)}_c+o(c^{-2})= \cE^{\rm HF}(\gamma_*^{\rm HF})+\cE^{(2)}_c(\gamma_*^{\rm HF})+o(c^{-2})
\end{align}
where $\gamma_*^{\rm HF}$ is any HF minimizer. To do so, we need to use $\gamma_*^{\rm HF}$ to construct some suitable test states $\gamma=\gamma(\gamma_*^{\rm HF})$ in DF theory, i.e, $\gamma\in \Gamma_q^+$. More precisely, our test state is $\gamma=\theta(\Lambda^+_c \widetilde{\gamma}_c^{\rm HF} \Lambda^+_c)$ with notations defined below. The proof will be split into the following steps:
\begin{itemize}
    \item Our first step is to renormalize the non-relativistic density matrix $\gamma_*^{\rm HF}$. An intuitive way is to use the HF minimizer $\gamma_*^{\rm HF}$ directly. However, it provides an error estimate of the order $\cO(c^{-2})$ which can only be used in the proof of Theorem \ref{th:rela-effect}. To get the leading order correction in Theorem \ref{th:rela-correction}, instead of considering one-body non-relativistic HF wavefunction $u^{\rm HF}\in L^2(\R^3;\C^2)$, when $c<\infty$, we need to use the following renormalized wavefunction
\begin{align}\label{eq:renormalization}
    \widetilde{u}^{\rm HF}:=\left(1+\frac{1}{4c^2}\left\|(-\Delta)^{1/2} u^{\rm HF}\right\|_{L^2(\R^3;\C^2)}^2\right)^{-1/2}\begin{pmatrix}
    u^{\rm HF}\\ \frac{1}{2c}\cL u^{\rm HF}
    \end{pmatrix} \in L^2(\R^3;\C^4),
\end{align}
Concerning the density matrix $\gamma_*^{\rm HF}$, the renormalized density matrix $\widetilde{\gamma}_c^{\rm HF}$ should be orthogonalized and satisfy $0\leq \widetilde{\gamma}_c^{\rm HF}\leq \1_{\cH}$. To do so, we introduce an overlap matrix $\mathcal{S}_{\rm HF}$ defined by \eqref{eq:S}. Using above renormalization \eqref{eq:renormalization} and the overlap matrix $\mathcal{S}_{\rm HF}$, we define the test density matrix $\widetilde{\gamma}_c^{\rm HF}$ in \eqref{eq:gammaHF'}. Using this overlap matrix, we overcome the problem in \cite[Section 6]{kutzelnigg2002perturbation} for non-degenerate cases.

\item The next step is to replace $\gamma_*^{\rm HF}$ by using the density matrix $\widetilde{\gamma}_c^{\rm HF}$ in the HF and DF functionals. Compared with quasi-relativistic problem, the relativistic correction of the kinetic term in Dirac problem is not obtained from the asymptotic expansion
\begin{align*}
    \sqrt{c^4-c^2\Delta}=c^2-\frac12\Delta+\cO(c^{-2}).
\end{align*}
Instead, using \eqref{eq:D-Delta} and the renormalization \eqref{eq:renormalization}, we can obtain the Laplace operator directly from the Dirac operator:
\begin{align*}
    \Tr_\cH [(\cD-c^2)\widetilde{\gamma}_c^{\rm HF}]=\Tr_\cH (H_0\cK_{\rm L}\widetilde{\gamma}_c^{\rm HF}\cK_{\rm L}).
\end{align*}
Furthermore, different from physical intuition, the leading order relativistic correction of the kinetic term (i.e., the term $\frac{1}{4c^2}E_{\rm mv}$ in Proposition \ref{prop:rela-decomp}) arises from the overlap matrix $\mathcal{S}_{\rm HF}$. Thus we should compare $\cE^{\rm HF}(\gamma_*^{\rm HF})$ with  $\cE_{c}(\widetilde{\gamma}_c^{\rm HF})$ rather than with $\cE^{\rm HF}(\widetilde{\gamma}_c^{\rm HF})$, and we show in Lemma \ref{lem:3.13} that
\begin{align*}
   \cE_{c}(\widetilde{\gamma}_c^{\rm HF})= \cE^{\rm HF}(\gamma_*^{\rm HF})+\cE^{(2)}_c(\gamma_*^{\rm HF})+\cO(c^{-4}).
\end{align*}

\item In the relativistic Dirac problem, it is useful to consider some positive projection of the Dirac operator. We now consider the projector $\Lambda^+_c$ and replace $\widetilde{\gamma}_c^{\rm HF}$ by $\Lambda^+_c \widetilde{\gamma}_c^{\rm HF} \Lambda^+_c$ in the DF functional:
\begin{align*}
    \cE_c(\widetilde{\gamma}_c^{\rm HF})=\cE_c(\Lambda^+_c \widetilde{\gamma}_c^{\rm HF} \Lambda^+_c)+\cO(c^{-4}).
\end{align*}
This proof is split into Sections \ref{sec:4.3.2}-\ref{sec:4.3.4}. In the proof, to deal with the potential terms, we need to split $\widetilde{\gamma}_c^{\rm HF}$ into four terms $\cK_{j}\widetilde{\gamma}_c^{\rm HF}\cK_{j'}$ with $j,j'\in\{\rm L,S\}$, and some delicate estimates for density matrix of the form $\cK_{j}\widetilde{\gamma}_c^{\rm HF}\cK_{j'}$ are given in Section \ref{sec:sometechnicaltools}. Gathering above three steps, we finally prove Theorem \ref{th:3.1}, i.e.,
\begin{align*}
     \cE_c(\Lambda^+_c \widetilde{\gamma}_c^{\rm HF} \Lambda^+_c)=\cE^{\rm HF}(\gamma_*^{\rm HF})+\cE^{(2)}_c(\gamma_*^{\rm HF})+\cO(c^{-4}).
\end{align*}

\item In the DF theory, the state of electrons should satisfy $\gamma=P^+_{c,\gamma}\gamma P^+_{c,\gamma}\in \Gamma_q^+$, see \eqref{eq:min-DF}. However, our state $\Lambda^+_c \widetilde{\gamma}_c^{\rm HF} \Lambda^+_c$ does not fulfill this requirement. In this step, using the retraction mapping $\gamma\mapsto \theta(\gamma)$ introduced in \cite{sere2023new}, we replace $\Lambda^+_c \widetilde{\gamma}_c^{\rm HF} \Lambda^+_c$ by $\theta(\Lambda^+_c \widetilde{\gamma}_c^{\rm HF} \Lambda^+_c)\in \Gamma_q^+$ and we show in Section \ref{sec:6} that
\begin{align*}
    \cE_c(\theta(\Lambda^+_c \widetilde{\gamma}_c^{\rm HF} \Lambda^+_c))=\cE_c(\Lambda^+_c \widetilde{\gamma}_c^{\rm HF} \Lambda^+_c)+\cO(c^{-4}).
\end{align*}
By the definition of the DF ground-state energy, we know
\begin{align*}
    E_{c,q}\leq  \cE_c(\theta(\Lambda^+_c \widetilde{\gamma}_c^{\rm HF} \Lambda^+_c))=\cE_c(\Lambda^+_c \widetilde{\gamma}_c^{\rm HF} \Lambda^+_c)+\cO(c^{-4}).
\end{align*}
The proof relies heavily on the author's previous work \cite{meng2024rigorous}, and some estimates on the projections are given in Section \ref{sec:5.2}.
\item Finally, using above estimates and taking the infimum over all HF minimizers, we finally obtain
\[E_{c,q}\leq  \inf_{\gamma^{\rm HF}_*\in \mathcal{G}_{\rm HF}}\Big(\cE^{\rm HF}(\gamma_*^{\rm HF})+\cE^{(2)}_c(\gamma_*^{\rm HF})\Big)+\cO(c^{-4})=E_q^{\rm HF}+E^{(2)}_c+\cO(c^{-4}).
\]
\end{itemize}

\medskip

Now we consider the inverse estimate
\begin{align}\label{eq:Ec>=Ehf}
    \cE_c(\gamma_*^c)=E_{c,q}\geq  E_q^{\rm HF}+E^{(2)}_c+o(c^{-2})
\end{align}
where $\gamma_*^c$ is a DF minimizer. The proof will be split into the following steps:
\begin{itemize}
    \item Analogous to the first step for proof of \eqref{eq:Ec<=Ehf}, we also need to renormalize the relativistic density matrix $\gamma_*^c$. Obviously, we can regard the density matrix $\cK_{\rm L}\gamma_*^c\cK_{\rm L}$ as a non-relativistic electronic state in HF theory satisfying $0\leq \cK_{\rm L}\gamma_*^c\cK_{\rm L} \leq \1_{\cH_{\rm L}}$. This works for Theorem \ref{th:rela-effect} but fails for Theorem \ref{th:rela-correction}. Therefore a renormalization is needed. We first introduce another overlap matrix $\mathcal{S}_{\rm DF}$ defined by \eqref{eq:S-DF}. Using $\mathcal{S}_{\rm DF}$, the needed renormalized density matrix $\widetilde{\gamma}_*^c$ is given in \eqref{eq:gammaDF'}.
    \item Our next step is to replace $\gamma_*^c$ by the renormalized density matrix $\widetilde{\gamma}_*^c$ in the HF and DF functionals. To do so, we first give some useful estimates for the eigenfunctions of DF operator $\cD_{\gamma}$ in Section \ref{sec:structure-gamma-c}. Then we show
    \begin{align*}
       E_{c,q}=\cE(\gamma_*^c)\geq   \cE^{\rm HF}(\widetilde{\gamma}^c_*)+\widetilde{\cE}_c^{(2)}(\gamma_*^c)+\cO(c^{-4})
    \end{align*}
where $\widetilde{\cE}_c^{(2)}(\gamma_*^c)$, defined in Theorem \ref{th:5.1}, is almost the relativistic correction term. This estimate is concluded in Theorem \ref{th:5.1}. In the proof, some delicate estimates are studied in Section \ref{sec:6.3.1} to control the kinetic term, and we also need to split $\widetilde{\gamma}_c^{\rm HF}$ into four terms $\cK_{j}\widetilde{\gamma}_c^{\rm HF}\cK_{j'}$ with $j,j'\in\{\rm L,S\}$ when dealing with potential terms.
\item Finally, using the definition of HF ground-state energy, we have
\begin{align*}
     E_{c,q}\geq   E_q^{\rm HF}+\widetilde{\cE}_c^{(2)}(\gamma_*^c)+\cO(c^{-4}).
\end{align*}
Then using the fact that $\gamma_*^c \to \gamma_*^{\rm HF}$ in $X^2$ for a HF minimizer $\gamma_*^{\rm HF}\in \mathcal{G}_{\rm HF}$ (see Theorem \ref{th:non-unfill}), we further show that as $c\to \infty$,
\begin{align*}
    4c^2\widetilde{\cE}_c^{(2)}(\gamma_*^c)\to 4c^2\cE^{(2)}_c(\gamma_*^{\rm HF}),
\end{align*}
namely
\begin{align*}
    \widetilde{\cE}_c^{(2)}(\gamma_*^c)=\cE^{(2)}_c(\gamma_*^{\rm HF})+o(c^{-2})\geq E^{(2)}_c+o(c^{-2}).
\end{align*}
This proves \eqref{eq:Ec>=Ehf}, i.e.,
\begin{align*}
      \cE_c(\gamma_*^c)=E_{c,q}\geq  E_q^{\rm HF}+E^{(2)}_c+o(c^{-2}).
\end{align*}
This step can be found in Section \ref{sec:8.3}.
\end{itemize}

\section{From HF problem to ``projected'' DF problem}\label{sec:5}
In this section, we are trying to understand the relationship between the HF ground-state energy and some DF energies associated with free picture (i.e., the state of electrons satisfy $\gamma=\Lambda^+_c \gamma \Lambda^+_c$). As explained in Section \ref{sec:3.4}, this is the first step to prove Theorem \ref{th:rela-effect} and Theorem \ref{th:rela-correction}. We will show the following.
\begin{theorem}[From HF problem to ``projected'' DF problem]\label{th:3.1}
Let $q\leq z$ and let $\gamma_*^{\rm HF}$ be any HF minimizer of $E_q^{\rm HF}$. Then under Assumption \ref{ass:c}, we have
    \begin{align}\label{eq:E_c-EHF}
        \cE_c(\Lambda^+_c \gamma_*^{\rm HF}\Lambda^+_c)\leq E_q^{\rm HF}+\cO(c^{-2}).
    \end{align}
In addition, under Assumption \ref{ass:V}, for $c$ large enough, there exists a density matrix $\widetilde{\gamma}_c^{\rm HF}=\cS\widetilde{\gamma}_c^{\rm HF}\cS^*\in \Gamma_q$ such that 
\begin{align}\label{eq:E_c-EHF-c2}
     \cE_c(\Lambda^+_c \widetilde{\gamma}_c^{\rm HF}\Lambda^+_c)\leq E_q^{\rm HF}+\cE^{(2)}_c(\gamma_*^{\rm HF})+\cO(c^{-4}).
\end{align}
More precisely, this density matrix $\widetilde{\gamma}_c^{\rm HF}$ is defined by \eqref{eq:gammaHF'} below. It is a relativistic renormalization of the HF minimizer $\gamma_*^{\rm HF}$.
\end{theorem}

This result means that we can pass from the HF ground-state energy to some DF energies associated with the free picture $\Lambda^+_c$ with an error term of the order $\cO(c^{-2})$ (or $\cO(c^{-4})$ under Assumption \ref{ass:V}).

\medskip

The proof of Theorem \ref{th:3.1} is organized as follows:
\begin{itemize}
    \item We first introduce some technical estimates in Section \ref{sec:sometechnicaltools}. We consider some estimates on the projector $\Lambda^+$ and $\Lambda^-:=1-\Lambda^+$. They are mainly used to study the terms:
\begin{align*}
   \gamma_*^{\rm HF}- \Lambda^+_c \gamma_*^{\rm HF}\Lambda^+_c= \Lambda^+_c \gamma_*^{\rm HF}\Lambda^-_c+ \Lambda^-_c \gamma_*^{\rm HF}\Lambda^+_c+\Lambda^-_c \gamma_*^{\rm HF}\Lambda^-_c
\end{align*}
and
\begin{align*}
 \widetilde{\gamma}_c^{\rm HF}-\Lambda^+_c \widetilde{\gamma}_c^{\rm HF}\Lambda^+_c=\Lambda^+_c \widetilde{\gamma}_c^{\rm HF}\Lambda^-_c+\Lambda^-_c \widetilde{\gamma}_c^{\rm HF}\Lambda^+_c+\Lambda^-_c \widetilde{\gamma}_c^{\rm HF}\Lambda^-_c. 
\end{align*}
In these estimates, the property $\widetilde{\gamma}_c^{\rm HF}=\cS\widetilde{\gamma}_c^{\rm HF}\cS^*$ is used to get a better estimates w.r.t. $c^{-1}$. In addition, we also summarize some Hardy-type inequalities used to control the potential terms in the functional.
\item We next prove \eqref{eq:E_c-EHF} and \eqref{eq:E_c-EHF-c2} separately in Section \ref{sec:4.2} and Section \ref{sec:4.3}. In their proof, we split the functional into the following terms: for any density matrix $\gamma$,
\begin{align}\label{EDF-EHF}
  \MoveEqLeft   \cE_c(\Lambda^+\gamma\Lambda^+)- \cE^{\rm HF}(\gamma)= \underbrace{\Tr_{\cH}[(\cD-c^2) \Lambda^+_c\gamma \Lambda^+_c] -\Tr_{\cH}[H_0 \gamma ]}_{\textrm{kinetic term}}\notag\\
     &-\left(\underbrace{\Tr_{\cH}[V\Lambda^+_c\gamma \Lambda^+_c]-\Tr[V\gamma]}_{\textrm{potential between electrons and nuclei}}\right)+\left(\underbrace{\Tr_{\cH}[W_{\Lambda^+_c\gamma \Lambda^+_c}\Lambda^+_c\gamma^{\rm HF}_{*} \Lambda^+_c]-\Tr_{\cH}[W_{\gamma}\gamma]}_{\textrm{potential between electrons and electrons}}\right).
\end{align}
Then we study each terms with $\gamma= \gamma_*^{\rm HF}$ and $\gamma=\widetilde{\gamma}_c^{\rm HF}$ respectively. 
\item  To prove \eqref{eq:E_c-EHF-c2}, we also need to introduce the density matrix $\widetilde{\gamma}_c^{\rm HF}$. It is constructed in Section \ref{sec:4.3.1} and its property is studied in Section \ref{sec:4.3.2'}. In addition, to use \eqref{EDF-EHF}, we show in Lemma \ref{lem:3.13} that 
\begin{align*}
     \cE_{c}(\widetilde{\gamma}_c^{\rm HF})= E_q^{\rm HF}+\cE^{(2)}_c(\gamma_*^{\rm HF})+\cO(c^{-4}).
\end{align*}
Then \eqref{eq:E_c-EHF-c2} can be obtained by studying \eqref{EDF-EHF}.
\end{itemize}

\subsection{Some technical tools}\label{sec:sometechnicaltools}
We now introduce some technical results used for Theorem \ref{th:3.1}. First, the H\"older inequality for Schatten norm (see e.g., \cite[Theorem 2.8]{simon2005trace}) will also be used frequently:
\begin{align}\label{eq:Holder}
    \|AB\|_{\mathfrak{S}_1}\leq \|A\|_{\mathfrak{S}_2}\|B\|_{\mathfrak{S}_2}=\|A A^*\|_{\mathfrak{S}_1}^{1/2}\| B^*B\|_{\mathfrak{S}_1}^{1/2}.
\end{align}

\medskip

\noindent{\bf Estimates on the projectors.} Next, we consider some estimates on the projector $\Lambda^\pm$. The following lemma and corollary are used for \eqref{eq:E_c-EHF}.
\begin{lemma}\label{lem:lambda+-non}
Let $u\in \cHn\cap H^2$, then for any $0\leq s\leq 1$,
\begin{align}\label{eq:lambda-non}
    \|\cK_{\rm L}\Lambda^-_c u\|_{\cH}\leq \frac{1}{4c^2}\|u\|_{H^2},\;\|\cK_{\rm L}\Lambda^+_c u\|_{H^s}\leq \|u\|_{H^s}, \; \|\cK_{\rm S}\Lambda^\pm_c u\|_{H^s}\leq \frac{1}{2c}\|u\|_{H^{s+1}}.
\end{align}
\end{lemma}

\begin{proof}
Observe that for $u\in \cHn$
\[
\cD u=\begin{pmatrix}
    c^2 & c \cL\\ c\cL&-c^2
\end{pmatrix}\begin{pmatrix}
        u^{\rm L}\\ 0
    \end{pmatrix}=\begin{pmatrix}c^2 u^{\rm L}\\ c\cL u^{\rm L}\end{pmatrix}.
\]
where we recall that $\cL$ is defined in \eqref{eq:cL}. Thus
\[
\Lambda_c^\pm u=\frac{1}{2}\left(1\pm\frac{\cD}{|\cD|}\right) u=\frac{1}{2|\cD|}\begin{pmatrix}[|\cD|\pm c^2]u^{\rm L}\\\pm c\cL  u^{\rm L}\end{pmatrix}.
\]
Then, as $c^2\leq |\cD|\leq c^2-\frac{1}{2}\Delta$, we have
\begin{align*}
    \MoveEqLeft\|\cK_{\rm L}\Lambda_c^- u\|_{\mathcal{H}}= \frac{1}{2}\left\|\frac{|\cD|-c^2}{|\cD|}u^{\rm L}\right\|_{L^2(\R^3;\C^2)}\leq \frac{1}{2c^2}\|(|\cD|-c^2)u^{\rm L}\|_{L^2(\R^3;\C^2)}\leq \frac{1}{4c^2}\|u\|_{H^2}
\end{align*}
and for any $s\geq 0$,
\begin{align*}
    \MoveEqLeft\|\cK_{\rm L}\Lambda_c^+ u\|_{H^s}= \frac{1}{2}\left\|\frac{|\cD|+c^2}{|\cD|} u^{\rm L}\right\|_{H^s(\R^3;\C^2)}\leq \|u^{\rm L}\|_{H^s(\R^3;\C^2)}\leq \|u\|_{H^s}.
\end{align*}
Analogously, it is easy to see that
\begin{align*}
    \MoveEqLeft\|\cK_{\rm S}\Lambda_c^\pm u\|_{H^{s}}\leq \frac{1}{2c}\|\cL u^{\rm L}\|_{H^s(\R^3;\C^2)}\leq \frac{1}{2c}\|u\|_{H^{s+1}}.
\end{align*}
This ends the proof.
\end{proof}
Then,
\begin{corollary}\label{cor:lambda-non-gamma}
    For any non-negative self-adjoint density matrix $\gamma\in \cB(\cH,\cHn)\cap X^4$ and for any $0\leq s\leq 2$, we have
\begin{align*}
    \|\cK_{\rm L}\Lambda_c^-\gamma \Lambda_c^-\cK_{\rm L}\|_{\mathfrak{S}_1}\leq \frac{1}{16c^4}\| \gamma\|_{X^4},\qquad
    \| \cK_{\rm L}\Lambda_c^+\gamma\Lambda_c^+\cK_{\rm L} \|_{X^s}\leq \|\gamma\|_{X^s}
\end{align*}
and
\begin{align*}
\|\cK_{\rm S}\Lambda_c^{\pm}\gamma \Lambda_c^{\pm}\cK_{\rm S}\|_{X^s}\leq \frac{1}{4c^2}\|\gamma\|_{X^{2+s}}.
\end{align*}
\end{corollary}
\begin{proof}
    As $\gamma\in \cB(\cH,\cHn)\cap X^4$ is non-negative, according to \eqref{eq:gamma-dec}, it can be written as
\begin{align}\label{eq:gammaHF-dec}
    \gamma=\sum_{n\geq 1}\lambda_n |u_n \rangle\,\langle u_n|
\end{align}
with $0\leq \lambda_n\leq 1$ and $\{u_n\}_{n\geq 1}$ being an orthonormal basis in $\cHn\cap H^2$. 

Then 
\begin{align*}
    \cK_{\rm L}\Lambda_c^-\gamma \Lambda_c^-\cK_{\rm L}= \sum_{n\geq 1}\lambda_n  |\cK_{\rm L}\Lambda^-  u_n \rangle\,\langle \cK_{\rm L}\Lambda^-   u_n|.
\end{align*}
As a result, from Lemma \ref{lem:lambda+-non}, we infer
\begin{align*}
    \|\cK_{\rm L}\Lambda_c^-\gamma \Lambda_c^-\cK_{\rm L}\|_{\mathfrak{S}_1}&=  \sum_{n\geq 1}\lambda_n   \|\cK_{\rm L}\Lambda_c^-  u_n\|_{\cH}^2\leq \frac{1}{16c^4}\sum_{n\geq 1}\lambda_n \|u_n\|_{H^2}^2=\frac{1}{16c^4}\|\gamma\|_{X^4}.
\end{align*}
Here the first equation and the last equation hold since $\cK_{\rm L}\Lambda_c^-\gamma \Lambda_c^-\cK_{\rm L}$ and $(1-\Delta)\gamma(1-\Delta)$ are non-negative density matrices.

Analogously, we can deduce other estimates. This ends the proof.
\end{proof}

Concerning the proof of \eqref{eq:E_c-EHF-c2}, we need the followings.
\begin{lemma}\label{lem:lambda+-non'}
    Let $u\in H^3\cap \HL$, then for any $0\leq s\leq 2$,
    \begin{align}
        \|(-\Delta)^{-1/2}\cK_\rL \Lambda^-_c \cS u\|_{\cH}&\lesssim c^{-4}\|u\|_{H^3},\quad &\|\cK_\rL \Lambda^+_c  \cS u\|_{H^s}\lesssim \| u\|_{H^{s+1}},\quad\,\;\label{eq:Lambda-KL}\\
         \|\cK_\rS \Lambda^-_c  \cS u\|_{\cH}&\lesssim c^{-3}\|u\|_{H^{3}},\quad &\|\cK_\rS \Lambda^+_c \cS  u\|_{H^s}\lesssim c^{-1}\| u\|_{H^{s+1}}.\label{eq:Lambda-KS}
    \end{align}
\end{lemma}
Under Assumption \ref{ass:V}, the non-relativistic wavefunction $u^{\rm HF}_n\in H^3$. Thus to reach an estimate of the order $\cO(c^{-4})$, we need to add the operator $(-\Delta)^{-1/2}$ to control the term $\cK_\rL \Lambda^-_c \cS u$ in \eqref{eq:Lambda-KL}.
\begin{proof}[Proof of Lemma \ref{lem:lambda+-non'}]
We have
\begin{align*}
    \Lambda^\pm \cS u &=\frac{1}{2|\cD|}\begin{pmatrix}\left(|\cD|\pm (c^2-\frac{1}{2}\Delta)\right)u^{\rm L}\\\frac{1}{2c}(|\cD|\pm c^2) \cL  u^{\rm L}\end{pmatrix}.
\end{align*}
According to Taylor's expansion, for any $\xi\in \R^3$,
\begin{align*}
    \sqrt{c^4+c^2|\xi|^2}=c^2+\frac{1}{2}|\xi|^2-\frac{1}{8c^2}|\xi|^4 -\frac{c^2}{16}\int_0^{c^{-2}|\xi|^2} (1+t)^{-3/2}(c^{-2}|\xi|^2-t)^2 dt.
\end{align*}
which implies 
\begin{align*}
   \MoveEqLeft \left| \sqrt{c^4+c^2|\xi|^2} -\Big(c^2+\frac{1}{2}|\xi|^2\Big)\right|\\
    &\leq \frac{1}{8c^2}|\xi|^4+ \frac{c^2}{16}\int_{0}^{c^{-2}|\xi|^2} (1+t)^{-3/2}(c^{-2}|\xi|^2-t)^2dt\\
    &\lesssim c^{-2}|\xi|^4 +c^{-2}|\xi|^4 \int_{0}^{c^{-2}|\xi|^2} (1+t)^{-3/2} dt\lesssim c^{-2}|\xi|^4.
\end{align*}
Analogously, we also have
\begin{align*}
    \left| \sqrt{c^4+c^2|\xi|^2} -c^2\right|\leq \frac{1}{2}|\xi|^2.
\end{align*}
Thus for any $u\in \cH$, as $\||\cD|^{-1}\|_{\cB(\cH)}=c^{-2}$, we have
\begin{align*}
    \left\|\frac{1}{2|\cD|}\left(|\cD|-(c^2-\frac{1}{2}\Delta)\right) (-\Delta)^{-1/2} u^\rL\right\|_{L^2(\R^3;\C^2)}&\lesssim c^{-4}\|(-\Delta)^{3/2} u^\rL\|_{L^2(\R^3;\C^2)},\\
    \frac{1}{2c}\left\|\frac{1}{2|\cD|}\left(|\cD|-c^2\right) \cL u^\rL\right\|_{L^2(\R^3;\C^2)}&\lesssim c^{-3}\|(-\Delta)^{3/2} u^\rL\|_{L^2(\R^3;\C^2)}.
\end{align*}
This gives the first inequalities in \eqref{eq:Lambda-KL} and \eqref{eq:Lambda-KS}. Analogously, we have
\begin{align*}
     \left\|\frac{1}{2|\cD|}\left(|\cD|+(c^2-\frac{1}{2}\Delta)\right) u^\rL\right\|_{H^s(\R^3;\C^2)}&\lesssim \| u^\rL\|_{H^{s+1}(\R^3;\C^2)} ,\\
     \frac{1}{2c}\left\|\frac{1}{2|\cD|}\left(|\cD|+c^2\right) \cL u^\rL\right\|_{H^{s}(\R^3;\C^2)}&\lesssim  c^{-1}\| u^\rL\|_{H^{s+1}(\R^3;\C^2)}.
\end{align*}
This gives the second inequalities in \eqref{eq:Lambda-KL} and \eqref{eq:Lambda-KS}. This ends the proof.
\end{proof}
Analogous to the proof of Corollary \ref{cor:lambda-non-gamma}, we also have the following. 
\begin{corollary}\label{cor:lambda-non-gamma'}
    For any non-negative self-adjoint density matrix $\gamma=\cS\gamma \cS  \in X$, such that $\cK_\rL \gamma \cK_\rL \in X^6$, and for any $0\leq s\leq 4$, we have
\begin{align*}
    \|(-\Delta)^{-1/2}\cK_{\rm L}\Lambda_c^-\gamma \Lambda_c^-\cK_{\rm L}(-\Delta)^{-1/2}\|_{\mathfrak{S}_1}&\lesssim c^{-8}\| \cK_\rL\gamma \cK_\rL\|_{X^6},\\
\|\cK_{\rm S}\Lambda_c^{-}\gamma \Lambda_c^-\cK_{\rm S}\|_{\mathfrak{S}_1}&\lesssim c^{-6}\|\cK_\rL\gamma \cK_\rL\|_{X^{6}}\\
    \| \cK_{\rm L}\Lambda_c^+\gamma\Lambda_c^+\cK_{\rm L} \|_{X^s}&\lesssim \|\cK_\rL \gamma \cK_\rL\|_{X^{s+2}}\\
\|\cK_{\rm S}\Lambda_c^{+}\gamma \Lambda_c^+\cK_{\rm S}\|_{X^s}&\lesssim c^{-2}\|\cK_\rL\gamma\cK_\rL\|_{X^{s+2}}.
\end{align*}
\end{corollary}

\medskip
\noindent{\bf Estimates on the potentials.} Finally, we need some Hardy-type inequalities for the potential terms. We now recall some Hardy-type inequality:
\begin{align}\label{eq:Hardy}
\||\cdot|^{-1}(-\Delta)^{-1/2}\|_{\cB(\cH)}=\|(-\Delta)^{-1/2}|\cdot|^{-1}\|_{\cB(\cH)}=2;
\end{align}
which also implies
\begin{align}\label{eq:W1-estimate}
    |\nabla W_{1, \gamma}(x)|\leq \int_{\R^3}\frac{|\rho_{\gamma}(y)|}{|x-y|^2}dy\lesssim \|\gamma\|_{X^2}.
\end{align}
Concerning $W_{2,\bullet}$, we have
\begin{lemma}\label{lem:W2-estimate}
    Let $\gamma\in \mathfrak{S}_1$ and $\|\gamma\|_{X^1}<\infty$. Then for any $f\in H^1$ ,
    \begin{align}
        \|W_{2,\gamma}f\|_{\cH}&\leq 2\|\gamma\|_{\mathfrak{S}_1}\|\nabla f\|_{\cH},\label{eq:5.17}\\
          \|W_{2,\gamma}f\|_{\cH}&\leq 2\|\gamma\|_{X}\|f\|_{\cH},\label{eq:5.1-1'}\\
         \|(-\Delta)^{1/2}W_{2,\gamma }f\|_{\cH}&\leq 6 \|\gamma \|_{X^2}\|\nabla f\|_{\cH}.\label{eq:5.19'}
    \end{align}
Using \eqref{eq:Holder}, we have in particular, for $j,j'\in\{\rL,\rS\}$
\begin{align}
     \|W_{2,\cK_j\gamma \cK_{j'}}f\|_{\cH}&\leq  2\|\cK_j \gamma \cK_j\|_{\mathfrak{S}_1}^{1/2}\|\cK_{j'} \gamma \cK_{j'}\|_{\mathfrak{S}_1}^{1/2}\|\nabla f\|_{\cH},\label{eq:5.18}\\
          \|(-\Delta)^{1/2}W_{2,\cK_j\gamma \cK_{j'}}f\|_{\cH}&\leq 6 \|\cK_j \gamma \cK_j\|_{X^2}^{1/2}\|\cK_{j'} \gamma \cK_{j'}\|_{X^2}^{1/2}\|\nabla f\|_{\cH}.\label{eq:5.19}
\end{align}
\end{lemma}
\begin{proof}
Note that
\begin{align*}
    W_{2,\gamma}f=\int_{\R^3}\frac{\gamma(x,y) f(y)}{|x-y|}dy.
\end{align*}
Thus,
\begin{align*}
     \|W_{2,\gamma} f\|_{\cH}&=\left\| \int_{\R^3}\frac{\gamma(x,y)f(y)}{|x-y|}dy\right\|_{\cH}\\
    &\leq \left(\int_{\R^3}\int_{\R^3}|\gamma(x,y)|^2dxdy\right)^{1/2}\sup_{x\in \R^3}\||x-\cdot|^{-1}f\|_{\cH}\\
    &\leq 2\|\gamma\|_{\mathfrak{S}_2}\|\nabla f\|_{\cH}\leq 2\|\gamma\|_{\mathfrak{S}_1}\|\nabla f\|_{\cH}.
\end{align*}
This proves \eqref{eq:5.17}. 

Then by \eqref{eq:Holder} and  \eqref{eq:5.17},
\begin{align*}
     \|W_{2,\cK_j\gamma \cK_{j'}} f\|_{\cH}&=\left\| \int_{\R^3}\frac{\big(\cK_j\gamma \cK_{j'}\big)(x,y)f(y)}{|x-y|}dy\right\|_{\cH}\\
     &\leq 2\|\cK_j\gamma \cK_{j'}\|_{\mathfrak{S}_1}\|\nabla f\|_{\cH}\leq 2\|\cK_j \gamma \cK_j\|_{\mathfrak{S}_1}^{1/2}\|\cK_{j'} \gamma \cK_{j'}\|_{\mathfrak{S}_1}^{1/2}\|\nabla f\|_{\cH}.
\end{align*}
We get \eqref{eq:5.18}. 

Concerning \eqref{eq:5.1-1'}, we have
\begin{align*}
  \|W_{2,\gamma} f\|_{\cH}&\leq \left(\int_{\R^3}\int_{\R^3}\frac{|\gamma(x,y)|^2}{|x-y|^2}dxdy\right)^{1/2}\|f\|_{\cH}\\
    &\leq 2\left(\int_{\R^3}\int_{\R^3}\left|\Big((-\Delta)^{1/4}\gamma (-\Delta)^{1/4}\Big)(x,y)\right|^2dxdy\right)^{1/2}\|f\|_{\cH}\\
     &=2\|(-\Delta)^{1/4}\gamma(-\Delta)^{1/4} \|_{\mathfrak{S}_2}\|\nabla f\|_{\cH}\\
    &\leq 2\|\gamma\|_{X}\|\nabla f\|_{\cH}\leq 2\|\gamma\|_{X}\|f\|_{\cH}.
\end{align*}
where $\Big((-\Delta)^{1/4}\gamma (-\Delta)^{1/4}\Big)(x,y)=(-\Delta_x)^{1/4}(-\Delta_y)^{1/4}\gamma(x,y)$ is the kernel of the operator $(-\Delta)^{1/4}\gamma (-\Delta)^{1/4}$, for any $h\in \mathfrak{S}_1$,
\begin{align*}
    \|h\|_{\mathfrak{S}_2}\leq \|h\|_{\mathfrak{S}_1}
\end{align*}
and we used the following inequality (see also \cite{meng2023mixed}): by Hardy's inequality,
\begin{align*}
 \MoveEqLeft   \|(-\Delta_x)^{-1/4}(-\Delta_y)^{-1/4}|x-y|^{-1}\|_{\cB(L^2(\R^3\times \R^3))}\\
 &\leq  \frac{1}{2}\|(-\Delta_x)^{-1/2}|x-y|^{-1}\|_{\cB(L^2(\R^3\times \R^3))}+\frac{1}{2}\|(-\Delta_y)^{-1/2}|x-y|^{-1}\|_{\cB(L^2(\R^3\times \R^3))}\leq 2.
\end{align*}

Concerning \eqref{eq:5.19'}, we have
\begin{align}
     \MoveEqLeft   \|(-\Delta)^{1/2} W_{2,\gamma} f\|_{\cH}=\|\nabla W_{2,\gamma}f\|_{\cH}\notag\\
    &\leq \left\|\int_{\R^3}\frac{|(\nabla \gamma)(x,y)||f|(y)}{|x-y|}dy\right\|_{\cH}+ \left\| \int_{\R^3}\frac{|\gamma(x,y)||f|(y)}{|x-y|^2}dy\right\|_{\cH}\notag\\
    &\leq \left(\int_{\R^3}\int_{\R^3} |(\nabla \gamma)(x,y)|^2 dxdy\right)^{1/2}\sup_{x\in \R^3}\||x-\cdot|^{-1}f\|_{\cH} \notag\\
    &\quad+ \left(\int_{\R^3}\int_{\R^3}\frac{ | \gamma(x,y)|^2}{|x-y|^2} dydx\right)^{1/2} \sup_{x\in \R^3}\left(\int_{\R^3}\frac{|f(y)|^2}{|x-y|^2}dy\right)^{1/2}  \notag\\
    &\leq 6\|\nabla \gamma\|_{\mathfrak{S}_2}\|\nabla f\|_{\cH} \leq 6\|\gamma\|_{X^2} \|\nabla f\|_{\cH}
\end{align}
where $(\nabla \gamma)(x,y)= \nabla_x \gamma(x,y)$ is the kernel of the operator $\nabla \gamma$. 

Finally, using H\"older's inequality \eqref{eq:Holder} and \eqref{eq:5.19'}, we get \eqref{eq:5.19}. This ends the proof.
\end{proof}
Then Hardy's inequality, \eqref{eq:W1-estimate} and Lemma \ref{lem:W2-estimate} show that
\begin{corollary}\label{cor:regularity-H3}
Let $u_1^{\rm HF},\cdots, u_n^{\rm HF}$ be the eigenfunctions of $\gamma_*^{\rm HF}$ given as in Remark \ref{rem:non-unfill}. Then
\begin{align*}
    \max_{1\leq n\leq q}\|u_n^{\rm HF}\|_{H^2}=\cO(1),\qquad \gamma_*^{\rm HF}\in X^4.
\end{align*}
Under Assumption \ref{ass:V}, we have in addition, 
\begin{align*}
    \max_{1\leq n\leq q}\|u_n^{\rm HF}\|_{H^3}=\cO(1),\qquad \gamma_*^{\rm HF}\in X^6.
\end{align*}
\end{corollary}
\begin{proof}
    According to Remark \ref{rem:non-unfill}, 
    \begin{align*}
        H_0 u_n^{\rm HF}=(\lambda_n^{\rm HF}- V-W_{\gamma_{*}^{\rm HF}})u_n^{\rm HF}
    \end{align*}
    with 
    \begin{align*}
    \lambda_1\leq  \lambda_2\leq \cdots\leq \lambda_q^{\rm HF}<0.
\end{align*}
   Then, for $s=0$ or $s=1$ and for $1\leq n\leq q$, under Assumption \ref{ass:V}, 
    \begin{align*}
     \|u_{n}^{\rm HF}\|_{H^{s+2}}&\leq \|u_{n}^{\rm HF}\|_{H^s}+ \|H_0 u_{n}^{\rm HF}\|_{H^s}\\
     &\leq \|Vu_{n}^{\rm HF}\|_{H^s}+\|W_{\gamma_*^{\rm HF}}u_{n}^{\rm HF}\|_{H^s}+(1+|\lambda_1|) \|u_{n}^{\rm HF}\|_{H^s}\lesssim \|u_n^{\rm HF}\|_{H^{s+1}}.
    \end{align*}
This completes the proof.
\end{proof}

\subsection{Proof of \texorpdfstring{\eqref{eq:E_c-EHF}}{}}\label{sec:4.2}
Now we prove \eqref{eq:E_c-EHF} by splitting \eqref{EDF-EHF} into three parts with $\gamma=\gamma_*^{\rm HF}$.
\subsubsection{Kinetic term} Note that $|\cD|=\sqrt{c^4-c^2\Delta}\leq c^2-\frac{1}{2}\Delta$. Then,
\begin{align}\label{eq:error-kinetic}
    \MoveEqLeft  \Tr_{\cH}[(\cD-c^2) \Lambda^+_c\gamma \Lambda^+_c]-\Tr_{\cH}[H_0 \gamma ] \notag\\
    &=\Tr_{\cH}[(|\cD|-c^2) \Lambda^+_c\gamma \Lambda^+_c]-\Tr_{\cH}[H_0 \gamma ]\notag\\
    &\leq \Tr_{\cH}[H_0 \Lambda^+_c\gamma \Lambda^+_c]-\Tr_{\cH}[H_0 \gamma ]=-\Tr_{\cH}[H_0 \Lambda^-_c\gamma \Lambda^-_c]\leq 0.
\end{align}

\subsubsection{Potential between electrons and nuclei} 
Concerning the potential between electrons and nuclei, we have the following.
\begin{lemma}\label{lem:error-V}
We have
    \begin{align*}
         \left|\Tr_{\cH}[V\gamma^{\rm HF}_{*}]-  \Tr_{\cH}[V\Lambda^+_c\gamma^{\rm HF}_{*} \Lambda^+_c]\right|\leq \frac{3z}{2c^2}\|\gamma^{\rm HF}_{*}\|_{X^4}.
    \end{align*}
\end{lemma}
\begin{proof}
We have
\begin{align}\label{eq:V1}
    \Tr_{\cH}[V\gamma^{\rm HF}_{*}]-  \Tr_{\cH}[V\Lambda^+_c\gamma^{\rm HF}_{*} \Lambda^+_c]=\Tr_{\cH}[V \Lambda^-_c \gamma^{\rm HF}_{*}]+ \Tr_{\cH}[V\Lambda^+_c\gamma^{\rm HF}_{*} \Lambda^-_c].
\end{align}

We study the first term on the right-hand side of \eqref{eq:V1}. From \eqref{eq:Hardy} and the fact that $\cK_{\rm L}\gamma^{\rm HF}_{*} \cK_{\rm L}=\gamma^{\rm HF}_{*}$, we have
\begin{align*}
    \left|\Tr_{\cH}[V \Lambda^-_c \gamma^{\rm HF}_{*}]\right|&=\left|\Tr_{\cH}[ \cK_{\rm L}\Lambda^-_c \gamma^{\rm HF}_{*}\cK_{\rm L} (-\Delta)^{1/2} (-\Delta)^{-1/2}V]\right|\\
    &\leq 2z\|\cK_{\rm L}\Lambda^-_c \gamma^{\rm HF}_{*}\cK_{\rm L}\nabla \|_{\mathfrak{S}_1}.
\end{align*}
Then from \eqref{eq:Holder} and Corollary \ref{cor:lambda-non-gamma}, we infer
\begin{align}\label{eq:K1-gamma-K1}
    \|\cK_{\rm L}\Lambda^-_c \gamma^{\rm HF}_{*}\cK_{\rm L}\nabla  \|_{\mathfrak{S}_1}&\leq \|\cK_{\rm L}\Lambda^-_c \gamma^{\rm HF}_{*}\Lambda^-_c\cK_{\rm L}\|_{\mathfrak{S}_1}^{1/2}\|\cK_{\rm L}\gamma^{\rm HF}_{*}\cK_{\rm L}\|_{X^2}^{1/2}\leq \frac{1}{4c^2}\|\gamma_*^{\rm HF}\|_{X^4}.
\end{align}
As a result,
\begin{align}\label{eq:V2}
    \left|\Tr_{\cH}[V \Lambda^-_c \gamma^{\rm HF}_{*}]\right|\leq \frac{z}{2c^2}\|\gamma_*^{\rm HF}\|_{X^4}.
\end{align}

Concerning the second term on the right-hand side of \eqref{eq:V1}, we have
\begin{align*}
    \Tr_{\cH}[V\Lambda^+_c\gamma^{\rm HF}_{*} \Lambda^-_c]=\sum_{j\in \{\rL,\rS\}}\Tr_{\cH}[V\cK_j\Lambda^+_c\gamma^{\rm HF}_{*} \Lambda^-_c \cK_j].
\end{align*}
Proceeding as for \eqref{eq:V2}, by \eqref{eq:Holder}, \eqref{eq:Hardy} and Corollary \ref{cor:lambda-non-gamma}, we have
\begin{align*}
  |\Tr_{\cH}[V\Lambda^+_c\gamma^{\rm HF}_{*} \Lambda^-_c]|\leq 2z\sum_{j\in \{\rL,\rS\}}\|\cK_j\Lambda^+_c\gamma^{\rm HF}_{*}\Lambda^+_c\cK_j \|_{X^2}^{1/2}\|\cK_j\Lambda^-_c\gamma^{\rm HF}_{*}\Lambda^-_c\cK_j \|_{\mathfrak{S}_1}^{1/2}\leq \frac{z}{c^2}\|\gamma_*^{\rm HF}\|_{X^4}.
\end{align*}
As a result, 
\begin{align}\label{eq:V-conclu}
    \left|\Tr_{\cH}[V\gamma^{\rm HF}_{*}]-  \Tr_{\cH}[V\Lambda^+_c\gamma^{\rm HF}_{*} \Lambda^+_c]\right|\leq  \frac{3z}{2c^2}\|\gamma^{\rm HF}_{*}\|_{X^4}.
\end{align}
This ends the proof.
\end{proof}
\medskip

\subsubsection{Potential between electrons and electrons} \label{sec:3.1.3}
First of all, we have
\begin{align}\label{eq:3.21}
\MoveEqLeft    \Tr_{\cH}[W_{\gamma^{\rm HF}_{*}}\gamma^{\rm HF}_{*}]-\Tr_{\cH}[W_{\Lambda^+_c\gamma^{\rm HF}_{*} \Lambda^+_c}\Lambda^+_c\gamma^{\rm HF}_{*} \Lambda^+_c]\notag\\
    &= \Tr_{\cH}[W_{\gamma^{\rm HF}_{*}}(\gamma^{\rm HF}_{*}-\Lambda^+_c\gamma^{\rm HF}_{*} \Lambda^+_c)]+\Tr_{\cH}[W_{\gamma^{\rm HF}_{*}-\Lambda^+_c\gamma^{\rm HF}_{*} \Lambda^+_c}\Lambda^+_c\gamma^{\rm HF}_{*} \Lambda^+_c].
\end{align}
Now we are going to study each term on the right-hand side separately. 

Concerning the first term on the right-hand side of \eqref{eq:3.21},
\begin{lemma}\label{lem:W2-1}We have
    \begin{align*}
        \left|\Tr_{\cH}[W_{\gamma^{\rm HF}_{*}}(\gamma^{\rm HF}_{*}-\Lambda^+_c\gamma^{\rm HF}_{*} \Lambda^+_c)]\right|\leq \frac{5q}{2c^2}\|\gamma_*^{\rm HF}\|_{X^4}.
    \end{align*}
\end{lemma}
\begin{proof}
According to \eqref{eq:W-dec}, $W_{\bullet}=W_{1,\bullet}-W_{2,\bullet}$. We can split $W_{\bullet}$ into terms $W_{1,\bullet}$ and $W_{2,\bullet}$. By \eqref{eq:Hardy},
\begin{align*}
    \|W_{1,\gamma} u\|\leq 2\|\gamma\|_{\mathfrak{S}_1}\|\nabla u\|_{\cH}\leq 2q\|\nabla u\|_{\cH}.
\end{align*}
Then replacing $V$ by $W_{1,\cdot}$ in \eqref{eq:V-conclu}, we infer
\begin{align}\label{eq:W1-conclu1}
    \left|\Tr_{\cH}[W_{1,\gamma^{\rm HF}_{*}}(\gamma^{\rm HF}_{*}- \Lambda^+_c\gamma^{\rm HF}_{*} \Lambda^+_c)]\right|\leq \frac{3}{2c^2}\|\gamma^{\rm HF}_{*}\|_{\mathfrak{S}_1}\|\gamma^{\rm HF}_{*}\|_{X^4}\leq \frac{3q}{2c^2}\|\gamma^{\rm HF}_{*}\|_{X^4}.
\end{align}

\medskip

We now consider the estimate for $W_{2,\bullet}$. First of all, we have
\begin{align*}
    \Tr_{\cH}[W_{2,\gamma^{\rm HF}_{*}}\gamma^{\rm HF}_{*}]-\Tr_{\cH}[W_{2,\gamma^{\rm HF}_{*}}\Lambda^+_c \gamma^{\rm HF}_{*}\Lambda^+_c]= \Tr_{\cH}[W_{2,\gamma^{\rm HF}_{*}}\Lambda^-_c\gamma^{\rm HF}_{*}]+\Tr_{\cH}[W_{2,\gamma^{\rm HF}_{*}}\Lambda^+_c \gamma^{\rm HF}_{*}\Lambda^-_c].
\end{align*}
Then by Lemma \ref{lem:W2-estimate}, 
\begin{align*}
    \|W_{2,\gamma^{\rm HF}_{*}}f\|_{\cH}\leq 2q\|\nabla f\|_{\cH}.
\end{align*}
Hence, as $\gamma^{\rm HF}_{*}=\cK_{\rm L}\gamma^{\rm HF}_{*}\cK_{\rm L}$, from \eqref{eq:Wpotential-decom}, \eqref{eq:W-K1-K2} and \eqref{eq:K1-gamma-K1}, we infer
\begin{align*}
    \left|\Tr_{\cH}[W_{2,\gamma^{\rm HF}_{*}}\Lambda^-_c\gamma^{\rm HF}_{*}]\right|&= \left|\Tr_{\cH}[\cK_{\rm L}\Lambda^-_c\gamma^{\rm HF}_{*}\cK_{\rm L} W_{2,\gamma^{\rm HF}_{*}}]\right|\\
    &\leq 2q\|\cK_{\rm L}\Lambda^-_c\gamma^{\rm HF}_{*}\cK_{\rm L}\nabla\|_{\mathfrak{S}_1}\leq \frac{q}{2c^2}\|\gamma_*^{\rm HF}\|_{X^4}.
\end{align*}
Analogously, by Corollary \ref{cor:lambda-non-gamma},
\begin{align*}
   \left|\Tr_{\cH}[W_{2,\gamma^{\rm HF}_{*}}\Lambda^+_c\gamma^{\rm HF}_{*}\Lambda^-_c]\right| \leq 2q\|\cK_{\rm L}\Lambda^-_c \gamma^{\rm HF}_{*}\Lambda^-_c\cK_{\rm L}\|_{\mathfrak{S}_1}^{1/2}\|\cK_{\rm L}\Lambda^+_c\gamma^{\rm HF}_{*}\Lambda^+_c\cK_{\rm L}\|_{X^2}^{1/2} \leq \frac{q}{2c^2}\|\gamma_*^{\rm HF}\|_{X^4}.
\end{align*}
Now we conclude that 
\begin{align*}
  \MoveEqLeft  \left|\Tr_{\cH}[W_{\gamma^{\rm HF}_{*}}(\gamma^{\rm HF}_{*}-\Lambda^+_c\gamma^{\rm HF}_{*} \Lambda^+_c)]\right|\\
  &\leq  \left|\Tr_{\cH}[W_{1,\gamma^{\rm HF}_{*}}(\gamma^{\rm HF}_{*}-\Lambda^+_c\gamma^{\rm HF}_{*} \Lambda^+_c)]\right|\\
  &\quad +\left|\Tr_{\cH}[W_{2,\gamma^{\rm HF}_{*}}(\gamma^{\rm HF}_{*}-\Lambda^+_c\gamma^{\rm HF}_{*} \Lambda^+_c)]\right|\leq \frac{5q}{2c^2}\|\gamma_*^{\rm HF}\|_{X^4}.
\end{align*}
This proves the lemma.
\end{proof}
Concerning the second term on the right-hand side of \eqref{eq:3.21},
\begin{lemma}\label{lem:W2-2}
    \begin{align*}
        \left|\Tr_{\cH}[W_{\gamma^{\rm HF}_{*}-\Lambda^+_c\gamma^{\rm HF}_{*} \Lambda^+_c}\Lambda^+_c\gamma^{\rm HF}_{*} \Lambda^+_c]\right|\leq  \frac{1}{2c^2}(5q+4\|\gamma_{*}^{\rm HF}\|_{X^4})\|\gamma_{*}^{\rm HF}\|_{X^4}.
    \end{align*}
\end{lemma}
\begin{proof}
First of all, note that
\begin{align}\label{eq:W2-2-dec}
  \MoveEqLeft  \left|\Tr_{\cH}[W_{\gamma^{\rm HF}_{*}-\Lambda^+_c\gamma^{\rm HF}_{*} \Lambda^+_c}\Lambda^+_c\gamma^{\rm HF}_{*} \Lambda^+_c]\right|\notag\\
  &=\left|\Tr_{\cH}[W_{\Lambda^+_c\gamma^{\rm HF}_{*} \Lambda^+_c}(\gamma^{\rm HF}_{*}-\Lambda^+_c\gamma^{\rm HF}_{*} \Lambda^+_c)]\right|\notag\\
    &\leq \sum_{m=1,2}\left|\Tr_{\cH}[W_{m,\Lambda^+_c\gamma^{\rm HF}_{*} \Lambda^+_c}(\gamma^{\rm HF}_{*}-\Lambda^+_c\gamma^{\rm HF}_{*} \Lambda^+_c)]\right|.
\end{align}
Then it suffice to study $\Tr[W_{m,\Lambda^+_c\gamma^{\rm HF}_{*} \Lambda^+_c}(\gamma^{\rm HF}_{*}-\Lambda^+_c\gamma^{\rm HF}_{*} \Lambda^+_c)]$ for $m=1,2$.

Analogous to \eqref{eq:W1-conclu1}, we have
\begin{align}\label{eq:W2-2-1}
\MoveEqLeft\left|\Tr_{\cH}[W_{1,\Lambda^+_c\gamma^{\rm HF}_{*} \Lambda^+_c}(\gamma^{\rm HF}_{*}-\Lambda^+_c\gamma^{\rm HF}_{*} \Lambda^+_c)]\right|\notag\\
&\leq \frac{3}{2c^2}\|\Lambda^+_c\gamma^{\rm HF}_{*} \Lambda^+_c\|_{\mathfrak{S}_1}\|\gamma^{\rm HF}_{*}\|_{X^4}\leq \frac{3q}{2c^2}\|\gamma^{\rm HF}_{*}\|_{X^4}.
\end{align}
Thus it remains to study the term associated with $W_{2,\bullet}$. We have
\begin{align}\label{eq:W2-2-2-dec}
  \MoveEqLeft \Tr_{\cH}\left[W_{2,\Lambda^+_c\gamma^{\rm HF}_{*} \Lambda^+_c}(\gamma^{\rm HF}_{*}-\Lambda^+_c\gamma^{\rm HF}_{*} \Lambda^+_c)\right]\notag\\
  &=\Tr_{\cH}\left[W_{2,\Lambda^+_c\gamma^{\rm HF}_{*} \Lambda^+_c}\Lambda^-_c\gamma^{\rm HF}_{*} \right]+\Tr_{\cH}\left[W_{2,\Lambda^+_c\gamma^{\rm HF}_{*} \Lambda^+_c}\Lambda^+_c\gamma^{\rm HF}_{*} \Lambda^-_c\right].
\end{align}
As $\gamma^{\rm HF}_{*}=\gamma^{\rm HF}_{*}\cK_{\rm L}$, using again \eqref{eq:Wpotential-decom} and \eqref{eq:W-K1-K2}, we infer
\begin{align}\label{eq:W2-2-2-1-dec}
\MoveEqLeft\left| \Tr_{\cH}[W_{2,\Lambda^+_c\gamma^{\rm HF}_{*} \Lambda^+_c}\Lambda^-_c\gamma^{\rm HF}_{*} ]\right|\leq \sum_{j\in \{\rL,\rS\}}\left| \Tr_{\cH}[W_{2,\cK_{\rm L}\Lambda^+_c\gamma^{\rm HF}_{*} \Lambda^+_c\cK_j}\cK_j\Lambda^-_c\gamma^{\rm HF}_{*} \cK_{\rm L}]\right|.
\end{align}
For $j=\rL$, by Lemma \ref{lem:W2-estimate},
\begin{align*}
\MoveEqLeft    \left| \Tr_{\cH}[W_{2,\cK_{\rm L}\Lambda^+_c\gamma^{\rm HF}_{*} \Lambda^+_c\cK_{\rm L}}\cK_{\rm L}\Lambda^-_c\gamma^{\rm HF}_{*} \cK_{\rm L}]\right|\\
&=\left| \Tr_{\cH}[\cK_{\rm L}\Lambda^-_c\gamma^{\rm HF}_{*} \cK_{\rm L} W_{2,\cK_{\rm L}\Lambda^+_c\gamma^{\rm HF}_{*} \Lambda^+_c\cK_{\rm L}}]\right|\\
    &\leq 2\|\cK_{\rm L} \Lambda^+_c\gamma_*^{\rm HF}\Lambda^+_c \cK_{\rm L} \|_{\mathfrak{S}_1}\|\cK_{\rm L}\Lambda^-_c\gamma^{\rm HF}_{*} \cK_{\rm L}\nabla\|_{\mathfrak{S}_1}.
\end{align*}
Then, from \eqref{eq:K1-gamma-K1} and Corollary \ref{cor:lambda-non-gamma},
\begin{align*}
    \left| \Tr_{\cH}[W_{2,\cK_{\rm L}\Lambda^+_c\gamma^{\rm HF}_{*} \Lambda^+_c\cK_{\rm L}}\cK_{\rm L}\Lambda^-_c\gamma^{\rm HF}_{*} \cK_{\rm L}]\right|\leq 2q\|\cK_{\rm L}\Lambda^-_c\gamma^{\rm HF}_{*} \cK_{\rm L}\nabla\|_{\mathfrak{S}_1}\leq \frac{q}{2c^2}\|\gamma^{\rm HF}_{*}\|_{X^4}.
\end{align*}
For $j=\rS$, from Corollary \ref{cor:lambda-non-gamma} and \eqref{eq:5.18},
\begin{align*}
\MoveEqLeft    \left| \Tr_{\cH}[W_{2,\cK_{\rm L}\Lambda^+_c\gamma^{\rm HF}_{*} \Lambda^+_c\cK_{\rm S}}\cK_{\rm S}\Lambda^-_c\gamma^{\rm HF}_{*} \cK_{\rm L}]\right|\\
    &\leq 2\|\cK_{\rm L} \Lambda^+_c\gamma_*^{\rm HF}\Lambda^+_c \cK_{\rm L} \|_{\mathfrak{S}_1}^{1/2}\|\cK_{\rm S} \Lambda^+_c\gamma_*^{\rm HF}\Lambda^+_c \cK_{\rm S} \|_{\mathfrak{S}_1}^{1/2}\|\nabla\cK_{\rm S}\Lambda^-_c\gamma^{\rm HF}_{*} \cK_{\rm L}\|_{\mathfrak{S}_1}\\
    &\leq \frac{q^{1/2}}{c}\|\gamma_*^{\rm HF}\|_{X^2}^{1/2}\|\nabla\cK_{\rm S}\Lambda^-_c\gamma^{\rm HF}_{*} \cK_{\rm L}\|_{\mathfrak{S}_1}.
\end{align*}
Then by H\"older's inequality \eqref{eq:Holder} and Corollary \ref{cor:lambda-non-gamma}, 
\begin{align*}
\MoveEqLeft    \left| \Tr_{\cH}[W_{2,\cK_{\rm L}\Lambda^+_c\gamma^{\rm HF}_{*} \Lambda^+_c\cK_{\rm S}}\cK_{\rm S}\Lambda^-_c\gamma^{\rm HF}_{*} \cK_{\rm L}]\right|\\
    &\leq \frac{q^{1/2}}{c}\|\gamma_*^{\rm HF}\|_{X^2}^{1/2}\|\cK_{\rm L} \gamma_*^{\rm HF}\cK_{\rm L}\|_{\mathfrak{S}_1}^{1/2}\|\cK_{\rm S} \Lambda^-_c\gamma_*^{\rm HF}\Lambda^-_c\cK_{\rm S}\|_{X^2}^{1/2}\\
    &\leq \frac{q}{2c^2}\|\gamma_*^{\rm HF}\|_{X^2}^{1/2}\|\gamma_*^{\rm HF}\|_{X^4}^{1/2}\leq  \frac{q}{2c^2}\|\gamma_*^{\rm HF}\|_{X^4}.
\end{align*}
Thus, for the first term on the right-hand side of \eqref{eq:W2-2-2-dec}, from \eqref{eq:W2-2-2-1-dec} we get
\begin{align}\label{eq:W2-2-2-1}
     \MoveEqLeft \left| \Tr_{\cH}[W_{2,\Lambda^+_c\gamma^{\rm HF}_{*} \Lambda^+_c}\Lambda^-_c\gamma^{\rm HF}_{*} ]\right|\leq \frac{q}{c^2}\|\gamma^{\rm HF}_{*} \|_{X^4}.
\end{align}

\medskip

Concerning the second term on the right-hand side of \eqref{eq:W2-2-2-dec}, we proceed as for \eqref{eq:W2-2-2-1}. By \eqref{eq:Wpotential-decom}, H\"older's inequality \eqref{eq:Holder} and Lemma \ref{lem:W2-estimate},
\begin{align*}
   \MoveEqLeft \left|\Tr_{\cH}[W_{2,\Lambda^+_c\gamma^{\rm HF}_{*} \Lambda^+_c}\Lambda^+_c\gamma^{\rm HF}_{*} \Lambda^-_c]\right|\\
   &\leq\sum_{j,j'\in\{\rL,\rS\}}\left|\Tr_{\cH}[W_{2,\cK_j\Lambda^+_c\gamma^{\rm HF}_{*} \Lambda^+_c\cK_{j'}}\cK_{j'}\Lambda^+_c\gamma^{\rm HF}_{*} \Lambda^-_c\cK_j]\right|\\
    &\leq 2\sum_{j,j'\in\{\rL,\rS\}}\|\cK_j\Lambda^+_c\gamma^{\rm HF}_{*} \Lambda^+_c\cK_{j}\|_{\mathfrak{S}_1}^{1/2}\|\cK_{j'}\Lambda^+_c\gamma^{\rm HF}_{*} \Lambda^+_c\cK_{j'}\|_{\mathfrak{S}_1}^{1/2}\|\nabla\cK_{j'}\Lambda^+_c\gamma^{\rm HF}_{*} \Lambda^-_c\cK_j\|_{\mathfrak{S}_1}\\
    &\leq 2 \left(\sum_{j\in \{\rL,\rS\}}\|\cK_j\Lambda^+_c\gamma^{\rm HF}_{*} \Lambda^+_c\cK_{j}\|_{\mathfrak{S}_1}^{1/2}\|\cK_{j}\Lambda^-_c\gamma^{\rm HF}_{*} \Lambda^-_c\cK_{j}\|_{\mathfrak{S}_1}^{1/2}\right)\\
    &\quad\times \left(\sum_{j'\in \{\rL,\rS\}}\|\cK_{j'}\Lambda^+_c\gamma^{\rm HF}_{*} \Lambda^+_c\cK_{j'}\|_{\mathfrak{S}_1}^{1/2}\|\cK_{j'}\Lambda^+_c\gamma^{\rm HF}_{*} \Lambda^+_c\cK_{j'}\|_{X^2}^{1/2}\right).
\end{align*}
By Corollary \ref{cor:lambda-non-gamma},
\begin{align*}
 \MoveEqLeft   \sum_{j\in \{\rL,\rS\}} \|\cK_j\Lambda^+_c\gamma^{\rm HF}_{*} \Lambda^+_c\cK_{j}\|_{\mathfrak{S}_1}^{1/2}\|\cK_{j}\Lambda^-_c\gamma^{\rm HF}_{*} \Lambda^-_c\cK_{j}\|_{\mathfrak{S}_1}^{1/2}\\
    &\leq \frac{1}{4c^2}\left(\|\gamma_*^{\rm HF}\|_{\mathfrak{S}_1}^{1/2}\|\gamma_*^{\rm HF}\|_{X^4}^{1/2}+ \|\gamma_*^{\rm HF}\|_{X^2}\right)\leq \frac{1}{2c^2}\|\gamma_*^{\rm HF}\|_{X^4}
\end{align*}
and
\begin{align*}
 \MoveEqLeft   \sum_{j'\in \{\rL,\rS\}}\|\cK_{j'}\Lambda^+_c\gamma^{\rm HF}_{*} \Lambda^+_c\cK_{j'}\|_{\mathfrak{S}_1}^{1/2}\|\cK_{j'}\Lambda^+_c\gamma^{\rm HF}_{*} \Lambda^+_c\cK_{j'}\|_{X^2}^{1/2}\leq 2\|\gamma^{\rm HF}_{*}\|_{X^4}.
\end{align*}
Hence,
\begin{align*}
    \left|\Tr_{\cH}[W_{2,\Lambda^+_c\gamma^{\rm HF}_{*} \Lambda^+_c}\Lambda^+_c\gamma^{\rm HF}_{*} \Lambda^-_c]\right|\leq \frac{2}{c^2}\|\gamma_{*}^{\rm HF}\|_{X^4}^2.
\end{align*}
Then from \eqref{eq:W2-2-2-dec},
\begin{align}\label{eq:W2-2-2}
    |\Tr_{\cH}[W_{2,\Lambda^+_c\gamma^{\rm HF}_{*} \Lambda^+_c}(\gamma^{\rm HF}_{*}-\Lambda^+_c\gamma^{\rm HF}_{*} \Lambda^+_c)]|\leq \frac{1}{c^2}(q+\|\gamma_{*}^{\rm HF}\|_{X^4})\|\gamma_{*}^{\rm HF}\|_{X^4}.
\end{align}

Finally, we can conclude from \eqref{eq:W2-2-dec}, \eqref{eq:W2-2-1} and \eqref{eq:W2-2-2-dec} that
\begin{align*}
   \left|\Tr_{\cH}[W_{\gamma^{\rm HF}_{*}-\Lambda^+_c\gamma^{\rm HF}_{*} \Lambda^+_c}\Lambda^+_c\gamma^{\rm HF}_{*} \Lambda^+_c]\right| \leq \frac{1}{2c^2}(5q+4\|\gamma_{*}^{\rm HF}\|_{X^4})\|\gamma_{*}^{\rm HF}\|_{X^4}.
\end{align*}
This ends the proof.
\end{proof}

\subsubsection{End of the proof of \texorpdfstring{\eqref{eq:E_c-EHF}}{}} 
From \eqref{eq:3.21},  Lemma \ref{lem:W2-1} and Lemma \ref{lem:W2-2} we know that
\begin{align*}
\left| \Tr_{\cH}[W_{\gamma^{\rm HF}_{*}}\gamma^{\rm HF}_{*}]-\Tr_{\cH}[W_{\Lambda^+_c\gamma^{\rm HF}_{*} \Lambda^+_c}\Lambda^+_c\gamma^{\rm HF}_{*} \Lambda^+_c]\right|\leq \frac{1}{2c^2}(10q+4\|\gamma_*^{\rm HF}\|_{X^4})\|\gamma_*^{\rm HF}\|_{X^4}.
\end{align*}
Then estimate \eqref{eq:E_c-EHF} follows from above estimate, \eqref{EDF-EHF}, \eqref{eq:error-kinetic} and Lemma \ref{lem:error-V}.

\subsection{Proof of \texorpdfstring{\eqref{eq:E_c-EHF-c2}}{}}\label{sec:4.3}
Now we turn to prove \eqref{eq:E_c-EHF-c2}. Before considering \eqref{EDF-EHF}, we first construct the test density matrix $\widetilde{\gamma}_c^{\rm HF}$ and study its property.

\subsubsection{Construction of \texorpdfstring{$\widetilde{\gamma}_c^{\rm HF}$}{}}\label{sec:4.3.1}

According to Remark \ref{rem:non-unfill},
\begin{align*}
    \gamma_*^{\rm HF}:= \sum_{n= 1}^q |u_n^{\rm HF} \rangle\,\langle u_n^{\rm HF}|
\end{align*}
with $\left<u_m^{\rm HF}, u_n^{\rm HF}\right>_{\cH_\rL}=\delta_{m,n}$ and $u_n\in\HL$. Thus $\gamma^{\rm HF}_*$ is the projection on the space spanned by $\{u_{n}^{\rm HF}\}_{1\leq n\leq q}$. To get the term $\cE^{(2)}_c(\gamma_*^{\rm HF})$ in \eqref{eq:E_c-EHF-c2}, we need to renormalize the density matrix $\gamma_*^{\rm HF}$. Inspired by \eqref{eq:D-Delta}, we consider the projection on the space spanned by $\{\cS u_n^{\rm HF}\}_{1\leq n\leq q}$. Note that
\begin{align*}
 \left<\cS u_m^{\rm HF},\cS u_n^{\rm HF}\right>_{\cH}= \delta_{m,n}+\frac{1}{4c^2}\left<\cL u_{m}^{\rm HF},\cL u_{n}^{\rm HF}\right>_\HL.
\end{align*}
Thus $\{\cS u_n^{\rm HF}\}_{1\leq n\leq q}$ are no longer orthogonal functions. We now introduce the following overlap matrix:
\begin{align}
    S_{\rm HF}:&= \Big(\left<\cS u_m^{\rm HF}, \cS u_n^{\rm HF}\right>_\cH\Big)_{1\leq m,n\leq q}=\1_{q\times q}+\frac{1}{4c^2} \widetilde{S}_{\rm HF},\label{eq:S}\\
     \widetilde{S}_{\rm HF}:&= \left(\left<\cL u_{m}^{\rm HF},\cL u_{n}^{\rm HF}\right>_\HL\right)_{1\leq m,n\leq q}.\label{eq:S'}
\end{align}
Note that for $c$ large enough, as $\cL^2=-\Delta$,
\begin{align}\label{eq:diagonally dominated1}
\sup_{1\leq m,n\leq q}\frac{1}{4c^2}\left|\left<\cL u_{m}^{\rm HF},\cL u_{n}^{\rm HF}\right>_\HL\right|=\sup_{1\leq m,n\leq q}\frac{1}{4c^2}\left|\left<\sqrt{-\Delta} u_{m}^{\rm HF},\sqrt{-\Delta} u_{n}^{\rm HF}\right>_\HL\right| =\cO(c^{-2}).
\end{align}
 This implies for $c$ large enough, $S_{\rm HF}$ is a strictly diagonally dominated matrix. Thus $S_{\rm HF}$ is invertible, 
\begin{align}\label{eq:S-1}
    \left\|S_{\rm HF}^{-1}-\left(\1_{q\times q}-\frac{1}{4c^2}\widetilde{S}_{\rm HF}\right)\right\|_{\mathfrak{S}_1(\C^q)}=\cO(c^{-4})
\end{align}
and there exists some constants $0<C_0<1<C_1$ such that for any $c$ large enough,
\begin{align}\label{eq:S-1-bound}
    C_0\1_{q\times q}\leq S_{\rm HF}^{-1}\leq C_1\1_{q\times q}
\end{align}
as a matrix. By using \eqref{eq:matrix-bra-ket}, the renormalized HF density matrix $\widetilde{\gamma}_c^{\rm HF}$ on the space spanned by $\{\cS u_n^{\rm HF}\}_{1\leq n\leq q}$ is defined as 
\begin{align}\label{eq:gammaHF'}
    \widetilde{\gamma}_c^{\rm HF}:&=\Big(\left|\cS u_1^{\rm HF}\right>,\cdots, \left|\cS u_q^{\rm HF}\right>\Big) S_{\rm HF}^{-1} \begin{pmatrix}
        \left<\cS u_1^{\rm HF}\right|\\
        \vdots\\
        \left<\cS u_q^{\rm HF}\right|
    \end{pmatrix}.
\end{align}
In particular, it is easy to see that
\begin{align}
    \widetilde{\gamma}_c^{\rm HF}=\cS\widetilde{\gamma}_c^{\rm HF}\cS^*.
\end{align}

\subsubsection{Property of \texorpdfstring{$\widetilde{\gamma}_c^{\rm HF}$}{}}\label{sec:4.3.2'}

By \eqref{eq:matrix-inner}, \eqref{eq:S} and direct calculation, we know that
\begin{align*}
    (\widetilde{\gamma}_c^{\rm HF})^2&=\Big(\left|\cS u_1^{\rm HF}\right>,\cdots, \left|\cS u_q^{\rm HF}\right>\Big) S_{\rm HF}^{-1} \Big(\left<\cS u_{m}^{\rm HF},\cS u_{n}^{\rm HF}\right>\Big)_{1\leq m,n\leq q} S_{\rm HF}^{-1}\begin{pmatrix}
        \left<\cS u_1^{\rm HF}\right|\\
        \vdots\\
        \left<\cS u_q^{\rm HF}\right|
    \end{pmatrix}\\
    &=\Big(\left|\cS u_1^{\rm HF}\right>,\cdots, \left|\cS u_n^{\rm HF}\right>\Big) S_{\rm HF}^{-1}\begin{pmatrix}
        \left<\cS u_1^{\rm HF}\right|\\
        \vdots\\
        \left<\cS u_n^{\rm HF}\right|
    \end{pmatrix}=\widetilde{\gamma}_c^{\rm HF}.
\end{align*}

This and \eqref{eq:S-1-bound} imply
\begin{lemma}\label{lem:gammaHF'-projection}
For $c$ large enough, $\widetilde{\gamma}_c^{\rm HF}=\cS\widetilde{\gamma}_c^{\rm HF}\cS^*$ is a projector with ${\rm Rank}(\widetilde{\gamma}_c^{\rm HF})=q$, thus $\widetilde{\gamma}_c^{\rm HF}\in \Gamma_q$.
\end{lemma}
Then we have the following estimates on $\widetilde{\gamma}_c^{\rm HF}$.
\begin{lemma}\label{lem:gammaHF'-estimate}
Under Assumption \ref{ass:V}, for $c$ large enough,
\begin{align*}
   \|\widetilde{\gamma}_c^{\rm HF}\|_{X^4}=\cO(1),\qquad  \|\cK_\rL \widetilde{\gamma}_c^{\rm HF}\cK_\rL \|_{X^6}=\cO(1),\qquad     \|\cK_\rS \widetilde{\gamma}_c^{\rm HF}\cK_\rS \|_{X^4}=\cO(c^{-2}).
\end{align*}
\end{lemma}
\begin{proof}
By \eqref{eq:S-1-bound}, 
\begin{align}\label{eq:4.37}
   0\leq  C_0 \cS\gamma_*^{\rm HF}\cS^*\leq \widetilde{\gamma}_c^{\rm HF}\leq C_1 \cS\gamma_*^{\rm HF}\cS^*.
\end{align}
This and $\cL^2=-\Delta$ give
\begin{align*}
     \|\widetilde{\gamma}_c^{\rm HF}\|_{X^4}&=\Tr_\cH\left[(1-\Delta)\widetilde{\gamma}_c^{\rm HF}(1-\Delta)\right]\\
     &\leq  C_1\Tr_\cH\left[(1-\Delta)\cS\gamma^{\rm HF}_*\cS^*(1-\Delta)\right]\\
     &= C_1\Tr_\cH\left[(1-\Delta)\gamma^{\rm HF}_*(1-\Delta)\right]+\frac{C_1}{4c^2}\Tr_\cH\left[(1-\Delta)\cL \gamma^{\rm HF}_*\cL(1-\Delta)\right]\\
     &\lesssim \|\gamma^{\rm HF}_*\|_{X^6}=\cO(1).
\end{align*}
where we used 
\begin{align*}
    \Tr_\cH\left[(1-\Delta)\cK_\rS\cS\gamma^{\rm HF}_*\cS^*\cK_\rS(1-\Delta)\right]=\frac{1}{4c^2}\Tr_\cH\left[(1-\Delta)\cL \gamma^{\rm HF}_*\cL(1-\Delta)\right]
\end{align*}
since the last $2\times 2$ submatrix of $\cK_\rS\cS\gamma_*^{\rm HF}\cS^*\cK_\rS$ is the first $2\times 2$ submatrix of $ \frac{1}{4c^2}\cL\gamma_*^{\rm HF}\cL$, i.e.,
\begin{align*}
    \cK_\rS\cS\gamma_*^{\rm HF}\cS^*\cK_\rS =\begin{pmatrix}
       0_{2\times 2} &0_{2\times 2}\\
       0_{2\times 2} & A^{\rm HF}
    \end{pmatrix},\qquad 
    \frac{1}{4c^2}\cL\gamma_*^{\rm HF}\cL=\begin{pmatrix}
       A^{\rm HF} &0_{2\times 2}\\
        0_{2\times 2} &  0_{2\times 2}
    \end{pmatrix}
\end{align*}
for some $A^{\rm HF}\in \mathfrak{S}_1(L^2(\R^3;\C^2))$. 

Next by \eqref{eq:4.37}, for $j\in\{\rL,\rS\}$,
\begin{align*}
      0\leq  C_0 \cK_j\cS\gamma_*^{\rm HF}\cS^* \cK_j\leq \cK_j \widetilde{\gamma}_c^{\rm HF}\cK_j \leq C_1 \cK_j\cS\gamma_*^{\rm HF}\cS^*\cK_j.
\end{align*}
As $\cK_\rL\cS\gamma_*^{\rm HF}\cS^*\cK_\rL=\gamma_*^{\rm HF}$, we get
\begin{align*}
    \|\cK_\rL \widetilde{\gamma}_c^{\rm HF}\cK_\rL \|_{X^6}&=\Tr_\cH\left[(1-\Delta)^{3/2}\cK_\rL \widetilde{\gamma}_c^{\rm HF}\cK_\rL (1-\Delta)^{3/2}\right]\\
    &\leq C_1\Tr_\cH\left[(1-\Delta)^{3/2}\cK_\rL \cS\gamma^{\rm HF}_*\cS^*\cK_\rL (1-\Delta)^{3/2}\right]\\
    &= C_1\Tr_\cH\left[(1-\Delta)^{3/2}\gamma^{\rm HF}_*(1-\Delta)^{3/2}\right]=C_1\|\gamma^{\rm HF}_*\|_{X^6}=\cO(1),
\end{align*}
and
\begin{align*}
    \|\cK_\rS \widetilde{\gamma}_c^{\rm HF}\cK_\rS \|_{X^4}&=\Tr_\cH\left[(1-\Delta)\cK_\rS \widetilde{\gamma}_c^{\rm HF}\cK_\rS (1-\Delta)\right]\\
    &\leq C_1\Tr_\cH\left[(1-\Delta)\cK_\rS \cS\gamma^{\rm HF}_*\cS^*\cK_\rS (1-\Delta)\right]\\
    &=\frac{C_1}{4c^2}\Tr_\cH\left[(1-\Delta)\cL\gamma^{\rm HF}_*\cL (1-\Delta)\right]=\frac{C_1}{4c^2}\|\gamma^{\rm HF}_*\|_{X^6}=\cO(c^{-2}).
\end{align*}
This ends the proof.
\end{proof}

By \eqref{eq:S-1}, the density matrix $\widetilde{\gamma}_c^{\rm HF}$ can be further approximated by  $ \widetilde{\widetilde{\gamma}}_c^{\rm HF}$ which is defined by
\begin{align}\label{eq:gammaHF''}
    \widetilde{\widetilde{\gamma}}_c^{\rm HF}:&=\Big(\left|\cS u_1^{\rm HF}\right>,\cdots, \left|\cS u_n^{\rm HF}\right>\Big)\left(\1_{q\times q}-\frac{1}{4c^2}\widetilde{S}_{\rm HF}\right)\begin{pmatrix}
        \left<\cS u_1^{\rm HF}\right|\\
        \vdots\\
        \left<\cS u_n^{\rm HF}\right|
    \end{pmatrix}\notag\\
    &= \cS \gamma_*^{\rm HF}\cS^*  -\frac{1}{4c^2}\sum_{1\leq m,n\leq q}  \left<\cL u_{m}^{\rm HF},\cL u_{n}^{\rm HF}\right>_\HL\left|\cS u_m^{\rm HF}\right>\left<\cS u_n^{\rm HF}\right|.
\end{align}
We have
\begin{lemma}\label{lem: gammaHF''-gammaHF'}
    For $c$ large enough, we have
    \begin{align*}
        \|\widetilde{\widetilde{\gamma}}_c^{\rm HF}- \widetilde{\gamma}_c^{\rm HF} \|_{X} =\cO(c^{-4}),\qquad \|\cK_\rL(\widetilde{\widetilde{\gamma}}_c^{\rm HF}- \widetilde{\gamma}_c^{\rm HF})\cK_\rL \|_{X^2} =\cO(c^{-4}).
    \end{align*}
\end{lemma}
\begin{proof}
    Let $A= S_{\rm HF}^{-1}-\left(\1_{q\times q}-\frac{1}{4c^2}\widetilde{S}_{\rm HF}\right)$. As a symmetric $q\times q$ matrix, there exists a unitary matrix $U$ such that
    \begin{align}\label{eq:A-S1}
        A=U^* {\rm Diag}(\lambda_1(A),\cdots,\lambda_q(A)) U
    \end{align}
    with $\lambda_1(A),\cdots, \lambda_q(A)$ being the eigenvalues of $A$. In particular, from \eqref{eq:S-1}, we infer
    \begin{align*}
        \|A\|_{\mathfrak{S}_1(\C^q)}=\sum_{1\leq n\leq q}|\lambda_n(A)|=\cO(c^{-4}).
    \end{align*}
    Let 
    \begin{align*}
        \begin{pmatrix}
            v_1\\
            \vdots\\
            v_q
        \end{pmatrix} = U \begin{pmatrix}
            u_1^{\rm HF}\\
            \vdots\\ 
             u_q^{\rm HF}
        \end{pmatrix}.
    \end{align*}
Then
\begin{align*}
    \widetilde{\widetilde{\gamma}}_c^{\rm HF}- \widetilde{\gamma}_c^{\rm HF}= \sum_{n=1}^q \lambda_n(A) \left|\cS v_n\right>\left< \cS v_n\right|.
\end{align*}
By \eqref{eq:A-S1} and Lemma \ref{cor:regularity-H3},
\begin{align*}
    \|\cK_\rL (\widetilde{\widetilde{\gamma}}_c^{\rm HF}- \widetilde{\gamma}_c^{\rm HF}
    )\cK_\rL \|_{X^2} \leq \sum_{n=1}^q |\lambda_n(A)|\|v_n\|_{H^1}^2\lesssim \|A\|_{\mathfrak{S}_1(\C^q)}= \cO(c^{-4})
\end{align*}
and
\begin{align*}
         \|\widetilde{\widetilde{\gamma}}_c^{\rm HF}- \widetilde{\gamma}_c^{\rm HF}\|_{X} \leq \sum_{n=1}^q |\lambda_n(A)|\|\cS v_n\|_{H^{1/2}}^2= \cO(c^{-4}) .
\end{align*}
Here we used
\begin{align*}
    |\lambda_n(A)|\leq \|A\|_{\cB(\C^q)}=\cO(c^{-4}).
\end{align*}
This ends the proof.
\end{proof}
Thus,
\begin{lemma}\label{lem:3.13}
    For $c$ large enough, we have
    \begin{align*}
        \cE_{c}(\widetilde{\gamma}_c^{\rm HF})&= \cE_{c}(\widetilde{\widetilde{\gamma}}_c^{\rm HF})+\cO(c^{-4})= E_q^{\rm HF}+\cE^{(2)}_c(\gamma_*^{\rm HF})+\cO(c^{-4}).
    \end{align*}
\end{lemma}
\begin{proof}
We have
\begin{align*}
    \cE_{c}(\widetilde{\gamma}_c^{\rm HF}) - \cE_{c}(\widetilde{\widetilde{\gamma}}_c^{\rm HF})&= \Tr_{\cH}[(\cD-c^2)(\widetilde{\gamma}_c^{\rm HF} - \widetilde{\widetilde{\gamma}}_c^{\rm HF})]+\Tr_{\cH}[(-V+W_{\widetilde{\gamma}_c^{\rm HF} + \widetilde{\widetilde{\gamma}}_c^{\rm HF}})(\widetilde{\gamma}_c^{\rm HF} - \widetilde{\widetilde{\gamma}}_c^{\rm HF})].
\end{align*}
By \eqref{eq:D-Delta} and \eqref{lem: gammaHF''-gammaHF'}, 
\begin{align*}
    \left|\Tr_{\cH}[(\cD-c^2)(\widetilde{\gamma}_c^{\rm HF} - \widetilde{\widetilde{\gamma}}_c^{\rm HF})]\right|&=  \left|\Tr_{\cH}[H_0\cK_{\rL}(\widetilde{\gamma}_c^{\rm HF} - \widetilde{\widetilde{\gamma}}_c^{\rm HF})\cK_\rL]\right|\\
    &\leq  \|\cK_\rL (\widetilde{\widetilde{\gamma}}_c^{\rm HF}- \widetilde{\gamma}_c^{\rm HF}
    )\cK_\rL \|_{X^2}=\cO(c^{-4})
\end{align*}
and by Kato's inequality and \eqref{eq:5.1-1},
\begin{align*}
    \left|\Tr_{\cH}\left[\left(-V+\frac{1}{2}W_{\widetilde{\gamma}_c^{\rm HF} + \widetilde{\widetilde{\gamma}}_c^{\rm HF}}\right)(\widetilde{\gamma}_c^{\rm HF} - \widetilde{\widetilde{\gamma}}_c^{\rm HF})\right]\right| \lesssim  \|\widetilde{\widetilde{\gamma}}_c^{\rm HF}- \widetilde{\gamma}_c^{\rm HF} \|_{X} =\cO(c^{-4}).
\end{align*}
Thus,
\begin{align*}
     \cE_{c}(\widetilde{\gamma}_c^{\rm HF})&= \cE_{c}(\widetilde{\widetilde{\gamma}}_c^{\rm HF})+\cO(c^{-4}).
\end{align*}

Now we prove
\begin{align*}
   \cE_{c}(\widetilde{\widetilde{\gamma}}_c^{\rm HF})=\cE^{\rm HF}(\gamma_*^{\rm HF})+\cE^{(2)}_c(\gamma_*^{\rm HF})+\cO(c^{-4}).
\end{align*}
By \eqref{eq:D-Delta}, \eqref{eq:gammaHF''} and $\cL^2=2H_0$, we have
\begin{align*}
    \Tr_\cH(H_0 \gamma^{\rm HF}_*)&= \Tr_\cH[(\cD-c^2)\cS \gamma^{\rm HF}_* \cS^* ]\\
    &= \Tr_\cH[(\cD-c^2) \widetilde{\widetilde{\gamma}}_c^{\rm HF}] +\frac{1}{2c^2}\sum_{1\leq m,n\leq q}  \left|\left<u_m^{\rm HF}, H_0 u_n^{\rm HF}\right>_\HL\right|^2.
\end{align*}
Next by \eqref{eq:gammaHF''} again,
\begin{align*}
    \Tr_\cH(V\gamma^{\rm HF}_*) &= \Tr_\cH(V\cS \gamma^{\rm HF}_* \cS^* )-\frac{1}{4c^2}\sum_{n=1}^q\left<\cL u_n^{\rm HF}, V \cL u_n^{\rm HF}\right>_\HL\\
    &= \Tr_\cH(V\widetilde{\widetilde{\gamma}}_c^{\rm HF})-\frac{1}{4c^2} \Tr_\cH [V \cL \gamma_*^{\rm HF}\cL] \\
    &\quad+\frac{1}{2c^2}\sum_{1\leq m,n\leq q}\left<u_{m}^{\rm HF},H_0 u_n^{\rm HF}\right>_\HL\left< u_n^{\rm HF}, V  u_m^{\rm HF}\right>_{\HL}+\cO(c^{-4}).
\end{align*}
Analogously, we also have
\begin{align*}
  \MoveEqLeft  \frac{1}{2}\Tr_\cH[W_{1,\gamma^{\rm HF}_*}\gamma^{\rm HF}_*]\\
    &=\frac{1}{2}\Tr_\cH[W_{1,\cS \gamma_*^{\rm HF}\cS^* }\cS \gamma_*^{\rm HF}\cS^* ]-\frac{1}{4c^2}\Tr_{\cH}[W_{1,\gamma_*^{\rm HF}}\cL \gamma_*^{\rm HF}\cL]\\
    &\quad-\frac{1}{32 c^4}\Tr[W_{1,\cL \gamma_*^{\rm HF}\cL}\cL \gamma_*^{\rm HF}\cL]\\
    &=\frac{1}{2}\Tr_\cH[W_{1,\widetilde{\widetilde{\gamma}}_c^{\rm HF}}\widetilde{\widetilde{\gamma}}_c^{\rm HF} ]-\frac{1}{4c^2} \Tr_\cH [W_{1,\gamma^{\rm HF}_*} \cL \gamma_*^{\rm HF}\cL] \\
    &\quad + \frac{1}{2c^2}\sum_{1\leq m,n\leq q}\left<u_{m}^{\rm HF},H_0 u_n^{\rm HF}\right>_\HL\left< u_n^{\rm HF}, W_{1,\gamma_*^{\rm HF}}  u_m^{\rm HF}\right>_{\HL}  +\cO(c^{-4}),
\end{align*}
since $u_n^{\rm HF}\in H^2$. Finally, concerning the term $W_{2,\cdot}$, by \eqref{eq:W-K1-K2}, Lemma \ref{lem:W2-estimate} and \ref{lem:gammaHF'-estimate} we have
\begin{align*}
 \MoveEqLeft \frac{1}{2}\Tr_\cH[W_{2,\gamma_*^{\rm HF}}\gamma_*^{\rm HF}]\\
    &=\frac{1}{2}\Tr_\cH[W_{2,\cS \gamma_*^{\rm HF}\cS^* }\cS \gamma_*^{\rm HF}\cS^* ] -\frac{1}{4c^2}\Tr_\cH[W_{2, \gamma_*^{\rm HF}\cL} \cL \gamma_*^{\rm HF}]\\
    &\quad-\frac{1}{32c^4}\Tr_\cH[W_{2, \cL\gamma_*^{\rm HF}\cL} \cL \gamma_*^{\rm HF}\cL]\\
    &=\frac{1}{2}\Tr_\cH[W_{2,\cS \gamma_*^{\rm HF}\cS^*}\cS \gamma_*^{\rm HF}\cS^* ] -\frac{1}{4c^2}\Tr_\cH[W_{2, \gamma_*^{\rm HF}\cL} \cL \gamma_*^{\rm HF}]+\cO(c^{-4})\\
    &= \frac{1}{2}\Tr_\cH[W_{2,\widetilde{\widetilde{\gamma}}_c^{\rm HF}}\widetilde{\widetilde{\gamma}}_c^{\rm HF} ] -\frac{1}{4c^2}\Tr_\cH[W_{2, \gamma_*^{\rm HF}\cL} \cL \gamma_*^{\rm HF}]\\
    &\quad +\frac{1}{2c^2}\sum_{1\leq m,n\leq q}\left<u_{m}^{\rm HF},H_0 u_n^{\rm HF}\right>_\HL\left< u_n^{\rm HF}, W_{2,\gamma_*^{\rm HF}}  u_m^{\rm HF}\right>_{\HL} +\cO(c^{-4}).
\end{align*}
Gathering these estimates together, we conclude that
\begin{align*}
     \cE_{c}(\widetilde{\widetilde{\gamma}}_c^{\rm HF})&=\cE^{\rm HF}(\gamma_*^{\rm HF})- \frac{1}{2c^2}\sum_{1\leq m, n\leq q}\left<u_{m}^{\rm HF},H_0 u_n^{\rm HF}\right>_\HL\left< u_n, H_{0,\gamma_*^{\rm HF}}  u_m\right>_{\HL}\\
     &\quad+\frac{1}{4c^2}\left(\Tr_\cH[ (-V+W_{1,\gamma_*^{\rm HF}})\cL\gamma_*^{\rm HF}\cL]-\Tr_\cH[W_{2, \gamma_*^{\rm HF}\cL} \cL \gamma_*^{\rm HF}]\right)+\cO(c^{-4})\\
     &= \cE^{\rm HF}(\gamma_*^{\rm HF}) -\frac{1}{4c^2}\sum_{n=1}^q \lambda_{n}^{\rm HF}\left<\cL u_{n}^{\rm HF},\cL  u_n^{\rm HF}\right>_\HL\\
     &\quad+\frac{1}{4c^2}\left(\Tr_\cH[ (-V+W_{1,\gamma_*^{\rm HF}})\cL\gamma_*^{\rm HF}\cL]-\Tr_\cH[W_{2, \gamma_*^{\rm HF}\cL} \cL \gamma_*^{\rm HF}]\right)+\cO(c^{-4})\\
     &= \cE^{\rm HF}(\gamma_*^{\rm HF}) +\cE^{(2)}_c(\gamma_*^{\rm HF})+\cO(c^{-4}).
\end{align*}
Here according to Remark \ref{rem:non-unfill}, we used that $u_1^{\rm HF},\cdots, u_q^{\rm HF} $ are orthonormal eigenfunctions of $H_{0,\gamma_*^{\rm HF}}$ and the fact that
\begin{align*}
    \sum_{n=1}^q \lambda_{n}^{\rm HF}\left<\cL u_{n}^{\rm HF},\cL  u_n^{\rm HF}\right>_\HL= \sum_{n=1}^q \left<\cL^2 u_{n}^{\rm HF}, H_{0,\gamma_*^{\rm HF}} u_n^{\rm HF}\right>_\HL=\Tr_{\cH}[H_{0,\gamma_*^{\rm HF}}\gamma_*^{\rm HF} \cL^2].
\end{align*}
This ends the proof.
\end{proof}

\medskip

Now we consider \eqref{EDF-EHF} with $\gamma=\widetilde{\gamma}_c^{\rm HF}$, and we split the proof into the following three parts.
\subsubsection{Kinetic term}\label{sec:4.3.2}
Concerning the kinetic term, by \eqref{eq:D-Delta}, H\"older's inequality \eqref{eq:Holder}, Corollary \ref{cor:lambda-non-gamma'} and as $[\cD,\Lambda^\pm_c]=0$, we infer
\begin{align}\label{eq:D-c2-gamma'}
  \MoveEqLeft  \left|\Tr_\cH[(\cD-c^2)\Lambda^+_c \widetilde{\gamma}_c^{\rm HF}\Lambda^+_c] -\Tr_\cH[(\cD-c^2) \widetilde{\gamma}_c^{\rm HF}]\right|\notag\\
  &= \left|\Tr_\cH[(\cD-c^2)\Lambda^-_c \widetilde{\gamma}_c^{\rm HF}]\right|=\left|\Tr_\cH[\Lambda^-_c \widetilde{\gamma}_c^{\rm HF} (\cD-c^2)]\right|\notag\\
  &= \left|\Tr_\cH[\cK_\rL \Lambda^-_c \widetilde{\gamma}_c^{\rm HF}\cK_\rL H_0]\right|=\frac{1}{2} \left|\Tr_\cH[(-\Delta)^{-1/2}\cK_\rL \Lambda^-_c \widetilde{\gamma}_c^{\rm HF}\cK_\rL (-\Delta)^{3/2}]\right|\notag\\
  &\leq \frac{1}{2}\|(-\Delta)^{-1/2}\cK_\rL \Lambda_c^- \widetilde{\gamma}_c^{\rm HF}\Lambda_c^-\cK_\rL(-\Delta)^{-1/2}\|_{\mathfrak{S}_1}^{1/2}\|\cK_\rL \widetilde{\gamma}_c^{\rm HF}\cK_\rL\|_{X^6}^{1/2}=\cO(c^{-4}).
\end{align}

\subsubsection{Potential between electrons and nuclei}\label{sec:4.3.3}
Concerning the potential between electrons and nuclei, we have the following.
\begin{lemma}\label{lem:4.4}
Under Assumption \ref{ass:V}, for $c$ large enough we have
    \begin{align*}
         \Tr_{\cH}[V\widetilde{\gamma}_c^{\rm HF}]-  \Tr_{\cH}[V\Lambda^+_c\widetilde{\gamma}_c^{\rm HF} \Lambda^+_c]=\cO(c^{-4})
    \end{align*}
\end{lemma}
\begin{proof}
We have
\begin{align}\label{eq:V1'}
    \Tr_{\cH}[V\widetilde{\gamma}_c^{\rm HF}]-  \Tr_{\cH}[V\Lambda^+_c\widetilde{\gamma}_c^{\rm HF}\Lambda^+_c]=\Tr_{\cH}[V \Lambda^-_c \widetilde{\gamma}_c^{\rm HF}]+ \Tr_{\cH}[V\Lambda^+_c\widetilde{\gamma}_c^{\rm HF} \Lambda^-_c].
\end{align}
We first consider $\Tr_{\cH}[V \Lambda^-_c \widetilde{\gamma}_c^{\rm HF}]$. By \eqref{eq:V-K1-K2},
\begin{align*}
    \Tr_{\cH}[V \Lambda^-_c \widetilde{\gamma}_c^{\rm HF}]=\Tr_{\cH}[  (-\Delta)^{-1/2}\cK_\rL\Lambda^-_c \widetilde{\gamma}_c^{\rm HF}\cK_\rL V (-\Delta)^{1/2}]+\Tr_{\cH}[ \cK_\rS\Lambda^-_c \widetilde{\gamma}_c^{\rm HF}\cK_\rS V].
\end{align*}
Then by H\"older's inequality \eqref{eq:Holder}, Corollary \ref{cor:lambda-non-gamma'} and Lemma \ref{lem:gammaHF'-estimate},
\begin{align*}
    \left| \Tr_{\cH}[V \Lambda^-_c \widetilde{\gamma}_c^{\rm HF}]\right|&\lesssim \|(-\Delta)^{-1/2}\cK_\rL\Lambda^-_c \widetilde{\gamma}_c^{\rm HF} \Lambda^-_c\cK_\rL (-\Delta)^{-1/2}\|_{\mathfrak{S}_1}^{1/2}\|\cK_\rL\widetilde{\gamma}_c^{\rm HF}\cK_\rL\|_{X^4}^{1/2}\\
    &\quad+\|\cK_\rS\Lambda^-_c \widetilde{\gamma}_c^{\rm HF} \Lambda^-_c\cK_\rS\|_{\mathfrak{S}_1}^{1/2}\|\cK_\rS\widetilde{\gamma}_c^{\rm HF}\cK_\rS\|_{X^2}^{1/2}=\cO(c^{-4}).
\end{align*}
Here in the first inequality, by Assumption \ref{ass:V} and Hardy's inequality,
\begin{align}\label{eq:nabla-V}
    \|(-\Delta)^{1/2} V u\|_{\cH}=\|\nabla (Vu)\|_{\cH}\leq \|(\nabla V) u\|_\cH+\|V \nabla u\|_{\cH}\lesssim \|(1-\Delta)u\|_{\cH}.
\end{align}

\medskip

Next, for the term $\Tr_{\cH}[V\Lambda^+_c\widetilde{\gamma}_c^{\rm HF} \Lambda^-_c]$, we have
\begin{align*}
    \Tr_{\cH}[V \Lambda^+_c \widetilde{\gamma}_c^{\rm HF}\Lambda^-_c]&\lesssim \|(-\Delta)^{-1/2}\cK_\rL\Lambda^-_c \widetilde{\gamma}_c^{\rm HF} \Lambda^-_c\cK_\rL (-\Delta)^{-1/2}\|_{\mathfrak{S}_1}^{1/2}\|\cK_\rL \Lambda^+_c\widetilde{\gamma}_c^{\rm HF}\Lambda^+_c \cK_\rL\|_{X^4}^{1/2}\\
    &\quad+\|\cK_\rS\Lambda^-_c \widetilde{\gamma}_c^{\rm HF} \Lambda^-_c\cK_\rS\|_{\mathfrak{S}_1}^{1/2}\|\cK_\rS\Lambda^+_c \widetilde{\gamma}_c^{\rm HF}\Lambda^+_c \cK_\rS\|_{X^2}^{1/2}=\cO(c^{-4}).
\end{align*}
Finally, using \eqref{eq:V1'}, we get
\begin{align*}
         \Tr_{\cH}[V\widetilde{\gamma}_c^{\rm HF}]-  \Tr_{\cH}[V\Lambda^+_c\widetilde{\gamma}_c^{\rm HF} \Lambda^+_c]=\cO(c^{-4}).
\end{align*}
This ends the proof.
\end{proof}

\subsubsection{Potential between electrons and electrons}\label{sec:4.3.4}
We now consider the estimates associated with $W_{\widetilde{\gamma}_c^{\rm HF}}$. Analogous to Section \ref{sec:3.1.3}, we have
\begin{align*}
\MoveEqLeft    \Tr_{\cH}[W_{\widetilde{\gamma}_c^{\rm HF}}\widetilde{\gamma}_c^{\rm HF}]-\Tr_{\cH}[W_{\Lambda^+_c\widetilde{\gamma}_c^{\rm HF} \Lambda^+_c}\Lambda^+_c\widetilde{\gamma}_c^{\rm HF} \Lambda^+_c]\\
    &= \Tr_{\cH}[W_{\widetilde{\gamma}_c^{\rm HF}}(\widetilde{\gamma}_c^{\rm HF}-\Lambda^+_c\widetilde{\gamma}_c^{\rm HF} \Lambda^+_c)]+\Tr_{\cH}[W_{\widetilde{\gamma}_c^{\rm HF}-\Lambda^+_c\widetilde{\gamma}_c^{\rm HF} \Lambda^+_c}\Lambda^+_c\widetilde{\gamma}_c^{\rm HF} \Lambda^+_c].
\end{align*}
We will study these two terms separately.

Concerning the first term on the right-hand side, we have
\begin{lemma}\label{lem:4.5}
    Under Assumption \ref{ass:V}, for $c$ large enough,
    \begin{align*}
       \Tr_{\cH}[W_{\widetilde{\gamma}_c^{\rm HF}}(\widetilde{\gamma}_c^{\rm HF}-\Lambda^+_c\widetilde{\gamma}_c^{\rm HF} \Lambda^+_c)]=\cO(c^{-4}).
    \end{align*}
\end{lemma}
\begin{proof}
Observe that
\begin{align*}
    \MoveEqLeft \Tr_{\cH}[W_{\widetilde{\gamma}_c^{\rm HF}}(\widetilde{\gamma}_c^{\rm HF}-\Lambda^+_c\widetilde{\gamma}_c^{\rm HF} \Lambda^+_c)]\\
     &= \Tr_{\cH}[W_{1,\widetilde{\gamma}_c^{\rm HF}}(\widetilde{\gamma}_c^{\rm HF}-\Lambda^+_c\widetilde{\gamma}_c^{\rm HF} \Lambda^+_c)]- \Tr_{\cH}[W_{2,\widetilde{\gamma}_c^{\rm HF}}(\widetilde{\gamma}_c^{\rm HF}-\Lambda^+_c\widetilde{\gamma}_c^{\rm HF} \Lambda^+_c)].
\end{align*}
We first consider the term associated with $W_{1,\bullet}$.

By \eqref{eq:W1-estimate} and Lemma \ref{lem:gammaHF'-estimate},
\begin{align}\label{eq:4.43'}
    \|\nabla (W_{1,\widetilde{\gamma}_c^{\rm HF}} u)\|_{\cH}\lesssim \|\widetilde{\gamma}_c^{\rm HF}\|_{X^2}\|u\|_{\cH}\lesssim\|(1-\Delta) u\|_{\cH}.
\end{align}
Taking $V= W_{1,\widetilde{\gamma}_c^{\rm HF}}$, from Lemma \ref{lem:4.4} we can conclude that
\begin{align*}
   \Tr_{\cH}[W_{1,\widetilde{\gamma}_c^{\rm HF}}(\widetilde{\gamma}_c^{\rm HF}-\Lambda^+_c\widetilde{\gamma}_c^{\rm HF} \Lambda^+_c)]=\cO(c^{-4}).
\end{align*}

Concerning the term associated with $W_{2,\bullet}$, by \eqref{eq:Wpotential-decom},
\begin{align*}
   \MoveEqLeft \Tr_{\cH}[W_{2,\widetilde{\gamma}_c^{\rm HF}}(\widetilde{\gamma}_c^{\rm HF}-\Lambda^+_c\widetilde{\gamma}_c^{\rm HF} \Lambda^+_c)]\\
   &=  \Tr_{\cH}[W_{2,\widetilde{\gamma}_c^{\rm HF}}\Lambda^-_c\widetilde{\gamma}_c^{\rm HF}]+\Tr_{\cH}[W_{2,\widetilde{\gamma}_c^{\rm HF}}\Lambda^+_c\widetilde{\gamma}_c^{\rm HF}\Lambda^-_c]\\
    &=\sum_{j,j'\in \{\rL,\rS\}}\left( \Tr_{\cH}[W_{2,\cK_j\widetilde{\gamma}_c^{\rm HF}\cK_{j'}}\cK_{j'}\Lambda^-_c\widetilde{\gamma}_c^{\rm HF}\cK_j]+\Tr_{\cH}[W_{2,\cK_j\widetilde{\gamma}_c^{\rm HF}\cK_{j'}} \cK_{j'}\Lambda^+_c\widetilde{\gamma}_c^{\rm HF}\Lambda^-_c \cK_j]\right).
\end{align*}
Thus, for $j\in \{\rS,\rL\}$,
\begin{align*}
 \MoveEqLeft  \left| \Tr_{\cH}[W_{2,\cK_j\widetilde{\gamma}_c^{\rm HF}\cK_{\rL}}\cK_{\rL}\Lambda^-_c\widetilde{\gamma}_c^{\rm HF}\cK_j]\right|\\
 &= \left|\Tr_{\cH}[(-\Delta)^{-1/2}\cK_{\rL}\Lambda^-_c\widetilde{\gamma}_c^{\rm HF}\cK_jW_{2,\cK_j\widetilde{\gamma}_c^{\rm HF}\cK_{\rL}} (-\Delta)^{1/2}]\right|\\
   &\lesssim \|\cK_j\widetilde{\gamma}_c^{\rm HF}\cK_j\|_{X^2}^{1/2}\|\cK_{\rL}\widetilde{\gamma}_c^{\rm HF}\cK_{\rL}\|_{X^2}^{1/2}  \|(-\Delta)^{-1/2}\cK_{\rL}\Lambda^-_c\widetilde{\gamma}_c^{\rm HF}\cK_j  \nabla \|_{\mathfrak{S}_1}\\
   &\leq \|\cK_j\widetilde{\gamma}_c^{\rm HF}\cK_j\|_{X^2}^{1/2}\| \nabla \cK_j\widetilde{\gamma}_c^{\rm HF}\cK_j \nabla \|_{\mathfrak{S}_1}^{1/2}\\
   &\quad \times \|\cK_{\rL}\widetilde{\gamma}_c^{\rm HF}\cK_{\rL}\|_{X^2}^{1/2} \|(-\Delta)^{-1/2}\cK_{\rL}\Lambda^-_c \widetilde{\gamma}_c^{\rm HF}\Lambda^-_c\cK_{\rL}(-\Delta)^{-1/2} \|_{\mathfrak{S}_1}^{1/2}=\cO(c^{-4}),
\end{align*}
where in the first inequality we used \eqref{eq:5.19}, in the second inequality we used  H\"older's inequality \eqref{eq:Holder}, and the final estimate is obtained by using Corollary \ref{cor:lambda-non-gamma'} and Lemma \ref{lem:gammaHF'-estimate}. Analogously,
\begin{align*}
 \MoveEqLeft  \left| \Tr_{\cH}[W_{2,\cK_j\widetilde{\gamma}_c^{\rm HF}\cK_{\rS}}\cK_{\rS}\Lambda^-_c\widetilde{\gamma}_c^{\rm HF}\cK_j]\right|\\
 &\leq \|\cK_{\rS}\Lambda^-_c\widetilde{\gamma}_c^{\rm HF}\cK_j W_{2,\cK_j\widetilde{\gamma}_c^{\rm HF}\cK_{\rS}}\|_{\mathfrak{S}_1}\\
   &\lesssim \|\cK_j\widetilde{\gamma}_c^{\rm HF}\cK_j\|_{\mathfrak{S}_1}^{1/2}\|\cK_{\rS}\widetilde{\gamma}_c^{\rm HF}\cK_{\rS}\|_{\mathfrak{S}_1}^{1/2}  \|\cK_{\rS}\Lambda^-_c\widetilde{\gamma}_c^{\rm HF}\cK_j\nabla  \|_{\mathfrak{S}_1}\\
   &\leq \|\cK_j\widetilde{\gamma}_c^{\rm HF}\cK_j\|_{\mathfrak{S}_1}^{1/2}\| \nabla \cK_j\widetilde{\gamma}_c^{\rm HF}\cK_j \nabla \|_{\mathfrak{S}_1}^{1/2} \\
   &\quad \times\|\cK_{\rS}\widetilde{\gamma}_c^{\rm HF}\cK_{\rS}\|_{\mathfrak{S}_1}^{1/2} \|\cK_{\rS}\Lambda^-_c \widetilde{\gamma}_c^{\rm HF}\Lambda^-_c\cK_{\rS} \|_{\mathfrak{S}_1}^{1/2}=\cO(c^{-4}).
\end{align*}
Thus
\begin{align}\label{eq:3.34}
    \sum_{j,j'\in \{\rL,\rS\}} \Tr_{\cH}[W_{2,\cK_j\widetilde{\gamma}_c^{\rm HF}\cK_{j'}}\cK_{j'}\Lambda^-_c\widetilde{\gamma}_c^{\rm HF}\cK_j] =\cO(c^{-4}).
\end{align}

Next, we study the term 
\begin{align*}
    \sum_{j,j'\in \{\rL,\rS\}}\Tr_{\cH}[W_{2,\cK_j\widetilde{\gamma}_c^{\rm HF}\cK_{j'}} \cK_{j'}\Lambda^+_c\widetilde{\gamma}_c^{\rm HF}\Lambda^-_c \cK_j].
\end{align*}
whose proof is essentially the same as for \eqref{eq:3.34}. For $j'\in \{\rS,\rL\}$,
\begin{align*}
  \MoveEqLeft  \left|\Tr_{\cH}[W_{2,\cK_\rL\widetilde{\gamma}_c^{\rm HF}\cK_{j'}} \cK_{j'}\Lambda^+_c\widetilde{\gamma}_c^{\rm HF}\Lambda^-_c \cK_\rL]\right|\\
  &\lesssim \|\cK_{j'}\widetilde{\gamma}_c^{\rm HF}\cK_{j'}\|_{X^2}^{1/2}\|\cK_{\rL}\widetilde{\gamma}_c^{\rm HF}\cK_{\rL}\|_{X^2}^{1/2}  \|(-\Delta)^{-1/2}\cK_\rL \Lambda^-_c\widetilde{\gamma}_c^{\rm HF}\Lambda^+_c\cK_{j'}  \nabla \|_{\mathfrak{S}_1}\\
   &\leq \|\cK_{j'}\widetilde{\gamma}_c^{\rm HF}\cK_{j'}\|_{X^2}^{1/2}\| \nabla \cK_{j'}\Lambda^+_c\widetilde{\gamma}_c^{\rm HF}\Lambda^+_c\cK_{j'} \nabla \|_{\mathfrak{S}_1}^{1/2}\\
   &\quad \times \|\cK_{\rL}\widetilde{\gamma}_c^{\rm HF}\cK_{\rL}\|_{X^2}^{1/2} \|(-\Delta)^{-1/2}\cK_{\rL}\Lambda^-_c \widetilde{\gamma}_c^{\rm HF}\Lambda^-_c\cK_{\rL}(-\Delta)^{-1/2} \|_{\mathfrak{S}_1}^{1/2}=\cO(c^{-4}),
\end{align*}
and
\begin{align*}
 \MoveEqLeft  \left| \Tr_{\cH}[W_{2,\cK_\rS\widetilde{\gamma}_c^{\rm HF}\cK_{j'}} \cK_{j'}\Lambda^+_c\widetilde{\gamma}_c^{\rm HF}\Lambda^-_c \cK_\rS]\right|\\
   &\lesssim \|\cK_{j'}\widetilde{\gamma}_c^{\rm HF}\cK_{j'}\|_{\mathfrak{S}_1}^{1/2}\| \nabla \cK_{j'}\Lambda^+_c \widetilde{\gamma}_c^{\rm HF}\Lambda^+_c \cK_{j'} \nabla \|_{\mathfrak{S}_1}^{1/2} \\
   &\quad\times \|\cK_{\rS}\widetilde{\gamma}_c^{\rm HF}\cK_{\rS}\|_{\mathfrak{S}_1}^{1/2} \|\cK_{\rS}\Lambda^-_c \widetilde{\gamma}_c^{\rm HF}\Lambda^-_c\cK_{\rS} \|_{\mathfrak{S}_1}^{1/2}=\cO(c^{-4}).
\end{align*}
Thus,
\begin{align*}
     \sum_{j,j'\in \{\rL,\rS\}}\Tr_{\cH}[W_{2,\cK_j\widetilde{\gamma}_c^{\rm HF}\cK_{j'}} \cK_{j'}\Lambda^+_c\widetilde{\gamma}_c^{\rm HF}\Lambda^-_c \cK_j]=\cO(c^{-4}).
\end{align*}
This and \eqref{eq:3.34} show
\begin{align*}
    \Tr_{\cH}[W_{2,\widetilde{\gamma}_c^{\rm HF}}(\widetilde{\gamma}_c^{\rm HF}-\Lambda^+_c\widetilde{\gamma}_c^{\rm HF} \Lambda^+_c)]=\cO(c^{-4}).
\end{align*}
Finally, we get
 \begin{align*}
       \Tr_{\cH}[W_{\widetilde{\gamma}_c^{\rm HF}}(\widetilde{\gamma}_c^{\rm HF}-\Lambda^+_c\widetilde{\gamma}_c^{\rm HF} \Lambda^+_c)]=\cO(c^{-4}).
    \end{align*}
The proof is completed.
\end{proof}

Concerning the second term on the right-hand side, we have
\begin{lemma}
    Under Assumption \ref{ass:V}, for $c$ large enough,
    \begin{align*}
       \Tr_{\cH}[W_{\widetilde{\gamma}_c^{\rm HF}-\Lambda^+_c\widetilde{\gamma}_c^{\rm HF} \Lambda^+_c}\Lambda^+_c\widetilde{\gamma}_c^{\rm HF} \Lambda^+_c]=\cO(c^{-4}).
    \end{align*}
\end{lemma}
\begin{proof}
    The proof is essentially the same as for Lemma \ref{lem:4.5}. Note that
    \begin{align*}
        \left|  \Tr_{\cH}[W_{\widetilde{\gamma}_c^{\rm HF}-\Lambda^+_c\widetilde{\gamma}_c^{\rm HF} \Lambda^+_c}\Lambda^+_c\widetilde{\gamma}_c^{\rm HF} \Lambda^+_c]\right|=\left|  \Tr_{\cH}[W_{\Lambda^+_c\widetilde{\gamma}_c^{\rm HF} \Lambda^+_c}\left(\widetilde{\gamma}_c^{\rm HF}-\Lambda^+_c\widetilde{\gamma}_c^{\rm HF} \Lambda^+_c\right)]\right|.
    \end{align*}
We just need to replace $\widetilde{\gamma}_c^{\rm HF}$ by $\Lambda^+_c \widetilde{\gamma}_c^{\rm HF} \Lambda^+_c$ in $W_{1,\bullet}$ and $W_{2,\bullet}$ in the proof of Lemma \ref{lem:4.5}. Then this lemma follows from Corollary \ref{cor:lambda-non-gamma'} and Lemma \ref{lem:gammaHF'-estimate}.
\end{proof}

\subsubsection{End of the proof of \texorpdfstring{\eqref{eq:E_c-EHF-c2}}{}} 

From Lemma \ref{lem:gammaHF'-projection}, we know $\widetilde{\gamma}_c^{\rm HF}=\cS\widetilde{\gamma}_c^{\rm HF}\cS^*\in \Gamma_q$. Then from \eqref{eq:3.21},  Lemma \ref{lem:W2-1} and Lemma \ref{lem:W2-2} we know that
\begin{align*}
\left| \Tr_{\cH}[W_{\widetilde{\gamma}^{\rm HF}_{c}}\widetilde{\gamma}^{\rm HF}_{c}]-\Tr_{\cH}[W_{\Lambda^+_c\widetilde{\gamma}^{\rm HF}_{c} \Lambda^+_c}\Lambda^+_c\widetilde{\gamma}^{\rm HF}_{c} \Lambda^+_c]\right|=\cO(c^{-4}).
\end{align*}
Thus \eqref{EDF-EHF}, \eqref{eq:D-c2-gamma'}, Lemma \ref{lem:4.4} and above estimate imply that
\begin{align*}
    |\cE_c(\Lambda^+\widetilde{\gamma}^{\rm HF}_{c}\Lambda^+)- \cE^{\rm HF}(\widetilde{\gamma}^{\rm HF}_{c})|=\cO(c^{-4}).
\end{align*}
This and Lemma \ref{lem:3.13} give \eqref{eq:E_c-EHF-c2}, thus the proof of Theorem \ref{th:3.1} is completed.

\section{From free picture to the DF ground-state energy}\label{sec:6}
In this section, we pass from fixed free picture $\Lambda^+_c$ to DF energy in the set $\Gamma_{q}^+$. More precisely, we are trying to prove the following.
\begin{theorem}[From free picture to the DF ground-state energy]\label{th:4.1}
 Let $\gamma_*^{\rm HF}, \widetilde{\gamma}_c^{\rm HF}$ be given as in Theorem \ref{th:3.1}. Then under Assumption \ref{ass:c}, we have
    \begin{align}\label{eq:E_cq-E_c-lambda}
      E_{c,q}\leq   \cE_c(\Lambda^+_c \gamma_*^{\rm HF}\Lambda^+_c)+\cO(c^{-2}).
    \end{align}
In addition, under Assumption \ref{ass:V}, for $c$ large enough, we also have
\begin{align}\label{eq:E_cq-E_c-lambda-c2}
      E_{c,q}\leq  \cE_c(\Lambda^+_c \widetilde{\gamma}_c^{\rm HF}\Lambda^+_c)+\cO(c^{-4}).
\end{align}
\end{theorem}
To prove Theorem \ref{th:4.1}, in Section \ref{sec:5.1} we recall a retraction mapping $\theta$ used to study the DF ground-state energy, and we summarize some useful arguments in \cite{sere2023new} and \cite{meng2024rigorous}. In particular, Lemma \ref{lem:error-bound} below plays an essential role in this part. To use Lemma \ref{lem:error-bound}, in  Section \ref{sec:5.2}, we study the difference between the  projector $P^+_{c,\gamma}$ and $\Lambda^+_c$. Note that
\begin{align*}
    P^+_{c,\gamma}- \Lambda^+_c= P^+_{c,\gamma}-P^+_{c,0} +P^+_{c,0} \Lambda^+_c.
\end{align*}
Thus in Section \ref{sec:5.2} we will study above terms separately. Finally, in Section \ref{sec:5.3}, we prove Theorem \ref{th:4.1} by using Lemma \ref{lem:error-bound}.

\subsection{New DF functional}\label{sec:5.1}
To prove the existence of the DF minimizers (i.e., Theorem \ref{th:min-DF}), the following retraction mapping $\theta(\gamma)$ is introduced in \cite{sere2023new}:
\begin{align}\label{eq:theta}
    \theta_c(\gamma):=\lim_{n\rightarrow +\infty} T_c^n(\gamma)
\end{align}
with
\[
T_c^n(\gamma)=T_c(T_c^{n-1}(\gamma)),\quad T_c(\gamma)=P^+_{c,\gamma}\gamma P^+_{c,\gamma}.
\]
To prove Theorem \ref{th:4.1}, we will also use this mapping.

The existence of the retraction $\theta$ is based on the following.
\begin{lemma}[Existence of the retraction {\cite{sere2023new,meng2024rigorous}}]\label{lem:retra}
Assume that $\kappa_c:=2c^{-1}(q+z) <1$ and $a_c:=\frac{\pi}{4c\sqrt{(1-\kappa_c)\lambda_{0,c}}}$. Given $R<\frac{1}{2a_c}$, let $A_c:=\max(\frac{1}{1-2a_cR},\frac{2+a_cq}{2})$, and let
\[
\mathcal{U}_{c,R}:=\left\{\gamma\in\Gamma_{q}; \frac{1}{c}\|\gamma|\cD|^{1/2}\|_{\mathfrak{S}_1}+\frac{A_c}{c^2}\|T_c(\gamma)-\gamma\|_{X_c}< R\right\}.
\]
Then, $T_c$ maps $\mathcal{U}_{c,R}$ into $\mathcal{U}_{c,R}$, and for any $\gamma\in\mathcal{U}_{c,R}$ the sequence $(T_c^n(\gamma))_{n\geq 0}$ converges to a limit $\theta_c(\gamma)\in\Gamma_{q}^+$. Moreover for any $\gamma\in \mathcal{U}_{c,R}$,
\begin{equation}\label{eq:retra}
    \begin{aligned}
    \|T_c^{n+1}(\gamma)-T_c^{n}(\gamma)\|_{X_c}&\leq L_c\|T_c^n(\gamma)-T_c^{n-1}(\gamma)\|_{X_c},\\
    \|\theta_c(\gamma)-T_c^n(\gamma)\|_{X_c}&\leq\frac{L_c^n}{1-L_c}\|T_c(\gamma)-\gamma\|_{X_c}
\end{aligned}
\end{equation}
with $L_c:=2a_cR$.
\end{lemma}

Therefore, one can define a new DF functional:
\begin{definition}[New DF functional]
Let $\kappa_c,\;a_c$, $R$ and $\mathcal{U}_{c,R}$ be given as in Lemma \ref{lem:retra}. For any $\gamma\in\mathcal{U}_{c,R}$, a new DF functional of $\gamma$ is defined by
\begin{align}\label{eq:def}
E_c(\gamma)=\cE_c(\theta_c(\gamma)).
\end{align}
\end{definition}
Then under Assumption \ref{ass:c}, if there exists some $R<\frac{1}{2a_c}$  such that
\begin{align}\label{eq:condition-R}
    \gamma_*^c,\; \Lambda^+_c\gamma_{*}^{\rm HF}\Lambda^+_c, \;\Lambda^+_c\widetilde{\gamma}_c^{\rm HF}\Lambda^+_c\in \mathcal{U}_{c,R},
\end{align}
then
\begin{align*}
   \gamma_*^c= \theta( \gamma_*^c),\qquad \theta(\Lambda^+_c\gamma_{*}^{\rm HF}\Lambda^+_c),\qquad \theta(\Lambda^+_c\widetilde{\gamma}_c^{\rm HF}\Lambda^+_c)
\end{align*}
are well-defined and are situated in $\Gamma_q^+$. Thus, according to the definition of the DF ground-state energy \eqref{eq:min-DF},
\begin{align}\label{eq:EDF-min-E-retra-HF}
    E_{c,q}\leq E_{c}(\Lambda^+_c\gamma_{*}^{\rm HF}\Lambda^+_c),\qquad  E_{c,q}\leq E_{c}(\Lambda^+_c\widetilde{\gamma}_c^{\rm HF}\Lambda^+_c).
\end{align}

Thus to prove Theorem \ref{th:4.1}, it remains to find some $R <\frac{1}{2a_c}$ such that  \eqref{eq:condition-R} holds, and that
\begin{align}\label{eq:error-theta-Lambda}
    |E_{c}(\Lambda^+_c\gamma_{*}^{\rm HF}\Lambda^+_c)-\cE_c(\Lambda^+_c\gamma_{*}^{\rm HF}\Lambda^+_c)|=\mathcal{O}(c^{-2})
\end{align}
or under Assumption \ref{ass:V},
\begin{align}\label{eq:error-theta-Lambda'}
     |E_{c}(\Lambda^+_c\widetilde{\gamma}_c^{\rm HF}\Lambda^+_c)-\cE_c(\Lambda^+_c\widetilde{\gamma}_c^{\rm HF}\Lambda^+_c)|=\mathcal{O}(c^{-4}).
\end{align}

To prove \eqref{eq:error-theta-Lambda} and \eqref{eq:error-theta-Lambda'}, the following lemma will be used.
 \begin{lemma}\cite[Lemma 5.1]{meng2024rigorous}\label{lem:error-bound}
 Let $R,z\in\mathbb{R}^+$ and $q\in\mathbb{N}^+$ be fixed. Assume that $\kappa_c<1$ and $L_c<1$ as in Lemma \ref{lem:retra}. Let $C_{\kappa_c,L_c}:=\frac{5\pi^2}{4(1-\kappa_c)^2\lambda_{0,c}^{3/2}(1-L_c)^2}$. Then for any $\gamma\in\mathcal{U}_{c,R}$, 
\begin{align}
    |E_c(\gamma)-\mathcal{E}_c(\gamma)|\leq C_{\kappa_c,L_c}(3c^{-1}R+3c^{-1}q+1)\frac{1}{c^3}\|T_c(\gamma)-\gamma\|_{X_c}^2+3\|P^-_{c,\gamma}\gamma P^-_{c,\gamma}\|_{X_c}.
\end{align}
 \end{lemma}

\subsection{Estimates on the projections}\label{sec:5.2}
To prove \eqref{eq:condition-R}-\eqref{eq:error-theta-Lambda'}, from Lemma \ref{lem:error-bound}, we need to study $P^+_{c,0}-\Lambda^+_c$ and $P^+_{c,\gamma}-P^+_{c,0}$. 

The following lemma will be used to prove \eqref{eq:condition-R} and \eqref{eq:error-theta-Lambda}.
\begin{lemma}\label{lem:Lambda-Pgamma}
    For any $\gamma\in \Gamma_{q}$ and $\kappa_c<1$, we have
    \begin{align}\label{eq:P0-Lambda}
        \||\cD|^{1/2}(P^+_{c,0}-\Lambda^+_c)(-\Delta)^{-1/2}\|_{\cB(\cH)}\leq \frac{z}{(1-\kappa_c)^{1/2}c}
    \end{align}
and
    \begin{align}\label{eq:P-P}
        \||\cD|^{1/2}(P^+_{c,\gamma}-P^+_{c,0})\|_{\mathcal{B}(\mathcal{H})}\leq  \frac{1}{4(1-\kappa_c)^{1/2}\lambda_{0,c}^{1/2}c}\|\gamma\|_X.
    \end{align}
\end{lemma}
\begin{proof}
    Using the resolvent formula, we have
    \begin{align}\label{eq:4.13}
        P^+_{c,0}-\Lambda^+_c=\frac{1}{2\pi}\int_{\R}\frac{1}{\cD-V-iz}V \frac{1}{\cD-iz}dz
    \end{align}
    and
    \begin{align}\label{eq:4.14}
       P^+_{c,\gamma}-P^+_{c,0}= \frac{1}{2\pi}\int_{\R}\frac{1}{\cD_{\gamma}-iz}W_{\gamma} \frac{1}{\cD_0-iz}dz.
    \end{align}
Note that
\begin{align*}
    \int_{\R}\frac{|A|}{|A|^2+z^2}dz=\pi,\qquad \textrm{for } A\neq 0,
\end{align*}
Then we infer from Hardy inequality and \eqref{eq:5.1-3},
\begin{align*}
  \MoveEqLeft  \left|\lc v, |\cD|^{1/2}(P^+_{c,0}-\Lambda^+_c)u \rc_{\cH}\right|\\
  &\leq \frac{1}{2\pi}\int_{\R}\|(\cD+iz)^{-1} |\cD|^{1/2}v\|_{\cH}\|V(\cD-iz)^{-1} u\|_{\cH}dz\\
    &\leq \frac{z}{\pi}\left(\int_{\R}\|(\cD-V+iz)^{-1} |\cD|^{1/2}v\|_{\cH}^2dz\right)^{1/2}\left(\int_{\R}\|(\cD-iz)^{-1} \nabla u\|_{\cH}^2 dz\right)^{1/2}\\
    &\leq z\||\mathcal{D}^c_0|^{-1/2}\cD|^{1/2}v\|_{\cH}\||\cD|^{-1/2}\nabla u\|_{\cH}\\
    &\leq zc^{-1}(1-\kappa_c)^{-1/2}\|v\|_{\cH}\|(-\Delta)^{1/2} u\|_{\cH}.
\end{align*}
Here $\mathcal{D}^c_{0}=\cD-V$ is the DF operator $\mathcal{D}^c_{\gamma}$ with $\gamma=0$. This gives \eqref{eq:P0-Lambda}.

The estimate \eqref{eq:P-P} can be found in \cite{sere2023new}. For the reader's convenience, we give the proof. As for \eqref{eq:P0-Lambda}, from \eqref{eq:5.1-1}, \eqref{eq:5.1-3} and \eqref{eq:5.1-5} we have
\begin{align*}
\MoveEqLeft \left|\lc v, |\cD|^{1/2}(P^+_{c,\gamma}-P^+_{c,0})u \rc_{\cH}\right|\\
 &\leq \frac{1}{2\pi}\|W_{\gamma}\|_{\mathcal{B}(\cH)}\int_{\R}\|(\mathcal{D}^c_{0}+iz)^{-1} |\cD|^{1/2}v\|_{\cH}\|(\mathcal{D}^c_{\gamma}-iz)^{-1} u\|_{\cH}dz\\
    &\leq \frac{\pi}{4}\|\gamma\|_{X}\||\mathcal{D}^c_0|^{-1/2}|\cD|^{1/2}v\|_{\cH}\||\mathcal{D}^c_{\gamma}|^{-1/2}u\|_{\cH}\\
    &\leq \frac{\pi}{4} c^{-1}(1-\kappa_c)^{-1/2}\lambda_{0,c}^{-1/2}\|\gamma\|_X\|v\|_{\cH}\|u\|_{\cH}.
\end{align*}
This ends the proof.
\end{proof}

Under Assumption \ref{ass:V}, the following lemma will be used to prove \eqref{eq:error-theta-Lambda'}.
\begin{lemma}\label{lem:Lambda-Pgamma'}
Under Assumption \ref{ass:V}, for $c$ large enough, we have
    \begin{align}\label{eq:P0-Lambda'}
        \||\cD|^{1/2}(P^+_{c,0}-\Lambda^+_c)(1-\Delta)^{-1}\|_{\cB(\cH)}\lesssim c^{-2}
    \end{align}
and
    \begin{align}\label{eq:P-P'}
        \||\cD|^{1/2}(P^+_{c,\gamma}-P^+_{c,0})\|_{\mathcal{B}(\mathcal{H})}\lesssim    c^{-2}( 1+ \|\gamma\|_{X^2}^2+c\|[W_{\gamma},\beta]\|_{\mathcal{B}(\cH)}).
    \end{align}
\end{lemma}
\begin{proof}
This proof is mainly based on \cite{meng2024rigorous}. We need a more delicate study of \eqref{eq:4.13} and \eqref{eq:4.14}. Note that for $A\neq 0$
\begin{align*}
    \int_{\R} \frac{1}{(A+iz)^2}dz=0.
\end{align*}
Using this equation and \eqref{eq:4.13},
\begin{align*}
   \MoveEqLeft  P^+_{c,0}-\Lambda^+_c\\
   &=\frac{1}{2\pi}\int_{\R}\frac{1}{\cD-iz}V \frac{1}{\cD-iz}dz +\frac{1}{2\pi}\int_{\R}\frac{1}{\cD-V-iz}V \frac{1}{\cD-iz}V \frac{1}{\cD-iz}dz\\
     &= \frac{1}{2\pi}\int_{\R}\frac{1}{\cD-iz}\Big[V, \frac{1}{\cD-iz}\Big]dz +\frac{1}{2\pi}\int_{\R}\frac{1}{\cD-V-iz}V \frac{1}{\cD-iz}V \frac{1}{\cD-iz}dz\\
     &=\frac{1}{2\pi}\int_{\R}\frac{1}{(\cD-iz)^2} [\cD,V] \frac{1}{\cD-iz}dz +\frac{1}{2\pi}\int_{\R}\frac{1}{\cD-V-iz}V \frac{1}{\cD-iz}V \frac{1}{\cD-iz}dz
\end{align*}
and by \eqref{eq:4.14},
\begin{align*}
  \MoveEqLeft   P^+_{c,\gamma}-P^+_{c,0}\\
     &= \frac{1}{2\pi}\int_{\R}\frac{1}{\cD_{0}-iz}W_{\gamma} \frac{1}{\cD_0-iz}dz+\frac{1}{2\pi}\int_{\R}\frac{1}{\cD_{\gamma}-iz}W_{\gamma}  \frac{1}{\cD_0-iz}V  \frac{1}{\cD_0-iz}dz\\
     &=\frac{1}{2\pi}\int_{\R}\frac{1}{\cD_{0}-iz}\Big[W_{\gamma}, \frac{1}{\cD_0-iz}\Big]dz+\frac{1}{2\pi}\int_{\R}\frac{1}{\cD_{\gamma}-iz}W_{\gamma}  \frac{1}{\cD_0-iz}W_{\gamma}  \frac{1}{\cD_0-iz}dz\\
     &=\frac{1}{2\pi}\int_{\R}\frac{1}{(\cD_{0}-iz)^2}[\cD_0,W_{\gamma}]\frac{1}{\cD_0-iz}dz+\frac{1}{2\pi}\int_{\R}\frac{1}{\cD_{\gamma}-iz}W_{\gamma}  \frac{1}{\cD_0-iz}W_{\gamma}  \frac{1}{\cD_0-iz}dz.
\end{align*}
Analogous to the proof of Lemma \ref{lem:Lambda-Pgamma}, we know
\begin{align*}
  \MoveEqLeft  \left|\left<v, |\cD|^{1/2}(P^+_{c,0}-\Lambda^+_c)\right>_\cH\right|\\
  &\leq \frac{1}{2\pi}\||\cD|^{-1}\|_{\cB(\cH)}\int_{\R}\|(\cD+iz)^{-1}|\cD|^{1/2}v\|_\cH\|[\cD,V](\cD-iz)^{-1}u\|_{\cH}dz\\
  &\quad+\frac{1}{2\pi}\int_{\R}\|(\cD-V+iz)^{-1}|\cD|^{1/2}v\|_\cH\|V(\cD-iz)^{-1}V(\cD-iz)^{-1}u\|_{\cH}dz\\
  &\lesssim c^{-2}\|v\|_{\cH}\|(1-\Delta)u\|_{\cH}
\end{align*}
where we used  Assumption \ref{ass:V}, $|[\cD,V]|\leq c|\nabla V|$ and by \eqref{eq:nabla-V},
\begin{align*}
\MoveEqLeft  \|V (\cD-iz)^{-1}V(\cD-iz)^{-1} u \|_{\cH}\lesssim \|(\cD-iz)^{-1} \nabla [V(\cD-iz)^{-1} u] \|_{\cH}\\
    &\lesssim c^{-2}\|\nabla [V(\cD-iz)^{-1} u] \|_{\cH} \lesssim c^{-2}\|(\cD-iz)^{-1}(1-\Delta) u] \|_{\cH}.
\end{align*}
Next by \eqref{eq:5.1-1}, \eqref{eq:5.1-3} and \eqref{eq:W-cD},
\begin{align*}
\MoveEqLeft \left|\lc v, (P^+_{c,\gamma}-P^+_{c,0})u \rc_{\cH}\right|\\
 &\leq \frac{1}{2\pi}\||\cD|^{-1}\|_{\cB(\cH)}\|[\cD_0,W_{\gamma}]\|_{\mathcal{B}(\cH)}\int_{\R}\|(\mathcal{D}^c_{0}-iz)^{-1} |\cD|^{1/2}v\|_{\cH}\|(\mathcal{D}^c_{\gamma}-iz)^{-1} u\|_{\cH}dz\\
 &\quad+\frac{1}{2\pi}\|W_{\gamma}\|_{\cB(\cH)}^2\||\cD_0|^{-1}\|_{\cB(\cH)}\int_{\R}\|(\mathcal{D}^c_{0}-iz)^{-1} |\cD|^{1/2}v\|_{\cH}\|(\mathcal{D}^c_{\gamma}-iz)^{-1} u\|_{\cH}dz\\
    &\lesssim c^{-3}\Big(\|\gamma\|_X^2+c\|\gamma\|_{X^2}+c^2\|[W_{\gamma},\beta]\|_{\mathcal{B}(\cH)}\Big)\|v\|_{\cH}\|u\|_{\cH}\\
    &\lesssim c^{-2}\Big( 1+ \|\gamma\|_{X^2}^2+c\|[W_{\gamma},\beta]\|_{\mathcal{B}(\cH)}\Big)\|v\|_{\cH}\|u\|_{\cH}.
\end{align*}
This ends the proof.
\end{proof}

Next, to use Lemma \ref{lem:error-bound}, we study the term $T_c(\gamma)-\gamma$ and $P^-_{c,\gamma}\gamma P^-_{c,\gamma}$.
\begin{lemma}\label{lem:T-I}
Let $\gamma\in \mathcal{U}_{c,R}\cap X^2$ and $\kappa_c<1$. If $\Lambda^+_c\gamma\Lambda^+_c=\gamma$, under Assumption \ref{ass:c} we have
\begin{align}\label{eq:6.7}
    \|T_c(\gamma)-\gamma\|_{X_c}\leq (2\pi+4\sqrt{2}z)(1+\|\gamma\|_{X^2})^2
\end{align}
and
\begin{align}\label{eq:P-gamma-P}
    \|P^-_{c,\gamma}\gamma P^-_{c,\gamma}\|_{X_c}\leq c^{-2}(q+8z^2)(1+\|\gamma\|_{X^2})^2.
\end{align}
In addition, under Assumption \ref{ass:V}, for $c$ large enough and $\gamma\in \cU_{c,R}\cap X^4$, we have
\begin{align}\label{eq:6.7'}
    \|T_c(\gamma)-\gamma\|_{X_c}\lesssim c^{-1}(1+ \|\gamma\|_{X^2}^2+c\|[W_{\gamma},\beta]\|_{\mathcal{B}(\cH)})\|\gamma\|_{X^4}
\end{align}
and
\begin{align}\label{eq:P-gamma-P'}
    \|P^-_{c,\gamma}\gamma P^-_{c,\gamma}\|_{X_c}\lesssim c^{-4}(1+\|\gamma\|_{X^2}^2+c\|[W_{\gamma},\beta\|_{\mathcal{B}(\cH)})^2\|\gamma\|_{X^4}.
\end{align}
\end{lemma}
\begin{proof}
Under Assumption \ref{ass:c}, we have
\begin{align*}
    \kappa_c\leq \frac{1}{2}, \qquad \lambda_{0,c}\geq \frac{1}{2}.
\end{align*}
Then we have
\[
T_c(\gamma)-\gamma=(P^+_{c,\gamma}-\Lambda^+_c)\gamma P^+_{c,\gamma}+\gamma (P^+_{c,\gamma}-\Lambda^+_c).
\]
Using Lemma \ref{lem:Lambda-Pgamma}, \eqref{eq:DP} and the fact that $\||\cD|^{1/2}(1-\Delta)^{-1/4}\|_{\cB(\cH)}\leq c$,
\begin{align*}
\|T_c(\gamma)-\gamma\|_{X_c}&\leq \frac{2(1+\kappa_c)^{1/2}}{(1-\kappa_c)^{1/2}}\||\cD|^{1/2}(P^+_{c,\gamma}-P^+_{c,0})\|_{\mathcal{B}(\mathcal{H})}\|\gamma|\cD|^{1/2}\|_{\mathfrak{S}_1}\\
&\quad+\frac{2(1+\kappa_c)^{1/2}}{(1-\kappa_c)^{1/2}}\||\cD|^{1/2}(P^+_{c,0}-\Lambda^+_c)(-\Delta)^{-1/2}\|_{\mathcal{B}(\mathcal{H})}\|(-\Delta)^{1/2}\gamma|\cD|^{1/2}\|_{\mathfrak{S}_1}\\
&\leq \frac{\sqrt{2}\pi}{2(1-\kappa_c)\lambda_{0,c}^{1/2}}\|\gamma\|_{X}^2+\frac{2\sqrt{2}z}{(1-\kappa_c)}\|\gamma\|_{X^2}\leq (2\pi+4\sqrt{2}z)(1+\|\gamma\|_{X^2})^2.
\end{align*}
Here we use the fact that $0\leq \kappa_c<1$ and $\lambda_{0,c}\leq 1$.

Concerning the second one, we have
\[
P^-_{c,\gamma}\gamma P^-_{c,\gamma}=P^-_{c,\gamma}(\Lambda^+_c-P^+_{c,\gamma})\gamma (\Lambda^+_c-P^+_{c,\gamma})P^-_{c,\gamma}.
\]
Using $\kappa_c<1$, $\lambda_{0,c}\leq 1$ and the identity $\Lambda^+_c-P^+_{c,\gamma}= (\Lambda^+_c-P^+_{c,0})+(P^+_{c,0}-P^+_{c,\gamma})$, then by H\"older's inequality \eqref{eq:Holder},
\begin{align*}
 \MoveEqLeft\|P^-_{c,\gamma}(T_c(\gamma)-\gamma)P^-_{c,\gamma}\|_{X_c}\\
 &\leq \frac{1+\kappa_c}{1-\kappa_c}\left(\||\cD|^{1/2}(P^+_{c,0}-P^+_{c,\gamma})\|_{\mathcal{B}(\mathcal{H})}\|\gamma\|_{\mathfrak{S}_1}^{1/2}+\||\cD|^{1/2}(P^+_{c,0}-\Lambda^+_c)(-\Delta)^{-1/2}\|_{\mathcal{B}(\mathcal{H})}\|\gamma\|_{X^2}^{1/2}\right)^2\\
    &\leq c^{-2}(q+8z^2)(1+\|\gamma\|_{X^2})^2.
\end{align*}
where we used $\gamma\in \cU_{c,R}\subset \Gamma_{q}$.

\medskip

Now we consider the case under Assumption \ref{ass:V}. For $c$ large enough, arguing as above and by Lemma \ref{lem:Lambda-Pgamma'} we have
\begin{align*}
\|T_c(\gamma)-\gamma\|_{X_c}&\lesssim\||\cD|^{1/2}(P^+_{c,\gamma}-P^+_{c,0})\|_{\mathcal{B}(\mathcal{H})}\|\gamma|\cD|^{1/2}\|_{\mathfrak{S}_1}\\
&\quad+\||\cD|^{1/2}(P^+_{c,0}-\Lambda^+_c)(1-\Delta)^{-1}\|_{\mathcal{B}(\mathcal{H})}\|(1-\Delta)\gamma|\cD|^{1/2}\|_{\mathfrak{S}_1}\\
&\lesssim c^{-2}\|(1-\Delta)\gamma|\cD|^{1/2}\|_{\mathfrak{S}_1}+ c^{-2}\Big(1+\|\gamma\|_{X^2}^2+c\|[W_{\gamma},\beta]\|_{\mathcal{B}(\cH)}\Big)\|\gamma|\cD|^{1/2}\|_{\mathfrak{S}_1}\\
&\lesssim c^{-1}\Big(1+ \|\gamma\|_{X^2}^2+c\|[W_{\gamma},\beta]\|_{\mathcal{B}(\cH)}\Big)\|(1-\Delta)\gamma(1-\Delta)^{1/4}\|_{\mathfrak{S}_1}\\
&\lesssim c^{-1}(1+ \|\gamma\|_{X^2}^2+c\|[W_{\gamma},\beta]\|_{\mathcal{B}(\cH)})\|\gamma\|_{X^4},
\end{align*}
and 
\begin{align*}
 \MoveEqLeft\|P^-_{c,\gamma}(T_c(\gamma)-\gamma)P^-_{c,\gamma}\|_{X_c}\\
 &\lesssim \||\cD|^{1/2}(P^+_{c,0}-P^+_{c,\gamma})\|_{\mathcal{B}(\mathcal{H})}^2\|\gamma\|_{\mathfrak{S}_1}\\
 &\quad+\||\cD|^{1/2}(P^+_{c,0}-\Lambda^+_c)(1-\Delta)^{-1}\|_{\mathcal{B}(\mathcal{H})}^2\|\gamma\|_{X^4}\\
 &\quad +\||\cD|^{1/2}(P^+_{c,0}-P^+_{c,\gamma})\|_{\mathcal{B}(\mathcal{H})}\||\cD|^{1/2}(P^+_{c,0}-\Lambda^+_c)(1-\Delta)^{-1}\|_{\mathcal{B}(\mathcal{H})} \|(1-\Delta) \gamma\|_{\mathfrak{S}_1} \\
    &\lesssim c^{-4}(1+\|\gamma\|_{X^2}^2+c\|[W_{\gamma},\beta]\|_{\mathcal{B}(\cH)})^2\|\gamma\|_{X^4}.
\end{align*}

This ends the proof.
\end{proof}

\subsection{Proof of Theorem \texorpdfstring{\ref{th:4.1}}{}}\label{sec:5.3}
We prove \eqref{eq:E_cq-E_c-lambda} and \eqref{eq:E_cq-E_c-lambda-c2} separately.
\subsubsection{Proof of \texorpdfstring{\eqref{eq:E_cq-E_c-lambda}}{}}

We first choose $R$ such that $\Lambda^+_c\gamma_{*}^{\rm HF}\Lambda^+_c\in\mathcal{U}_{c,R}$.
\begin{lemma}\label{lem:R}
 Let $q\in \R^+$ and $z\in \R^+$ such that $q\leq z$. Then for any $c$ satisfying Assumption \ref{ass:c}, we have 
 \begin{align*}
     \gamma_*^c\in\mathcal{U}_{c,R_0},\qquad \Lambda^+_c\gamma_{*}^{\rm HF}\Lambda^+_c\in\mathcal{U}_{c,R_0}
 \end{align*}
 with $R_0$ given by \eqref{def:R_0} independently of $c$ and $L_c=2a_c R_0\leq \frac{1}{2}$.
\end{lemma}
\begin{proof}
First of all, $\gamma_*^c=P^+_{c,\gamma_*^c}\gamma_*^c P^+_{c,\gamma_*^c}\in \Gamma_{q}^+$. According to Remark \ref{rem:R} and \eqref{eq:a-c==RcDF},
\begin{align*}
  \MoveEqLeft  \frac{1}{c}\|\gamma_*^c|\cD|^{1/2}\|_{\mathfrak{S}_1}+ \frac{A_c}{c^2}\|T(\gamma_*^c)-\gamma_*^c\|_{X_c}\\
  &\qquad=\frac{1}{c}\|\gamma_*^c|\cD|^{1/2}\|_{\mathfrak{S}_1}< R^{\rm DF}_c\leq 1+4q<R_0.
\end{align*}
Thus, $\gamma_*^c\in\mathcal{U}_{c,R_0}$.

We turn to prove $\Lambda^+_c\gamma_{*}^{\rm HF}\Lambda^+_c\in\mathcal{U}_{c,R_0}$. We have $\Lambda^+_c\gamma_*^{\rm HF}\Lambda^+_c\in \Gamma_{q}$ and
\begin{align*}
    \frac{1}{c}\|\Lambda^+_c\gamma_*^{\rm HF}\Lambda^+_c|\cD|^{1/2}\|_{\mathfrak{S}_1}\leq  \frac{1}{c}\|\gamma_*^{\rm HF}|\cD|^{1/2}\|_{\mathfrak{S}_1} \leq \|\gamma_*^{\rm HF}\|_{X}.
\end{align*}
where we used the fact that $\Lambda^+_c$ is a projector, $[\Lambda^+_c,\cD]=0$ and $|\cD|^{1/2}\leq c(1-\Delta)^{1/4}$. From \eqref{eq:6.7} and as $0\leq \Lambda^+\leq 1$, we also have
\begin{align}\label{eq:T-+-gammaHF}
   \MoveEqLeft \|T_c(\Lambda^+_c\gamma_*^{\rm HF}\Lambda^+_c)-\Lambda^+_c\gamma_*^{\rm HF}\Lambda^+_c\|_{X_c}\leq (2\pi+4\sqrt{2}z)(1+\|\gamma_*^{\rm HF}\|_{X^2})^2.
\end{align}
Then from \eqref{eq:a-c==RcDF} and by Assumption \ref{ass:c},
\begin{align}\label{eq:2ac-R<1}
    2a_c R_0\leq \frac{\pi }{c}R_0 \leq \frac{1}{2}.
\end{align}
As $\frac{2q}{c}\leq \kappa_c\leq \frac{1}{2}$, we have
\begin{align*}
    A_c=\max\left\{\frac{1}{1-2a_cR_0},\frac{2+a_c q}{2}\right\}\leq \max\left\{2, \frac{2+\pi c^{-1}q}{2}\right\}\leq \max\left\{2, \frac{2+\pi/4}{2}\right\}\leq 2.
\end{align*}
Thus, $\Lambda^+_c\gamma_*^{\rm HF}\Lambda^+_c$ satisfies
\begin{align*}
 \MoveEqLeft\frac{1}{c}\|\Lambda^+_c\gamma_*^{\rm HF}\Lambda^+_c|\cD|^{1/2}\|_{\mathfrak{S}_1}+\frac{A_c}{c^2}\|T_c(\Lambda^+_c\gamma_*^{\rm HF}\Lambda^+_c)-\Lambda^+_c\gamma_*^{\rm HF}\Lambda^+_c\|_{X_c}\\
    &\qquad\qquad \leq \|\gamma_*^{\rm HF}\|_{X^2} + 4(\pi+2\sqrt{2}z)(1+\|\gamma_*^{\rm HF}\|_{X^2})^2<R_0.
\end{align*}
This shows that $\Lambda^+_c\gamma_{*}^{\rm HF}\Lambda^+_c\in\mathcal{U}_{c,R_0}$. From \eqref{eq:2ac-R<1}, we also know $L_c=2a_c R_0\leq \frac{1}{2}$. This ends the proof.
\end{proof}

Now assumptions in Lemma \ref{lem:error-bound} are satisfied under Assumption \ref{ass:c}. Gathering together Lemma \ref{lem:error-bound}, Lemma \ref{lem:T-I} with $R=R_0$ and Lemma \ref{lem:R}, we conclude that
\begin{lemma}\label{lem:6.9}
 Let $q\in \R^+$ and $z\in \R^+$ such that $q\leq z$.Then for any $c$ satisfying Assumption \ref{ass:c}, we have 
    \begin{align*}
     \MoveEqLeft   |E_c(\Lambda^+_c\gamma_*^{\rm HF}\Lambda^+_c)-\mathcal{E}_c(\Lambda^+_c\gamma_*^{\rm HF}\Lambda^+_c)|=\cO(c^{-2}).
    \end{align*}
\end{lemma}
\begin{proof}
  By Lemma \ref{lem:R}, we can choose $R=R_0$, then $L_c\leq \frac{1}{2}$. In addition, by Assumption \ref{ass:c},
  \begin{align*}
      c^{-1}R_0\leq \frac{1}{4\pi},\qquad c^{-1}q\leq \frac{1}{4}.
  \end{align*}
  Then by Lemma \ref{lem:error-bound}, \eqref{eq:P-gamma-P} and \eqref{eq:T-+-gammaHF}
  \begin{align*}
   \MoveEqLeft   |E_c(\Lambda^+_c\gamma_*^{\rm HF}\Lambda^+_c)-\mathcal{E}_c(\Lambda^+_c\gamma_*^{\rm HF}\Lambda^+_c)|\\
   &\lesssim (1+\|\gamma_*^{\rm HF}\|_{X^2})^4 c^{-2} +(1+\|\gamma_*^{\rm HF}\|_{X^2})^2c^{-2}=\cO(c^{-2}).
  \end{align*}
This ends the proof.
\end{proof}

Lemma \ref{lem:6.9} and \eqref{eq:EDF-min-E-retra-HF} shows that
\begin{align}
    E_{c,q}\leq \mathcal{E}_c(\Lambda^+_c\gamma_*^{\rm HF}\Lambda^+_c)+\cO(c^{-2})
\end{align}
which proves \eqref{eq:E_cq-E_c-lambda}.

 \subsubsection{Proof of \texorpdfstring{\eqref{eq:E_cq-E_c-lambda-c2}}{}}
We are now in the position to prove \eqref{eq:E_cq-E_c-lambda-c2}. We first verify \eqref{eq:condition-R},
\begin{lemma}\label{lem:R'}
Under Assumption \ref{ass:V}, for $c$ large enough, we have 
 \begin{align*}
     \Lambda^+_c\widetilde{\gamma}_c^{\rm HF}\Lambda^+_c\in\mathcal{U}_{c,R_0}
 \end{align*}
 with the same $R_0<\frac{1}{2a_c}$ as in Lemma \ref{lem:R}.
\end{lemma}
\begin{proof}
    The proof is essentially the same as for Lemma \ref{lem:R}. We only need to use Lemma \ref{lem:gammaHF'-estimate} in addition. Indeed, by Lemma \ref{lem:gammaHF'-projection}, $  \widetilde{\gamma}_c^{\rm HF}\in \Gamma_{q}$. Then by Lemma \ref{lem: gammaHF''-gammaHF'} and \eqref{eq:gammaHF''}, as $c\to \infty$,
    \[
  \widetilde{\gamma}_c^{\rm HF}\to \gamma_*^{\rm HF}
    \] 
    in $X$. According to \eqref{eq:P-P}, the mapping $\gamma \mapsto T_c(\gamma)$ is continuous in $X$. This continuity implies that for $c$ large enough, 
    \begin{align*}
        \Lambda^+_c\widetilde{\gamma}_c^{\rm HF}\Lambda^+_c\in\mathcal{U}_{c,R_0}
    \end{align*}
    with the same $R_0<\frac{1}{2a_c}$ as in Lemma \ref{lem:R}.
\end{proof}
    \begin{lemma}\label{lem:6.10}
 Let $q\in \R^+$ and $z\in \R^+$ such that $q\leq z$.Then for any $c$ satisfying Assumption \ref{ass:c} and Assumption \ref{ass:V}, we have 
    \begin{align*}
     \MoveEqLeft   |E_c(\Lambda^+_c\gamma_*^{\rm HF}\Lambda^+_c)-\mathcal{E}_c(\Lambda^+_c\gamma_*^{\rm HF}\Lambda^+_c)|=\cO(c^{-4}).
    \end{align*}
\end{lemma}
\begin{proof}
    By \eqref{eq:6.7'}, \eqref{eq:P-gamma-P'} and Lemma \ref{lem:error-bound}, under Assumption \ref{ass:V}, for $c$ large enough, we have
\begin{align}\label{eq:4.23}
 \MoveEqLeft   |E_c(\Lambda^+_c \widetilde{\gamma}_c^{\rm HF}\Lambda^+_c )-\cE_c((\Lambda^+_c \widetilde{\gamma}_c^{\rm HF}\Lambda^+_c ))|\notag\\
 &\lesssim c^{-4}\Big(1+\|\Lambda^+_c \widetilde{\gamma}_c^{\rm HF}\Lambda^+_c\|_{X^4}^2+c\|[W_{\Lambda^+_c \widetilde{\gamma}_c^{\rm HF}\Lambda^+_c},\beta]\|_{\mathcal{B}(\cH)}\Big)^2.
\end{align}
By Lemma \ref{lem:gammaHF'-estimate}, 
\begin{align}\label{eq:4.24}
    \|\Lambda^+_c \widetilde{\gamma}_c^{\rm HF}\Lambda^+_c\|_{X^4}\leq \|\widetilde{\gamma}_c^{\rm HF}\|_{X^4}=\cO(1).
\end{align}
It remains to study $[W_{\Lambda^+_c \widetilde{\gamma}_c^{\rm HF}\Lambda^+_c},\beta]$. We decompose $\Lambda^+_c \widetilde{\gamma}_c^{\rm HF}\Lambda^+_c$ into four blocks:
\begin{align*}
    \Lambda^+_c \widetilde{\gamma}_c^{\rm HF}\Lambda^+_c=\sum_{j,j'\in\{\rL,\rS\} }\cK_{j} \Lambda^+_c \widetilde{\gamma}_c^{\rm HF}\Lambda^+_c \cK_{j'}.
\end{align*}
Note that for any $4\times 4$ matrix 
\begin{align*}
    A:=\begin{pmatrix}
        A_{1,1} & A_{1,2} \\
        A_{2,1} &A_{2,2}
    \end{pmatrix},\qquad A_{m,n}\in {\rm Mat}_{2\times 2},
\end{align*}
we have
\begin{align*}
   [A,\beta]=\begin{pmatrix}
       A_{1,1} & -A_{1,2}\\ A_{2,1} &-A_{2,2}
   \end{pmatrix}- \begin{pmatrix}
       A_{1,1} & A_{1,2}\\ -A_{2,1} &-A_{2,2}
   \end{pmatrix}= 2\begin{pmatrix}
       0& -A_{1,2}\\ A_{2,1}&0
   \end{pmatrix}.
\end{align*}
Thus,
\begin{align*}
   [W_{\Lambda^+_c \widetilde{\gamma}_c^{\rm HF}\Lambda^+_c},\beta]= W_{2,[\Lambda^+_c \widetilde{\gamma}_c^{\rm HF}\Lambda^+_c,\beta]}=-2 W_{2,\cK_{\rL } \Lambda^+_c \widetilde{\gamma}_c^{\rm HF}\Lambda^+_c \cK_\rS}+2 W_{2,\cK_{\rS } \Lambda^+_c \widetilde{\gamma}_c^{\rm HF}\Lambda^+_c \cK_\rL}.
\end{align*}
Then by H\"older's inequality \eqref{eq:Holder}, \eqref{eq:5.1-1'}, Corollary \ref{cor:lambda-non-gamma'} and Lemma \ref{lem:gammaHF'-estimate},
\begin{align}\label{eq:4.25}
    \|[W_{\Lambda^+_c \widetilde{\gamma}_c^{\rm HF}\Lambda^+_c},\beta]\|_{\cB(\cH)}&\lesssim \|\cK_{\rS } \Lambda^+_c \widetilde{\gamma}_c^{\rm HF}\Lambda^+_c \cK_\rL\|_{X}\notag\\
    &\lesssim \|\cK_\rS \Lambda^+_c \widetilde{\gamma}_c^{\rm HF} \Lambda^+_c \cK_\rS\|_X^{1/2}\|\cK_\rL \Lambda^+_c \widetilde{\gamma}_c^{\rm HF} \Lambda^+_c \cK_\rL\|_X^{1/2}\notag\\
    &\lesssim c^{-1}\|\cK_\rL \widetilde{\gamma}_c^{\rm HF}\cK_\rL\|_{X^3}\leq c^{-1}\|\widetilde{\gamma}_c^{\rm HF}\|_{X^3}=\cO(c^{-1}).
\end{align}
From \eqref{eq:4.23}-\eqref{eq:4.25}, we conclude that
\begin{align*}
     |E_c(\Lambda^+_c \widetilde{\gamma}_c^{\rm HF}\Lambda^+_c )-\cE_c((\Lambda^+_c \widetilde{\gamma}_c^{\rm HF}\Lambda^+_c ))|=\cO(c^{-4}).
\end{align*}
\end{proof}
Lemma \ref{lem:6.10} and \eqref{eq:EDF-min-E-retra-HF} shows that
\begin{align}
    E_{c,q}\leq \mathcal{E}_c(\Lambda^+_c\widetilde{\gamma}_*^{\rm HF}\Lambda^+_c)+\cO(c^{-4})
\end{align}
which proves \eqref{eq:E_cq-E_c-lambda-c2}. Now the proof of Theorem \ref{th:4.1} is completed.

\section{From DF problem to HF problem}\label{sec:7}
In this section, we are trying to understand the relationship between the DF ground-state energy and some HF energies. The main result of this section is the following.
\begin{theorem}[From DF problem to HF problem]\label{th:5.1}
    Let $\gamma_*^c\in \Gamma_{q}^+$ be any DF minimizer of $E_{c,q}$. Then under Assumption \ref{ass:c}, we have
    \begin{align}\label{eq:EHF-gammac}
        \cE^{\rm HF}(\cK_\rL \gamma^c_*\cL_\rL)\leq E_{c,q}+\cO(c^{-2}).
    \end{align}
In addition, under Assumption \ref{ass:V} and assume that $\delta_c=0$ with $\delta_c$ being defined by \eqref{eq:gamma-DF}, then for $c$ large enough, there exists a density matrix $\widetilde{\gamma}^c_*\in \Gamma_q^{\rm HF}$ such that
\begin{align}\label{eq:EHF-gammac-c2}
    \cE^{\rm HF}(\widetilde{\gamma}^c_*)\leq E_{c,q}-\widetilde{\cE}_c^{(2)}(\gamma_*^c)+\cO(c^{-4})
\end{align}
where
\begin{align*}
    \widetilde{\cE}_c^{(2)}(\gamma_*^c)&=-\frac{1}{4c^2}\Tr_{\cH}[(\cD_{\gamma^c_*}-c^2)\gamma^c_*\cL^2] \\
    &+\frac{1}{4c^2}\left(\Tr_\cH[(-V+W_{1,\cK_\rL \gamma_*^c\cK_\rL })\cL \cK_\rL \gamma_*^c  \cK_\rL \cL]-\Tr_\cH[W_{2,\cK_\rL \gamma_*^c \cK_\rL \cL}\cL \cK_\rL \gamma_*^c \cK_\rL]\right).
\end{align*}
More precisely, $\widetilde{\gamma}^c_* \in \Gamma_q^{\rm HF}$ is defined by \eqref{eq:gammaDF'} below and it is a non-relativistic renormalization of $\gamma_*^c$. 
\end{theorem}

The proof of Theorem \ref{th:5.1} is organized as follows: In Section \ref{sec:6.1}, we recall some useful estimates on the structure of the DF minimizer $\gamma_*^c$. The aim is to use $\cK_\rL\gamma_*^c\cK_\rL$ or $\widetilde{\gamma}^c_*$ to replace $\gamma_*^c$ in the proof of Theorem \ref{th:5.1}. In Section \ref{sec:6.2}, we prove \eqref{eq:EHF-gammac}. This is a direct application of the structure of $\gamma_*^c$. In Section \ref{sec:6.3}, we first construct the density matrix $\widetilde{\gamma}^c_*$ and study its property. The finer structure of $\widetilde{\gamma}^c_*$ gives the correction term $\widetilde{\cE}_c^{(2)}(\gamma_*^c)$. In particular, we shall point out that in Section \ref{sec:6.3.1}, we provide a new estimates on the kinetic term, which plays an essential role for the proof of Theorem \ref{th:5.1}.

As \eqref{EDF-EHF}, we study the kinetic term, the potential between nuclei and electrons, the potential between electrons separately:
\begin{align}\label{eq:EDF-EHF2}
  \MoveEqLeft   \cE_c(\gamma_{*}^c)- \cE^{\rm HF}(\gamma)\notag\\
     &= \underbrace{\Tr_{\cH}[(\cD-c^2) \gamma_{*}^c)] -\Tr_{\cH}[H_0 \gamma]}_{\textrm{kinetic term}}\notag\\
     &\quad- \underbrace{\Tr_{\cH}[V\gamma_*^c]-\Tr_{\cH}[V\gamma]}_{\textrm{potential between electrons and nuclei}}+\underbrace{\Tr_{\cH}[W_{\gamma_*^c}\gamma_*^c]-\Tr_{\cH}[W_{\gamma}\gamma]}_{\textrm{potential between electrons and electrons}}.
\end{align}
Before going further, we study the structure of $\gamma_*^c$ for the construction of $\widetilde{\gamma}_*^c$ and also for the proof of \eqref{eq:EHF-gammac} and \eqref{eq:EHF-gammac-c2}.

\subsection{Structure of \texorpdfstring{$\gamma_*^c$}{}}\label{sec:structure-gamma-c}\label{sec:6.1}

According to \eqref{eq:gamma-dec}, we can rewrite $\gamma_*^c$ as
\begin{align}\label{eq:gamma-c-dec}
\gamma_*^c=\sum_{n=1}^\infty\mu_n^c\left|u_{c,n}\right>\left<u_{c,n}\right|,\qquad \cD_{\gamma_*^c} u_{c,n}=\lambda_{n}^c u_{c,n}
\end{align}
where $\mu_n^c\geq 0$, $\sum_{n=1}^\infty \mu_n^c=q$, and, for any $n\geq 1$, $u_{c,n}$ is a normalized eigenfunction of $\mathcal{D}^c_{\gamma_*^c}$. In particular, if Theorem \ref{th:non-unfill} holds, then we can write
\begin{align}\label{eq:gamma-c-dec'}
\gamma_*^c=\sum_{n=1}^q\left|u_{c,n}\right>\left<u_{c,n}\right|.
\end{align}

We have
\begin{lemma}\label{lem:u-dec}
Under Assumption \ref{ass:c}, for any $n\geq 1$,
\begin{align*}
    \|u_{c,n}\|_{H^1}\leq K_1, \quad 
    \|u_{c,n}^{\rm S}\|_{H^1(\R^3;\C^2)}\leq \frac{1}{c}K_2,\quad \left\|u_{c,n}^{\rm S}-\frac{1}{2c}\cL u_{c,n}^{\rm L}\right\|_{L^2(\R^3;\C^2)}\leq \frac{1}{c^3}K_3
\end{align*}
and
\begin{align*}
  \|\gamma_*^c\|_{X^2}\leq K_1^2q,\qquad \|\cK_\rS \gamma_*^c \cK_\rS \|_{X^2}\leq \frac{1}{c^2}K_2^2q.  
\end{align*}
where $K_1,K_2,K_3$ are constants independent of $c$ and $n$, and are defined by \eqref{eq:K1}, \eqref{eq:K2} and \eqref{eq:K3} below respectively.
\end{lemma}
\begin{proof}
Similar proof can be found in \cite[Lemma 7 and Theorem 3]{esteban2001nonrelativistic} or \cite[Lemma B.1]{meng2024rigorous}. For the reader's convenience and for the proof of Lemma \ref{lem:u-dec'} below, we provide the details. Before going further, we need some estimates on the eigenvalues $\lambda^c_n$. 

Note that the potential $W_{\gamma}$ is non-negative. Applying the abstract min-max theorem
(see e.g., \cite{dolbeault2000eigenvalues}) to the self-adjoint operator $\mathcal{D}^c_{\gamma_*^c}$ and the splitting of $\cH$ associated with the free projectors $\Lambda_c^\pm$, we infer that for any $n\in \N^+$,
\begin{align}
c^2\geq \sigma_n^+(\mathcal{D}_{\gamma_*^c}^c)\geq \sigma_n^+(\cD-V),
\end{align}
where $\sigma_n^+(A)$ is the $n$-th positive eigenvalue (counted with multiplicity) of the operator $A$. Here assumptions in \cite{dolbeault2000eigenvalues} are verified under Assumption \ref{ass:c} (see e.g., \cite[Lemma 3.6]{sere2023new}).

According to the spectral analysis of Dirac operator (see e.g.,\cite{thaller2013dirac,sere2023new}), it is easy to see that there exists a constant $e>0$ independent of $c$ such that
\begin{align}\label{eq:7.7}
   0\leq  c^2-\sigma_n^+(\cD-V) \leq e.
\end{align}
This implies that for any $n\in \N^+$,
\begin{align}\label{eq:lambda-c2}
  c^2-e\leq \inf \sigma^+(\cD-V)\leq  \lambda_n^c \leq c^2.
\end{align}

\medskip

Next, we prove $\|u_{c,n}\|_{H^1}\leq K_1$. As $\cD_{\gamma_*^c} u_{c,n}=\lambda_{n}^c u_{c,n}$, we know
\begin{align*}
    \|\cD u_{c,n}\|_{\cH}=\|(\lambda_n^c+V-W_{\gamma})u_{c,n}\|_{\cH}.
\end{align*}
By Hardy's inequality and $|\lambda|\leq c^2$, we infer that
\begin{align*}
\MoveEqLeft    c^4\|u_{c,n}\|_{\cH}^2+c^2\|\nabla u_{c,n}\|_{\cH}^2\\
&\leq c^4\|u_{c,n}\|_{\cH}^2+4(z+q)c^2\|\nabla u_{c,n}\|_\cH\|u_{c,n}\|_\cH+4(z+q)^2\|\nabla u_{c,n}\|_\cH^2
\end{align*}
which implies that for any $n\geq 1$
\begin{align*}
    \|\nabla u_{c,n}\|_{\cH}\leq \left(\frac{4z+4q}{1-\kappa_c^2}\right)^{1/2},\qquad \|u_{c,n}\|_{H^1}\leq \left(\frac{1+ 4z+4q}{1-\kappa_c^2}\right)^{1/2}.
\end{align*}
Note that under Assumption \ref{ass:c}, $\kappa_c\leq \frac{1}{2}$. Thus 
\begin{align*}
    \|\gamma_*^c\|_{X^2}=\sum_{n\geq 1}\mu_{n}^c\|u_{c,n}\|^2_{H^1} \leq K_1^2q.
\end{align*}
with
\begin{align}\label{eq:K1}
    K_1:=2\left(1+ 4z+4q\right)^{1/2}.
\end{align}

\medskip

Note that the equation $\mathcal{D}^c_{\gamma_*^c} u_{c,n}=\lambda^c_n u_{c,n}$ can be rewritten as
\begin{align}\label{eq:decom-Dirac}
    \begin{pmatrix}
    c\mathcal{L} u_{c,n}^{\rm S}\\
     c\mathcal{L} u_{c,n}^{\rm L}
\end{pmatrix} +(-V+W_{\gamma_*^c}) u_{c,n} =\begin{pmatrix}
    (\lambda_n^c-c^2)u_{c,n}^{\rm L}\\
    (\lambda_n^c+c^2)u_{c,n}^{\rm S}
\end{pmatrix}.
\end{align}
Dividing by $c$ the first equation of \eqref{eq:decom-Dirac} and by Hardy's inequality, \eqref{eq:lambda-c2} and \eqref{eq:5.1-2}, we get
\begin{align}\label{eq:nabla-uS}
    \left\|\nabla u_{c,n}^{\rm S}\right\|_{L^2(\R^3;\C^2)}&= \left\|\cL u_{c,n}^{\rm S}\right\|_{L^2(\R^3;\C^2)}\\
    &\leq \frac{1}{c}\|\cK_{\rm L}(-V+W_{\gamma_*^c})u_{c,n}\|_{\cH}+\frac{|\lambda_n^c-c^2|}{c^2}\|\cK_{\rm L} u_{c,n}\|_{\cH}\notag\\
    &\leq \frac{2(z+q)+e}{c}\|u_{c,n}\|_{H^1}.
\end{align}
Dividing by $2c^2$ the second equation of \eqref{eq:decom-Dirac} and using \eqref{eq:5.18} and \eqref{eq:lambda-c2}, we get
\begin{align}\label{eq:uS-LuL}
   \left\|u_{c,n}^{\rm S}- \frac{1}{2c}\cL u_{c,n}^{\rm L}\right\|_{L^2(\R^3;\C^2)}&\leq 
  \frac{1}{2c^2}\|\cK_{\rm S}(-V+W_{\gamma_*^c})u_{c,n} \|_{\cH}+\frac{|\lambda_n^c-c^2|}{2c^2}\|\cK_{\rm S} u_{c,n}\|_{\cH}\notag\\
  &\leq \frac{2(q+z)+e}{2c^2}\|\cK_{\rm S} u_{c,n}\|_{H^1}+ \frac{q^{1/2}}{c^2}\|\cK_\rS \gamma_*^c \cK_\rS\|_{\mathfrak{S}_1}^{1/2}\|u_{c,n}\|_{H^1}
\end{align}
since $\cK_\rS W_{2,\gamma_*^c}=W_{2,\cK_\rS \gamma_*^c}$. Then for any $n\geq 1$,
\begin{align*}
   \left\|u_{c,n}^{\rm S}- \frac{1}{2c}\cL u_{c,n}^{\rm L}\right\|_{L^2(\R^3;\C^2)}&\leq \frac{2(q+z)+e}{2c^2}\|u_{c,n}\|_{H^1}+ \frac{q}{c^2}\|u_{c,n}\|_{H^1}\leq \frac{2(z+2q)+e}{2c^2}K_1
\end{align*}
and as $\|\cL u_{c,n}^\rL\|_{L^2(\R^3;\C^2)}=\|\nabla  u_{c,n}^\rL\|_{L^2(\R^3;\C^2)}\leq \|u_{c,n}\|_{H^1}$,
\begin{align*}
    \|u_{c,n}^{\rm S}\|_{L^2(\R^3;\C^2)}&\leq \frac{1}{2c}\|\cL u_{c,n}^{\rL}\|_{L^2(\R^3;\C^2)}+ \left\|u_{c,n}^{\rm S}- \frac{1}{2c}\cL u_{c,n}^{\rm L}\right\|_{L^2(\R^3;\C^2)}\\
    &\leq  \frac{1}{2c}\left(1+\frac{2(z+2q)+e}{c}\right)K_1.
\end{align*}
This and \eqref{eq:nabla-uS} show that under Assumption \ref{ass:c},
\begin{align*}
    \|u_{c,n}^{\rm S}\|_{H^1(\R^3;\C^2)}\leq \frac{1}{c}K_2
\end{align*}
with
\begin{align}\label{eq:K2}
    K_2:=  \left(1+4(z+2q)+2e\right)K_1.
\end{align}
Then
\begin{align*}
    \|\cK_\rS \gamma_*^c \cK_\rS\|_{X^2}\leq \frac{1}{c^2}K_2^2q.
\end{align*}
Inserting above two estimates into \eqref{eq:uS-LuL},
\begin{align*}
     \left\|u_{c,n}^{\rm S}- \frac{1}{2c}\cL u_{c,n}^{\rm L}\right\|_{L^2(\R^3;\C^2)}\leq \frac{1}{c^3}\left(\frac{2(q+z)+e}{2} K_2+ K_1K_2 q\right).
\end{align*}
Thus
\begin{align}\label{eq:K3}
    K_3:=\frac{2(q+z)+e}{2} K_2+ K_1K_2 q.
\end{align}
This completes the proof.
\end{proof}

In addition, we can get higher regularity under Assumption \ref{ass:V}. More precisely,
\begin{lemma}\label{lem:u-dec'}
Under Assumption \ref{ass:V}, for $c$ large enough
\begin{align*}
    \|u_{c,n}\|_{H^2}=\cO(1), \quad 
    \|u_{c,n}^{\rm S}\|_{H^2(\R^3;\C^2)}=\cO(c^{-1}),\quad \left\|u_{c,n}^{\rm S}-\frac{1}{2c}\cL u_{c,n}^{\rm L}\right\|_{H^1(\R^3;\C^2)}= \cO(c^{-3})
\end{align*}
and
\begin{align*}
  \|\gamma_*^c\|_{X^4}=\cO(1),\qquad \|\cK_\rS \gamma_*^c \cK_\rS \|_{X^4}=\cO(c^{-2}).  
\end{align*}
\end{lemma}
\begin{proof}
We begin from $\cD_{\gamma_*^c}u_{c,n}=\lambda_n^c u_{c,n}$. Taking $(-\Delta)^{1/2}$ to both side of this eigenvalue equation and proceeding as for Lemma \ref{lem:u-dec}, we know
\begin{align*}
    \|(-\Delta)^{1/2} \cD u_{c,n}\|_{\cH}=\|(-\Delta)^{1/2} (\lambda_n^c+V-W_{\gamma})u_{c,n}\|_{\cH}.
\end{align*}
Then arguing as for \eqref{eq:K1}, the estimate $\|u_{c,n}\|_{H^2}=\cO(1)$ follows immediately from Assumption \ref{ass:V}, \eqref{eq:W1-estimate} and \eqref{eq:5.19}. 

Next, multiplying \eqref{eq:decom-Dirac} by $(-\Delta)^{1/2}$ again, we get
\begin{align}\label{eq:decom-Dirac'}
    \begin{pmatrix}
    c\mathcal{L} (-\Delta)^{1/2} u_{c,n}^{\rm S}\\
     c\mathcal{L} (-\Delta)^{1/2} u_{c,n}^{\rm L}
\end{pmatrix} +(-\Delta)^{1/2}(-V+W_{\gamma_*^c}) u_{c,n} =\begin{pmatrix}
    (\lambda_n^c-c^2)(-\Delta)^{1/2}u_{c,n}^{\rm L}\\
    (\lambda_n^c+c^2)(-\Delta)^{1/2} u_{c,n}^{\rm S}
\end{pmatrix}.
\end{align}
Then arguing as for \eqref{eq:K2}, \eqref{eq:K3}, and by Assumption \ref{ass:V}, \eqref{eq:W1-estimate} and \eqref{eq:5.19}, we get the remaining estimates in this lemma.

\end{proof}

\subsection{Proof of \texorpdfstring{\eqref{eq:EHF-gammac}}{}}\label{sec:6.2}
We now turn to the proof of \eqref{eq:EHF-gammac}. We use \eqref{eq:EDF-EHF2} with $\gamma=\cK_\rL \gamma_*^c \cK_\rL$.

\medskip

\subsubsection{Kinetic term}\label{sec:6.3.1} Note that
\begin{align*}
  \MoveEqLeft  \lc\begin{pmatrix}
        u_{c,m}^{\rm L}\\u_{c,m}^{\rm S}
    \end{pmatrix}, (\cD-c^2) \begin{pmatrix}
        u_{c,n}^{\rm L}\\u_{c,n}^{\rm S}
    \end{pmatrix}\rc_{\cH}\\
    &= \lc u_{c,m}^{\rm L},c\cL u_{c,n}^{\rm S}\rc_{L^2(\R^3;\C^2)}+\lc u_{c,m}^{\rm S}, c\cL u_{c,n}^{\rm L}-2c^2 u_{c,n}^{\rm S}\rc_{L^2(\R^3;\C^2)}.
\end{align*}
Concerning the terms on the right-hand side, as $\cL^2=-\Delta$, we have
\begin{align*}
 \MoveEqLeft   \lc u_{c,m}^{\rm L},c\cL u_{c,n}^{\rm S}\rc_{L^2(\R^3;\C^2)}\\
 &=\lc u_{c,m}^{\rm L},-\frac{1}{2}\Delta u_{c,n}^{\rm L}\rc_{L^2(\R^3;\C^2)}+\lc u_{c,m}^{\rm L},\;c\cL \left(u_{c,n}^{\rm S}-\frac{1}{2c}\cL u_{c,n}^{\rm L}\right)\rc_{L^2(\R^3;\C^2)},
\end{align*}
and
\begin{align*}
  \MoveEqLeft \lc u_{c,m}^{\rm S}, c\cL u_{c,n}^{\rm L}-2c^2 u_{c,n}^{\rm S}\rc_{L^2(\R^3;\C^2)}\\
   &= \lc \frac{1}{2c}\cL u_{c,m}^{\rm L}, \left(c\cL u_{c,n}^{\rm L}- 2c^2u_{c,n}^{\rm S}\right)\rc_{L^2(\R^3;\C^2)}\\
    &\quad +\lc \left(u_{c,m}^{\rm S}-\frac{1}{2c}\cL u_{c,m}^{\rm L}\right), \left(c\cL u_{c,n}^{\rm L}- 2c^2u_{c,n}^{\rm S}\right)\rc_{L^2(\R^3;\C^2)}\\
    &=-\lc u_{c,m}^{\rm L},\;c\cL \left(u_{c,n}^{\rm S}-\frac{1}{2c}\cL u_{c,n}^{\rm L}\right)\rc_{L^2(\R^3;\C^2)}\\
    &\quad -2c^2 \lc \left(c\cL u_{c,n}^{\rm L}- 2c^2u_{c,n}^{\rm S}\right), \left(\frac{1}{2c}\cL u_{c,n}^{\rm L}- u_{c,n}^{\rm S}\right)\rc_{L^2(\R^3;\C^2)}.
\end{align*}
Then using above estimates,
\begin{align}\label{eq:D-c2-uL-S}
\MoveEqLeft \lc\begin{pmatrix}
        u_{c,m}^{\rm L}\\u_{c,m}^{\rm S}
    \end{pmatrix}, (\cD-c^2) \begin{pmatrix}
        u_{c,n}^{\rm L}\\u_{c,n}^{\rm S}
    \end{pmatrix}\rc_{\cH}=\lc u_{c,m}^{\rm L},H_0 u_{c,n}^{\rm L}\rc_{L^2(\R^3;\C^2)}\notag\\
    &\qquad\qquad\qquad-2c^2 \lc \left(c\cL u_{c,n}^{\rm L}- 2c^2u_{c,n}^{\rm S}\right), \left(\frac{1}{2c}\cL u_{c,n}^{\rm L}- u_{c,n}^{\rm S}\right)\rc_{L^2(\R^3;\C^2)}.
\end{align}
By Lemma \ref{lem:u-dec},
\begin{align}\label{eq:H0-D-c2-gammac-L}
\MoveEqLeft  \Big| \Tr_{\cH}[(\cD-c^2)\gamma_*^c] -\Tr_\cH[H_0 \cK_\rL  \gamma_*^c \cK_\rL]\Big|\notag\\
&= \left|\sum_{n=1}^\infty\mu_n\left<u_{c,n},(\cD-c^2)u_{c,n}\right>_\cH- \sum_{n=1}^\infty\mu_n \lc u_{c,n}^{\rm L},H_0 u_{c,n}^{\rm L}\rc_{L^2(\R^3;\C^2)}\right|\notag\\
   &\leq  2c^2\sum_{n=1}^\infty\mu_n \left\|u_{c,n}^{\rS}-\frac{1}{2c}\cL u_{c,n}^{\rL}\right\|_{L^2(\R^3;\C^2)}^2 \leq  2K_3^2q c^{-4}.
\end{align}

\subsubsection{Potential term} By Lemma \ref{lem:u-dec} and Kato's inequality,
\begin{align}\label{eq:V-L-S}
   \Big| \Tr_{\cH}[V\gamma_*^c]-\Tr_\cH[V\cK_\rL  \gamma_*^c \cK_\rL]\Big|=\Big|\Tr[V\cK_\rS\gamma_*^c\cK_\rS]\Big|\leq \frac{\pi}{2} z \|\cK_\rS\gamma_*^c\cK_\rS\|_{X}\leq \frac{\pi}{2c^2}K_2^2qz.
\end{align}
Analogously,
\begin{align}\label{eq:W1-L-S}
   \MoveEqLeft  \Big| \Tr_{\cH}[W_{1,\gamma_*^c}\gamma_*^c]-\Tr_\cH[W_{1,\cK_\rL  \gamma_*^c \cK_\rL} \cK_\rL  \gamma_*^c \cK_\rL]\Big|\notag\\
     &= \Big|\Tr[W_{\cK_\rS  \gamma_*^c \cK_\rS}\cK_\rS\gamma_*^c\cK_\rS]\Big|\leq \frac{\pi}{2}  \|\cK_\rS\gamma_*^c\cK_\rS\|_{\mathfrak{S}_1}\|\cK_\rS\gamma_*^c\cK_\rS\|_{X}\leq \frac{\pi }{2c^2}K_2^2q^2.
\end{align}
Concerning the term associated with $W_{2,\bullet}$, by \eqref{eq:Wpotential-decom}, H\"older's inequality \eqref{eq:Holder}, \eqref{eq:5.1-1'} and Lemma \ref{lem:u-dec},
\begin{align}\label{eq:W2-L-S}
\MoveEqLeft  \Big|\Tr_{\cH}[W_{2,\gamma_*^c}\gamma_*^c] -\Tr_\cH[W_{2,\cK_\rL \gamma_*^c \cK_\rL}\cK_\rL \gamma_*^c \cK_\rL]\Big|\notag\\
&\leq \sum_{j\in\{\rL,\rS\}}\Big|\Tr_{\cH}[W_{2,\cK_j\gamma_*^c\cK_\rS}\cK_\rS\gamma_*^c\cK_j]\Big|+\Big|\Tr_{\cH}[W_{2,\cK_\rS\gamma_*^c\cK_\rL}\cK_\rL\gamma_*^c\cK_\rS]\Big|\notag\\
&\leq 2\|\cK_\rS\gamma_*^c\cK_\rL\|_{X}^2+\|\cK_\rS\gamma_*^c\cK_\rS\|_{X}^2\leq \frac{1}{c^2}(1+2K_1^2)K_1^2q^2.
\end{align}

\subsubsection{Conclusion} From above estimates, we know
\begin{align*}
    \cE_c(\gamma_*^c)=\cE^{\rm HF}(\cK_{\rm L}\gamma_*^c \cK_{\rm L})+\mathcal{O}(c^{-2}).
\end{align*}
As $\gamma_*^c$ is a DF minimizer, we get
\begin{align*}
    \cE^{\rm HF}(\cK_\rL \gamma_*^c \cK_\rL)\leq \cE_c(\gamma_*^c)+\cO(c^{-2}) =E_{c,q}+\mathcal{O}(c^{-2}).
\end{align*}
This ends the proof of \eqref{eq:EHF-gammac}.

\subsection{Proof of \texorpdfstring{\eqref{eq:EHF-gammac-c2}}{}}\label{sec:6.3}
We are now in the position to prove \eqref{eq:EHF-gammac-c2}. To do so, we first study the properties of $\widetilde{\gamma}_*^c$, and then we use \eqref{eq:EDF-EHF2} to split the calculation into three parts. For further convenience, we assume $\delta_c=0$ throughout this subsection. 

\subsubsection{Construction of \texorpdfstring{$\widetilde{\gamma}_*^c$}{}}
In this subsection, we introduce $\widetilde{\gamma}_*^c$ and study its properties under Assumption \ref{ass:V} and $\delta_c=0$. In this case, we can write
\begin{align*}
    \gamma_*^c:= \sum_{n= 1}^q |u_{c,n} \rangle\,\langle u_{c,n}|
\end{align*}
with $\left<u_{c,m},u_{c,n}\right>_\cH=\delta_{m,n}$. To get the term $\cE^{(2)}_c(\cK_\rL\gamma_*^c\cK_\rL )$ in \eqref{eq:EHF-gammac-c2}, we need to modify the density matrix $\gamma_*^c$. We will use the projection on the space spanned by $\{\cK_\rL u_{c,n}\}_{1\leq n\leq q}$. Note that
\begin{align*}
 \left<\cK_\rL u_{c,m},\cK_\rL u_{c,n}\right>_{\cH}= \delta_{m,n}- \left<\cK_\rS u_{c,m},\cK_\rS u_{c,n}\right>_\cH.
\end{align*}
We introduce now the overlap matrix on $\{\cK_\rL u_{c,n}\}_{1\leq n\leq q}$:
\begin{align}
    S_{\rm DF}:&= \Big(\left<\cK_\rL u_{c,m},\cK_\rL u_{c,n}\right>_{\cH}\Big)_{1\leq m,n\leq q}=\1_{q\times q}-\frac{1}{4c^2} \widetilde{S}_{\rm DF},\label{eq:S-DF}\\
     \widetilde{S}_{\rm DF}:&=\Big(4c^2 \left<\cK_\rS u_{c,m},\cK_\rS u_{c,n}\right>_\cH\Big)_{1\leq m,n\leq q}.\label{eq:S'-DF}
\end{align}
For $c$ large enough, by Lemma \ref{lem:u-dec'},
\begin{align}\label{eq:diagonally dominated1'}
\sup_{1\leq m,n\leq q} \left|\left<\cK_\rS u_{c,m},\cK_\rS u_{c,n}\right>_\cH\right|=\cO(c^{-2}).
\end{align}
 This implies $S_{\rm DF}$ is a strictly diagonally dominated matrix, thus $S_{\rm DF}$ is invertible, 
\begin{align}\label{eq:S-1'}
    \left\|S_{\rm DF}^{-1}-\left(\1_{q\times q}+\frac{1}{4c^2}\widetilde{S}_{\rm DF}\right)\right\|_{\mathfrak{S}_1(\C^q)}=\cO(c^{-4})
\end{align}
and there exists some constants $0<C_0'<1<C_1'$ such that for any $c$ large enough,
\begin{align}\label{eq:S-1-bound'}
    C_0'\1_{q\times q}\leq S_{\rm DF}^{-1}\leq C_1'\1_{q\times q}
\end{align}
in the sense of operator. The density matrix $\widetilde{\gamma}_c^{\rm HF}$ on the space spanned by $\{\cS \cK_\rL u_{c,n}\}_{1\leq n\leq q}$ is defined by 
\begin{align}\label{eq:gammaDF'}
    \widetilde{\gamma}_*^c:&=\Big(\left|\cK_\rL u_{c,1}\right>,\cdots, \left|\cK_\rL u_{c,q}\right>\Big) S_{\rm DF}^{-1} \begin{pmatrix}
        \left<\cK_\rL u_{c,1}\right|\\
        \vdots\\
        \left<\cK_\rL u_{c,q}\right|
    \end{pmatrix}.
\end{align}

\subsubsection{Property of \texorpdfstring{$\widetilde{\gamma}_*^c$}{}}
Analogous to Lemma \ref{lem:gammaHF'-projection} and Lemma \ref{lem:gammaHF'-estimate}, we have
\begin{lemma}\label{lem:gammaDF'-projection}
For $c$ large enough, $\widetilde{\gamma}_*^c\in \cB(\cH;\cH_\rL)$ is a projector with ${\rm Rank}(\widetilde{\gamma}_*^c)=q$.
\end{lemma}
and
\begin{lemma}\label{lem:gammaDF'-estimate}
Under Assumption \ref{ass:V}, for $c$ large enough, $ \|\widetilde{\gamma}_*^c\|_{X^4}=\cO(1)$.
\end{lemma}

Next, by Lemma \ref{lem:u-dec'}, 
\begin{align*}
    \left<\cK_\rS u_{c,m}, \cK_\rS u_{c,n}\right>_\cH= \frac{1}{4c^2}\left<\cL \cK_\rL u_{c,m}, \cL \cK_\rL u_{c,n}\right>_\cH+\cO(c^{-4}).
\end{align*}
Let
\begin{align*}
    \widetilde{\widetilde{S}}_{\rm DF}:=\Big(\left<\cL \cK_\rL u_{c,m}, \cL \cK_\rL u_{c,n}\right>_\cH\Big)_{1\leq m,n\leq q},
\end{align*}
Thus
\begin{align*}
    \|\widetilde{S}_{\rm DF}-\widetilde{\widetilde{S}}_{\rm DF}\|_{\cB(\C^q)}=\cO(c^{-2}).
\end{align*}
This implies
\begin{align}
    S_{\rm DF}^{-1}=\1_{q\times  q}+\frac{1}{4c^2}\widetilde{S}_{\rm DF}+\cO(c^{-4}) =\1_{q\times  q}+\frac{1}{4c^2}\widetilde{\widetilde{S}}_{\rm DF}+\cO(c^{-4}).
\end{align}
Now the density matrix $\widetilde{\gamma}_*^c$ can be further approximated by  $ \widetilde{\widetilde{\gamma}}_*^c$ which is defined by
\begin{align}\label{eq:gammaDF''}
   \widetilde{\widetilde{\gamma}}_*^c:&=\Big(\left|\cK_\rL u_{c,1}\right>,\cdots, \left|\cK_\rL u_{c,q}\right>\Big)\left(\1_{q\times q}+ \frac{1}{4c^2}\widetilde{\widetilde{S}}_{\rm DF}\right)\begin{pmatrix}
        \left<\cK_\rL u_{c,1}\right|\\
        \vdots\\
        \left<\cK_\rL u_{c,q}\right|
    \end{pmatrix}\notag\\
    &= \cK_\rL \gamma_*^c\cK_\rL  +\frac{1}{4c^2}\sum_{1\leq m,n\leq q}  \left<\cL \cK_\rL u_{c,m},\cL \cK_\rL u_{c,n}\right>_\HL\left|\cK_\rL u_{c,m}\right>\left<\cK_\rL u_{c,n}\right|.
\end{align}
Analogous to the proof of Lemma \ref{lem: gammaHF''-gammaHF'}, by Lemma \ref{lem:u-dec'} we get
\begin{lemma}\label{lem: gammaDF''-gammaDF'}
    For $c$ large enough, we have
    \begin{align*}
        \|\widetilde{\widetilde{\gamma}}_*^c- \widetilde{\gamma}_*^c\|_{X^4}=\cO(c^{-4}).
    \end{align*}
\end{lemma}

\subsubsection{End of the proof}
Estimate \eqref{eq:EHF-gammac-c2} follows from the following lemma.
\begin{lemma}
    For $c$ large enough,
    \begin{align*}
        \cE^{\rm HF}(\widetilde{\gamma}_*^c)=\cE^{\rm HF}(\widetilde{\widetilde{\gamma}}_*^c)+\cO(c^{-4})=E_{c,q}-\widetilde{\cE}_c^{(2)}(\gamma_*^c)+\cO(c^{-4}).
    \end{align*}
\end{lemma}

\begin{proof}
We have
    \begin{align*}
        \cE^{\rm HF}(\widetilde{\gamma}_*^c)-\cE^{\rm HF}(\widetilde{\widetilde{\gamma}}_*^c)=\Tr_\cH\Big[H_{0} (\widetilde{\gamma}_*^c- \widetilde{\widetilde{\gamma}}_*^c)\Big]+\Tr_\cH\left[\left(-V+\frac{1}{2}W_{\widetilde{\gamma}_*^c+\widetilde{\widetilde{\gamma}}_*^c}\right)(\widetilde{\gamma}_*^c-\widetilde{\widetilde{\gamma}}_*^c)\right].
    \end{align*}
Then analogous to the proof of Lemma \ref{lem:3.13}, by Lemma \ref{lem: gammaDF''-gammaDF'}, Kato's inequality and \eqref{eq:5.1-1}, we infer that
    \begin{align*}
        \left| \cE^{\rm HF}(\widetilde{\gamma}_*^c)-\cE^{\rm HF}(\widetilde{\widetilde{\gamma}}_*^c)\right|\lesssim \|\widetilde{\gamma}_*^c - \widetilde{\widetilde{\gamma}}_*^c\|_{X^2}=\cO(c^{-4}).
    \end{align*}

\medskip

Now we prove $\cE^{\rm HF}(\widetilde{\widetilde{\gamma}}_*^c)=E_{c,q}+\widetilde{\cE}_c^{(2)}(\cK_\rL \gamma_*^c \cK_\rL)+\cO(c^{-4})$. Note that
\begin{align*}
    \cE^{\rm HF}(\gamma+h)=\cE^{\rm HF}(\gamma)+\Tr_\HL\big[H_{0,\gamma}h\big]+\frac{1}{2}\Tr_\HL\big[W_h h\big].
\end{align*}
Then according to the definition of $\widetilde{\widetilde{\gamma}}_*^c$ (see \eqref{eq:gammaDF''}), by the Taylor's expansion of the HF functional,
\begin{align}\label{eq:7.28}
 \cE^{\rm HF}(\widetilde{\widetilde{\gamma}}_*^c)&=\cE^{\rm HF}(\cK_\rL \gamma_*^c\cK_\rL  )\notag\\
  &\quad+\Tr_{\cH_\rL}[H_{0,\cK_\rL \gamma_*^c\cK_\rL}(\widetilde{\widetilde{\gamma}}_*^c-\cK_\rL \gamma_*^c\cK_\rL)]+\frac{1}{2}\Tr_{\cH_\rL}[W_{\widetilde{\widetilde{\gamma}}_*^c-\cK_\rL \gamma_*^c\cK_\rL}(\widetilde{\widetilde{\gamma}}_*^c-\cK_\rL \gamma_*^c\cK_\rL)] \notag\\
  &=\cE^{\rm HF}(\cK_\rL \gamma_*^c\cK_\rL  )+\Tr_{\cH_\rL}[H_{0,\cK_\rL \gamma_*^c\cK_\rL}(\widetilde{\widetilde{\gamma}}_*^c-\cK_\rL \gamma_*^c\cK_\rL)]+\cO(c^{-4})\notag\\
  &=\frac{1}{4c^2}\sum_{1\leq m,n\leq q}\left<\cL \cK_\rL u_{c,m},\cL \cK_\rL u_{c,n}\right>_\cH\left<\cK_\rL u_{c,n},H_{0, \cK_\rL \gamma_*^c \cK_\rL}\cK_\rL u_{c,m} \right>_\cH\notag\\
  &\quad+\cE^{\rm HF}(\cK_\rL \gamma_*^c \cK_\rL)+\cO(c^{-4}),
\end{align}
where in the second estimate, we used \eqref{eq:5.1-1} and Lemma \ref{lem:u-dec}. To end the proof, we need to study 
\begin{align*}
    I:&=\left<\cK_\rL u_{c,n},H_{0, \cK_\rL \gamma_*^c \cK_\rL}\cK_\rL u_{c,m} \right>_\cH,\\
    II:&=\cE^{\rm HF}(\cK_\rL \gamma_*^c \cK_\rL).
\end{align*}
We study them separately.

\medskip

\noindent{\bf Estimate on $I$.} We now consider the term $\left<\cK_\rL u_{c,n},H_{0, \cK_\rL \gamma_*^c \cK_\rL}\cK_\rL u_{c,m} \right>_\cH$ on the right-hand side of \eqref{eq:7.28}. According to \eqref{eq:D-c2-uL-S} and Lemma \ref{lem:u-dec}, we have
\begin{align*}
    \left<\cK_\rL u_{c,n},H_0\cK_\rL u_{c,m} \right>_\cH=  \left<u_{c,n},(\cD-c^2)u_{c,m} \right>_\cH+\cO(c^{-2}).
\end{align*}
Concerning the potential terms, analogous to \eqref{eq:V-L-S}-\eqref{eq:W2-L-S} we have
\begin{align*}
     \left<\cK_\rL u_{c,n},(-V+W_{\cK_\rL \gamma_*^c \cK_\rL}) \cK_\rL u_{c,m} \right>_\cH - \left<u_{c,n},(-V+W_{ \gamma_*^c}) u_{c,m} \right>_\cH=\cO(c^{-2}).
\end{align*}
Thus,
\begin{align*}
    \left<\cK_\rL u_{c,n},H_{0, \cK_\rL \gamma_*^c \cK_\rL}\cK_\rL u_{c,m} \right>_\cH=\left< u_{c,n},(\cD_{ \gamma_*^c }-c^2)u_{c,m} \right>_\cH+\cO(c^{-2}).
\end{align*}
Then as $\left<u_{c,n},\cD_{\gamma_*^c}u_{c,m}\right>_{\cH}=\lambda_n\delta_{n,m}$, we get
\begin{align}\label{eq:estimate-I}
    \cE^{\rm HF}(\widetilde{\widetilde{\gamma}}_*^c)&=\cE^{\rm HF}(\cK_\rL \gamma_*^c \cK_\rL)\notag\\
    &+\frac{1}{4c^2}\sum_{1\leq m,n\leq q}\left<\cL \cK_\rL u_{c,m},\cL \cK_\rL u_{c,n}\right>_\cH\left<u_{c,n},(\cD_{\gamma_*^c} -c^2)u_{c,m} \right>_\cH+\cO(c^{-4})\notag\\
    &=\cE^{\rm HF}(\cK_\rL \gamma_*^c \cK_\rL) +\frac{1}{4c^2}\Tr_{\cH}[(\cD_{\gamma_*^c} -c^2)\gamma_*^c \cL^2]  +\cO(c^{-4}).
\end{align}

\medskip

\noindent{\bf Estimate on $II$.} We now consider the term $\cE^{\rm HF}(\cK_\rL \gamma_*^c \cK_\rL) $. By \eqref{eq:H0-D-c2-gammac-L}, for the kinetic term, we have
\begin{align}\label{eq:7.30}
    \Tr_\cH[H_0 \cK_\rL \gamma_*^c \cK_\rL]=\Tr_\cH[(\cD-c^2)\gamma_*^c]+\cO(c^{-4}).
\end{align}

Concerning the potential between electrons and nuclei, by Lemma \ref{lem:u-dec},
\begin{align*}
    \Tr_\cH[-V \cK_\rL \gamma_*^c \cK_\rL]&=\Tr_\cH[-V\gamma_*^c]-\Tr[-V \cK_\rS \gamma_*^c \cK_\rS]\\
    &=\Tr_\cH[-V\gamma_*^c]-\Tr_\cH[-V \cK_\rS \cS\gamma_*^c \cS^* \cK_\rS]+\Tr_\cH[-V \cK_\rS (\cS\gamma_*^c \cS^* -\gamma_*^c) \cK_\rS].
\end{align*}
By Lemma \ref{lem:u-dec'} and Hardy's inequality,
\begin{align*}
 \MoveEqLeft  \Big| \Tr[V \cK_\rS (\cS\gamma_*^c \cS^* -\gamma_*^c) \cK_\rS]\Big|\\
 &=\left|\sum_{n=1}^q \left<(u_{c,n}^\rS +\frac{1}{2c}\cL u_{c,n}^\rL), V(u_{c,n}^\rS -\frac{1}{2c}\cL u_{c,n}^\rL)\right>_{L^2(\R^3;\C^2)}\right|\\
   &\lesssim \sum_{n=1}^q \left(\|u_{c,n}^\rS\|_{H^1(\R^3;\C^2)}+\frac{1}{2c}\|u_{c,n}^\rL\|_{H^2(\R^3;\C^2)}\right)\left\|u_{c,n}^\rS -\frac{1}{2c}\cL u_{c,n}^\rL\right\|_{L^2(\R^3;\C^2)}=\cO(c^{-4}).
\end{align*}
Thus,
\begin{align}\label{eq:7.31}
     \Tr_\cH[-V \cK_\rL \gamma_*^c \cK_\rL]&=\Tr_\cH[-V\gamma_*^c]-\Tr_\cH[-V \cK_\rS \cS\gamma_*^c \cS^* \cK_\rS]+\cO(c^{-4})\notag\\
    &=\Tr_\cH[-V\gamma_*^c]-\frac{1}{4c^2}\Tr_\cH[-V \cL\cK_\rL \gamma_*^c  \cK_\rL \cL]+\cO(c^{-4}).
\end{align}

Concerning the potential between electrons and electrons, we write $W_{\bullet}=W_{1,\bullet}-W_{2,\bullet}$. For the term associated with $W_{1,\bullet}$, analogous to above estimate, by \eqref{eq:V-K1-K2} and Lemma \ref{lem:u-dec},
\begin{align}\label{eq:7.32}
   \MoveEqLeft \frac{1}{2}\Tr_\cH[W_{1,\cK_\rL \gamma_*^c \cK_\rL} \cK_\rL \gamma_*^c \cK_\rL]\notag\\
   &=\frac{1}{2}\Tr_\cH[W_{1,\gamma_*^c}\gamma_*^c]-\Tr_\cH[W_{1,\cK_\rL \gamma_*^c \cK_\rL} \cK_\rS \gamma_*^c \cK_\rS]- \frac{1}{2}\Tr_\cH[W_{1,\cK_\rS \gamma_*^c \cK_\rS} \cK_\rS \gamma_*^c \cK_\rS]\notag\\
    &=\frac{1}{2}\Tr_\cH[W_{1,\gamma_*^c}\gamma_*^c]-\Tr_\cH[W_{1,\cK_\rL \gamma_*^c \cK_\rL}\cK_\rS \cS\gamma_*^c \cS^* \cK_\rS]+\cO(c^{-4})\notag\\
    &=\frac{1}{2}\Tr_\cH[W_{1,\gamma_*^c}\gamma_*^c]-\frac{1}{4c^2}\Tr_\cH[W_{1,\cK_\rL \gamma_*^c \cK_\rL}\cL\cK_\rL \gamma_*^c\cK_\rL \cL]+\cO(c^{-4}).
\end{align}
Next, we consider the term associated with $W_{2,\bullet}$. By \eqref{eq:W-K1-K2}, we have
\begin{align*}
   \MoveEqLeft \frac{1}{2}\Tr_\cH[W_{2,\cK_\rL \gamma_*^c \cK_\rL}\cK_\rL \gamma_*^c \cK_\rL]\\
   &=\frac{1}{2}\Tr_\cH[W_{2, \gamma_*^c } \gamma_*^c ]-\Tr_\cH[W_{2, \cK_\rL\gamma_*^c\cK_\rS } \cK_\rS \gamma_*^c \cK_\rL]-\frac{1}{2}\Tr_\cH[W_{ 2,\cK_\rS\gamma_*^c\cK_\rS } \cK_\rS \gamma_*^c \cK_\rS]\\
    &=\frac{1}{2}\Tr_\cH[W_{2, \gamma_*^c } \gamma_*^c ]-\Tr_\cH[W_{ 2,\cK_\rL\gamma_*^c\cK_\rS } \cK_\rS \gamma_*^c \cK_\rL]+\cO(c^{-4})
\end{align*}
where in the last estimate we used Lemma \ref{lem:W2-estimate} and Lemma \ref{lem:u-dec}. Now we try to replace the term $\Tr_\cH[W_{ 2,\cK_\rL\gamma_*^c\cK_\rS } \cK_\rS \gamma_*^c \cK_\rL]$ in above estimate by $\Tr_\cH[W_{ 2,\cK_\rL\gamma_*^c \cS^* \cK_\rS } \cK_\rS\cS \gamma_*^c \cK_\rL]$: by H\"older's inequality \eqref{eq:Holder} in addition,
\begin{align*}
   \MoveEqLeft \Big|\Tr_\cH[W_{2, \cK_\rL\gamma_*^c\cK_\rS } \cK_\rS \gamma_*^c \cK_\rL] - \Tr_\cH[W_{ 2,\cK_\rL\gamma_*^c \cS^* \cK_\rS } \cK_\rS\cS \gamma_*^c \cK_\rL] \Big|\\
    &\leq \Big| \Tr_\cH[W_{ 2,\cK_\rL\gamma_*^c(1-\cS^*)\cK_\rS } \cK_\rS \gamma_*^c \cK_\rL]\Big|+\Big|\Tr_\cH[W_{2, \cK_\rL\gamma_*^c\cS^*\cK_\rS } \cK_\rS (1-\cS)\gamma_*^c \cK_\rL]\Big|\\
    &\lesssim \left(\|\cK_\rS \gamma_*^c \cK_\rL\|_{X^2}+\|\cK_\rL \gamma_*^c \cS^*\cK_\rS\|_{X^2} \right)\|\cK_\rL \gamma_*^c (1-\cS^*)\cK_\rS\|_{\mathfrak{S}_1}\\
    &\lesssim \left(\|\cK_\rS \gamma_*^c \cK_\rS\|_{X^2}^{1/2}+c^{-1}\|\cL \cK_\rL \gamma_*^c \cK_\rL \cL\|_{X^2}^{1/2} \right) \|\cK_\rL \gamma_*^c \cK_\rL\|_{X^2} \|\cK_\rS (1-\cS)\gamma_*^c (1-\cS^*)\cK_\rS\|_{\mathfrak{S}_1}^{1/2}\\
    &\lesssim c^{-1}\|\cK_\rS (1-\cS)\gamma_*^c (1-\cS^*)\cK_\rS\|_{\mathfrak{S}_1}^{1/2}\\
    &\lesssim  c^{-1}\left(\sum_{n=1}^q \left\|u_{c,n}^\rS -\frac{1}{2c}\cL u_{c,n}^\rL\right\|_{H^1(\R^3;\C^4)}^2\right)^{1/2}=\cO(c^{-4}).
\end{align*}
Thus,
\begin{align}\label{eq:7.33}
   \MoveEqLeft \frac{1}{2}\Tr_\cH[W_{2,\cK_\rL \gamma_*^c \cK_\rL}\cK_\rL \gamma_*^c \cK_\rL]\notag\\
   &=\frac{1}{2}\Tr_\cH[W_{2, \gamma_*^c } \gamma_*^c ]-\Tr_\cH[W_{2, \cK_\rL\gamma_*^c \cS^*\cK_\rS } \cK_\rS\cS \gamma_*^c \cK_\rL]+\cO(c^{-4})\notag\\
    &=\frac{1}{2}\Tr_\cH[W_{2, \gamma_*^c } \gamma_*^c ]-\frac{1}{4c^2}\Tr_\cH[W_{2, \cK_\rL\gamma_*^c \cK_\rL \cL } \cL \cK_\rL \gamma_*^c \cK_\rL]+\cO(c^{-4}).
\end{align}
By \eqref{eq:7.30}-\eqref{eq:7.33}, we get
\begin{align}\label{eq:estimate-II}
    \cE^{\rm HF}(\cK_\rL \gamma_*^c \cK_\rL)&= \cE_c(\gamma_*^c)-\frac{1}{4c^2}\Tr_\cH[(-V+W_{1,\cK_\rL \gamma_*^c\cK_\rL })\cL \cK_\rL\gamma_*^c \cK_\rL \cL]\notag\\
    &\quad +\frac{1}{4c^2}\Tr_\cH[W_{2,\cK_\rL \gamma_*^c \cK_\rL \cL}\cL \cK_\rL\gamma_*^c \cK_\rL]+\cO(c^{-4}).
\end{align}

\noindent{\bf Conclusion.} Thus, from \eqref{eq:estimate-I} and \eqref{eq:estimate-II}, we conclude that
\begin{align*}
    \cE^{\rm HF}(\widetilde{\widetilde{\gamma}}_*^c)
    &=\cE_c(\gamma_*^c) -\widetilde{\cE}_c^{(2)}(\gamma_*^c) +\cO(c^{-4})=E_{c,q}-\widetilde{\cE}_c^{(2)}(\gamma_*^c) +\cO(c^{-4})
\end{align*}
where we recall that 
\begin{align*}
    \widetilde{\cE}_c^{(2)}(\gamma_*^c)&=-\frac{1}{4c^2}\Tr_{\cH}[(\cD_{\gamma^c_*}-c^2)\gamma^c_*\cL^2]  \\
    &+\frac{1}{4c^2}\Big(\Tr_\cH[(-V+W_{1,\cK_\rL \gamma_*^c\cK_\rL })\cL \cK_\rL \gamma_*^c  \cK_\rL \cL]-\Tr_\cH[W_{2,\cK_\rL \gamma_*^c \cK_\rL \cL}\cL \cK_\rL \gamma_*^c \cK_\rL]\Big).
\end{align*}
Finally, from Lemma \ref{lem:gammaDF'-projection} and Lemma \ref{lem:gammaDF'-estimate}, we conclude that $\gamma\in \Gamma_q^{\rm HF}$. This ends the proof.
\end{proof}

\section{Proofs}\label{sec:8}
We are now in the position to prove Theorems \ref{th:rela-effect}, \ref{th:non-unfill}, \ref{th:rela-correction} and Proposition \ref{prop:rela-decomp}.
\subsection{Proof of Theorem \ref{th:rela-effect}}\label{sec:8.1}
From Theorem \ref{th:3.1} and Theorem \ref{th:4.1}, we infer that
\begin{align*}
   E_{c,q} \leq E_q^{\rm HF}+\cO(c^{-2}).
\end{align*}
To prove the inverse estimate, we study first $\cK_\rL\gamma_*^c \cK_\rL$. As $0\leq \cK_\rL\gamma_*^c \cK_\rL\leq \gamma_*^c$, we know $\cK_\rL\gamma_*^c \cK_\rL\in \Gamma_{q}$. By the definition of $\cK_\rL$, we have $\cK_\rL\gamma_*^c \cK_\rL \in \Gamma_{q}\cap \cB(\cH,\HL)\cap X^2$. From Theorem \ref{th:5.1}, we know that
\begin{align}\label{eq:L-gamma-L-min}
    E_q^{\rm HF}\leq  \cE^{\rm HF}(\cK_\rL \gamma^c_*\cL_\rL)\leq E_{c,q}+\cO(c^{-2})\leq E_q^{\rm HF}+\cO(c^{-2}).
\end{align}
This proves Theorem \ref{th:rela-effect}.

\subsection{Proof of Theorem \ref{th:non-unfill}}\label{sec:8.2}
As $q\leq z$, from Theorem \ref{th:min-DF}, we can find a set of DF minimizers  $(\gamma^{c}_*)_c$ of $E_{c,q}$ satisfying \eqref{eq:gamma-DF} and $\Tr_\cH[\gamma_*^c]=q$. We argue by contradiction: there exists a subsequence of $(\gamma^{c_j}_*)_{c_j}$ such that $0<\delta_{c_j}<\1_{\nu_{c_j}}(\mathcal{D}^{c_j}_{\gamma_*^{c_j}})$ for any $j\geq 1$ with $c_j\to \infty$ when $j\to \infty$. From Theorem \ref{th:rela-effect} and \eqref{eq:L-gamma-L-min}, $(\cK_\rL\gamma_*^{c_j}\cK_\rL)_j$ is a minimizing sequence of the HF minimum problem \eqref{eq:min-HF}. Thus, from the existence of HF minimizers (see e.g., \cite{lieb1977hartree}), up to subsequences,
   \begin{align*}
       \cK_\rL\gamma_*^{c_j}\cK_\rL\to \gamma_*^{\rm HF}\qquad\textrm{in}\quad X^2.
   \end{align*}
From Lemma \ref{lem:u-dec}, we infer that for $j\to \infty$,
\begin{align}\label{eq:7.2}
    \gamma_*^{c_j}\to \gamma_*^{\rm HF}\qquad\textrm{in}\quad X^2.
\end{align}
In addition, from \eqref{eq:lambda-c2}, we know $-e\leq v_{c_j}-c^2\leq 0$. Thus up to subsequences, there exists $-e\leq \nu_*\leq 0$ such that $\nu_{c_j}-c^2\to \nu_*$ as $j\to \infty$. Now we prove $\nu_*=\nu$ with $\nu$ being given in \eqref{eq:gamma-HF}. Note that $\nu_{c_j}$ is the $q$-th eigenvalue of $\mathcal{D}^{c_j}_{\gamma_*^{c_j}}$. According to min-max principle (see e.g. \cite{dolbeault2000eigenvalues}) and \eqref{eq:7.2},
\begin{align*}
 \nu_{c_j}-c^2=   \sigma_q^+(\mathcal{D}^{c_j}_{\gamma_*^{c_j}})-c^2=\sigma_q^+(\mathcal{D}^{c_j}_{\gamma_*^{\rm HF}})-c^2+o_{j\to \infty}(1).
\end{align*}
Here assumptions in \cite{dolbeault2000eigenvalues} are verified under Assumption \ref{ass:c} (see e.g., \cite[Lemma 3.6]{sere2023new}). Then according to \cite[Theorem 6.6 and Theorem 6.7]{thaller2013dirac} or \cite[Theorem 3 and Eq. (8)]{esteban2001nonrelativistic},
\begin{align*}
 \nu_{c_j}-c^2= \sigma_q^+(\mathcal{D}^{c_j}_{\gamma_*^{\rm HF}})-c^2=\sigma_q(H_{0,\gamma_*^{\rm HF}}) +o_{j\to \infty}(1)=\nu+o_{j\to \infty}(1).
\end{align*}
As $\nu_{c_j}-c^2 \to \nu_*$, we know $\nu_*=\nu$.

\medskip

Let
\begin{align*}
  d:=  \frac{1}{2}{\rm dist}\Big(\sigma(H_{0,\gamma^{\rm HF}_*})\setminus\{\nu\},\nu\Big)
\end{align*}
Then
\begin{align*}
  q= \Tr_\HL\Big[\1_{(-\infty,\nu+d]}(H_{0,\gamma_{*}^{\rm HF}})\Big],
\end{align*}
and there are $q$ eigenvalues of $H_{0,\gamma_*^{\rm HF}}$ in the interval $(-\infty,\nu+d]$. Thus for $j$ large enough, by the non-relativistic limit of eigenvalues of Dirac operators (see e.g., \cite[Chp. 6 and Theorem 6.7]{thaller2013dirac}) we know that for $c$ large enough, there are at most $q$ eigenvalues of $\mathcal{D}^{c_j}_{\gamma_*^{\rm HF}}$ in the interval $(0,c^2+\nu+\frac{3}{4}d]$. Thus, for $c$ large enough
\begin{align*}
    \Tr_\cH\Big[\1_{\left(0,\nu+\frac{1}{2}d\right]}(\mathcal{D}^{c_j}_{\gamma_*^{c_j}})\Big]\leq  \Tr_\cH\Big[\1_{(0,c^2+\nu+\frac{3}{4}d]}(\mathcal{D}^{c_j}_{\gamma_*^{\rm HF}})\Big]\leq q.
\end{align*}
However as $0<\delta_{c_j}<\1_{\nu_{c_j}}\left(\mathcal{D}^{c_j}_{\gamma_*^{c_j}}\right)$,
\begin{align*}
  q=\Tr_\cH[\gamma_*^{c_j}]<\Tr_\cH\Big[\1_{\left(0,c^2+\nu_{c_j}\right]}(\mathcal{D}^{c_j}_{\gamma_*^{c_j}})\Big] \leq  \Tr_\cH\Big[\1_{\left(0,c^2+\nu+\frac{1}{2}d\right]}(\mathcal{D}^{c_j}_{\gamma_*^{c_j}})\Big]\leq q.
\end{align*}
 which is impossible. Thus $\delta_c=0$.

\medskip

Now we can write
\begin{align*}
    \gamma_*^c=\sum_{n=1}^q \left|u_{c,n}\right>\left<u_{c,n}\right|
\end{align*}
with $\cD_{\gamma_*^c} u_{c,n}=\lambda_n^c u_{c,n}$, and we can use the orthonormal set $\{u_{c,1},\cdots,u_{c,q}\}$ to represent $\gamma_*^c$. In this sense, the orthonormal set $\{u_{c,1},\cdots, u_{c,q}\}$ is a minimizing sequence of $E_q^{\rm HF}$ in $(H^1)^q$. From \cite{lieb1977hartree}, we infer that as $c\to \infty$,
\begin{align*}
   u_{c,n} \to u_n\qquad \mbox{in } H^1.
\end{align*}
Then by \eqref{lem:u-dec}, as $c\to \infty$,
\begin{align*}
    \cK_\rL  u_{c,n} \to u_n\qquad \mbox{in } H^1.
\end{align*}
Next, from \cite[Corollary 6.5]{thaller2013dirac}, for $c\to \infty$, we have 
\begin{align*}
   \left\| \frac{1}{\cD_{\gamma_*^{\rm HF}}-c^2+i }-  \cK_\rL\frac{1}{H_{0,\gamma_*^{\rm HF}}+i } \cK_\rL\right\|_{\cB(\cH)}\to 0.
\end{align*}
Thus,
\begin{align*}
    \frac{1}{\lambda_n^c-c^2+i}&=\left<u_{c,n},\frac{1}{\cD_{\gamma_*^c}-c^2+i } u_{c,n}\right>_{\cH}=\left<u_{c,n,},\frac{1}{\cD_{\gamma_*^{\rm HF}}-c^2+i } u_{c,n}\right>_{\cH} +o_{c\to \infty}(1)\\
    &=\left<\cK_\rL u_{c,n,},\frac{1}{H_{0,\gamma_*^{\rm HF}}+i } \cK_\rL u_{c,n}\right>_{\HL}+o_{c\to \infty}(1) \\
    &=\left<u_{n},\frac{1}{H_{0,\gamma_*^{\rm HF}}+i } u_{n}\right>_{\HL}+o_{c\to \infty}(1)= \frac{1}{\lambda_n+i}+o_{c\to \infty}(1).
\end{align*}
This shows that $\lambda_n^c-c^2\to \lambda_n$ as $c\to \infty$. Thus this ends the proof.

 \subsection{Proof of Theorem \ref{th:rela-correction}}\label{sec:8.3}
Finally, we prove Theorem \ref{th:rela-correction}. According to Theorem \ref{th:3.1} and Theorem \ref{th:4.1}, we know that under Assumption \ref{ass:V} and for $c$ large enough,
\begin{align*}
    E_{c,q}\leq E_q^{\rm HF}+\cE^{(2)}_c(\gamma_*^{\rm HF})+\cO(c^{-4})
\end{align*}
where $\gamma_*^{\rm HF}$ is any HF minimizer of $E_q^{\rm HF}$. Thus
\begin{align*}
    E_{c,q}\leq E_q^{\rm HF}+\min_{\gamma_*^{\rm HF}\in \mathcal{G}_{\rm HF}}\cE^{(2)}_c(\gamma_*^{\rm HF})+\cO(c^{-4})
\end{align*}
where we recall that $\mathcal{G}_{\rm HF}:=\{\gamma_*^{\rm HF}\in \Gamma_{q}\cap \cB(\cH,\HL)\cap X^2;\; \gamma_*^{\rm HF} \mbox{is a HF minimizer of }E_q^{\rm HF}\}$.

Now we turn to prove
\begin{align}\label{eq:6.1}
  E_q^{\rm HF}\leq   E_{c,q}-\min_{\gamma_*^{\rm HF}\in \mathcal{G}_{\rm HF}} \cE^{(2)}_c(\gamma_*^{\rm HF})+o(c^{-2}).
\end{align}
From Theorem \ref{th:non-unfill}, we know $\delta_c=0$. Thus under Assumption \ref{ass:V}, estimate \eqref{eq:EHF-gammac-c2} holds with $\widetilde{\gamma}_*^c\in \Gamma_{q}^{\rm HF}$. Thus, from Theorem \ref{th:5.1}, under Assumption \ref{ass:V} and for $c$ large enough,
\begin{align*}
    E_q^{\rm HF}\leq  \cE^{\rm HF}(\widetilde{\gamma}^c_*)= E_{c,q}-\widetilde{\cE}_c^{(2)}(\gamma_*^c)+\cO(c^{-4}).
\end{align*}
To end the proof, it remains to show that
\begin{align}\label{eq:6.2}
    \min_{\gamma_*^{\rm HF}\in \mathcal{G}_{\rm HF}} \cE^{(2)}_c(\gamma_*^{\rm HF})\leq \widetilde{\cE}_c^{(2)}(\gamma_*^c) +o(c^{-2}).
\end{align}
Once \eqref{eq:6.2} is proven, 
\begin{align*}
     E_q^{\rm HF}+\min_{\gamma_*^{\rm HF}\in \mathcal{G}_{\rm HF}} \cE^{(2)}_c(\gamma_*^{\rm HF}) \leq E_q^{\rm HF}+\widetilde{\cE}_c^{(2)}(\gamma_*^c) +o(c^{-2}) \leq E_{c,q}+o(c^{-2}). 
\end{align*}
This and \eqref{eq:6.1} show
\begin{align*}
     E_q^{\rm HF}= E_{c,q}- \min_{\gamma_*^{\rm HF}\in \mathcal{G}_{\rm HF}} \cE^{(2)}_c(\gamma_*^{\rm HF})+o(c^{-2}).
\end{align*}
Thus the proof is completed.

\medskip

Now we prove \eqref{eq:6.2}. From Theorem \ref{th:non-unfill}, we write
\begin{align*}
    \gamma_*^c=\sum_{n=1}^q \left|u_{c,n}\right>\left<u_{c,n}\right|
\end{align*}
Then proceeding as for Theorem \ref{th:non-unfill}, there exists $\gamma_*^{\rm HF}\in \mathcal{G}_{\rm HF}$ such that as $c\to \infty$,
\begin{align*}
    \gamma_*^c  \to \gamma_*^{\rm HF},\quad \cK_\rL \gamma_*^c\cK_\rL \to \gamma_*^{\rm HF}, \quad \mbox{in }X^2
\end{align*}
and 
\begin{align}\label{eq:uL-uHF}
  \cK_\rL  u_{c,n}\to u_n\quad \mbox{in }H^1,\qquad \lambda_n^c\to \lambda_n.
\end{align}

Now we claim that as $c\to \infty$,
\begin{align*}
    4c^2\widetilde{\cE}_c^{(2)}(\gamma_*^c) \to 4c^2\cE^{(2)}_c(\gamma_*^{\rm HF}).
\end{align*}
Recall that
\begin{align*}
    4c^2\widetilde{\cE}_c^{(2)}(\gamma_*^c)&=-\sum_{1\leq n\leq q}(\lambda_n^c-c^2)\left<\cL \cK_\rL u_{c,n},\cL \cK_\rL u_{c,n}\right>_\cH \\
    &+\Tr_\cH[(-V+W_{1,\cK_\rL \gamma_*^c\cK_\rL })\cL \cK_\rL \gamma_*^c  \cK_\rL \cL]-\Tr_\cH[W_{2,\cK_\rL \gamma_*^c \cK_\rL \cL}\cL \cK_\rL \gamma_*^c \cK_\rL].
\end{align*}
We study it term by term. By Lemma \ref{lem:u-dec'} and \eqref{eq:uL-uHF},
\begin{align*}
  \left| \left<\cL \cK_\rL u_{c,n},\cL \cK_\rL u_{c,n}\right>_\cH -\left<\cL  u_{n},\cL \cK_\rL u_{n}\right>_\cH \right|\leq \|\cK_\rL u_{c,n}+u_n\|_{H^1}\|\cK_\rL u_{c,n}- u_n\|_{H^1}\to 0.
\end{align*}
Thus,
\begin{align*}
    \sum_{1\leq n\leq q}(\lambda_n^c-c^2)\left<\cL \cK_\rL u_{c,n},\cL \cK_\rL u_{c,n}\right>_\cH  \to \sum_{1\leq n\leq q}\lambda_n\left<\cL u_{n},\cL  u_{n}\right>_\cH.
\end{align*}
Next, by Hardy's inequality and Lemma \ref{lem:u-dec'},
\begin{align*}
   \MoveEqLeft \left|\Tr_\cH[V \cL \cK_\rL \gamma_*^c  \cK_\rL \cL]-\Tr_\cH[V \cL \gamma_*^{\rm HF} \cL]\right|\\
   &=\left|\sum_{n=1}^q\left<V(\cL u_{c,n}^{\rL} +\cL u_n^\rL), \cL u_{c,n}^{\rL} -\cL u_n^\rL\right>_{L^2(\R^3;\C^2)}\right|\\
   &\lesssim \sum_{n=1}^q\| u_{c,n}^{\rL} +u_n^\rL\|_{H^2(\R^3;\C^2)}\|u_{c,n}^{\rL} -u_n^\rL\|_{H^1(\R^3;\C^2)}\to 0.
\end{align*}
Concerning $W_{1,\bullet}$, we have
\begin{align*}
   \MoveEqLeft \left|\Tr_\cH[W_{1,\cK_\rL\gamma_*^c \cK_\rL} \cL \cK_\rL \gamma_*^c  \cK_\rL \cL]-\Tr_\cH[W_{1,\gamma_*^{\rm HF}} \cL \gamma_*^{\rm HF} \cL]\right|\\
   &\leq \Big|\Tr_\cH[W_{1,\cK_\rL\gamma_*^c \cK_\rL- \gamma_*^{\rm HF}} \cL \cK_\rL \gamma_*^c  \cK_\rL \cL]\Big|\\
   &\quad +\Big|\Tr_\cH[W_{1,\gamma_*^{\rm HF}} (\cL \cK_\rL \gamma_*^c \cK_\rL\cL- \cL \gamma_*^{\rm HF}\cL)]\Big|\\
   &\lesssim \|\cK_\rL\gamma_*^c \cK_\rL- \gamma_*^{\rm HF}\|_{X}\|\gamma_*^c\|_{X^2}+\|\gamma_*^{\rm HF}\|_{X}\|\cK_\rL\gamma_*^c \cK_\rL- \gamma_*^{\rm HF}\|_{X^2}\to 0.
\end{align*}
Concerning $W_{2,\bullet}$, by Lemma \ref{lem:W2-estimate}
\begin{align*}
\MoveEqLeft   \left| \Tr_\cH[W_{2,\cK_\rL \gamma_*^c \cK_\rL \cL}\cL \cK_\rL \gamma_*^c \cK_\rL] -\Tr_\cH[W_{2,\gamma_*^{\rm HF} \cL}\cL \gamma_*^{\rm HF}]\right|\\
&\leq \left| \Tr_\cH[\cL \cK_\rL \gamma_*^c \cK_\rL W_{2,\cK_\rL \gamma_*^c \cK_\rL \cL- \gamma_*^{\rm HF} \cL}]\right| +\left|\Tr_\cH[(\cL \cK_\rL \gamma_*^c \cK_\rL -\cL \gamma_*^{\rm HF})W_{2,\gamma_*^{\rm HF} \cL} ]\right|\\
&\lesssim \left(\|\cK_\rL \gamma_*^c \cK_\rL\|_{X^2}+\|\gamma_*^{\rm HF}\|_{X^2}\right)\|\cK_\rL \gamma_*^c \cK_\rL - \gamma_*^{\rm HF} \|_{X^2}\to 0.
\end{align*}
Thus we can conclude that
\begin{align*}
     4c^2\widetilde{\cE}_c^{(2)}(\gamma_*^c) = 4c^2\cE^{(2)}_c(\gamma_*^{\rm HF})+o_{c\to \infty}(1)
\end{align*}
for some $\gamma_*^{\rm HF}\in\mathcal{G}_{\rm HF}$. As a result, 
\begin{align*}
    \min_{\gamma\in \mathcal{G}_{\rm HF}} \cE^{(2)}_c(\gamma)\leq \cE^{(2)}_c(\gamma_*^{\rm HF}) = \widetilde{\cE}_c^{(2)}(\gamma_*^c) +o(c^{-2}).
\end{align*}
This proves \eqref{eq:6.2}. Now the proof is completed.

\subsection{Proof of Proposition \ref{prop:rela-decomp}}\label{sec:8.4}
We mainly focus on the term $- \sum_{n=1}^q \lambda_{n}^{\rm HF}\left<\cL u_{n}^{\rm HF},\cL  u_n^{\rm HF}\right>_\HL$. The others can be counterbalanced by reformulating this term. We have
\begin{align*}
 - \sum_{n=1}^q \lambda_{n}^{\rm HF}\left<\cL u_{n}^{\rm HF},\cL  u_n^{\rm HF}\right>_\HL&=  -\sum_{n=1}^q\Re \left<H_{0,\gamma_*^{\rm HF}} u_n^{\rm HF},  \cL^2u_n^{\rm HF}\right>_\HL\\
    &= -\sum_{n=1}^q\Re \left<(H_0-V+W_{1,\gamma_*^{\rm HF}}-W_{2,\gamma_*^{\rm HF}}) u_n^{\rm HF},  \cL^2u_n^{\rm HF}\right>_\HL.
\end{align*}
We study terms associated with $H_0, (-V+W_{1,\bullet})$ and $W_{2,\bullet}$ separately.

\medskip

\noindent{\bf Term with $H_0$.} Concerning $H_0$,
\begin{align}\label{eq:8.5}
     -\sum_{n=1}^q\Re \left<H_0 u_n^{\rm HF}, \cL^2 u_n^{\rm HF}\right>_\HL= -2\sum_{n=1}^q\left<u_n^{\rm HF}, H_0^2 u_n^{\rm HF}\right>=E_{\rm mv} .
\end{align}

\medskip

\noindent{\bf Term with $-V+W_{1,\bullet}$.} Next, we study the term associated with $-V+W_{1,\gamma_*^{\rm HF}}$:
\begin{align*}
   \MoveEqLeft  -\sum_{n=1}^q\Re \left<(-V+W_{1,\gamma_*^{\rm HF}}) u_n^{\rm HF}, \cL^2 u_n^{\rm HF}\right>_\HL\\
     &=-\sum_{n=1}^q\Re \left<\Big[\cL, (-V+W_{1,\gamma_*^{\rm HF}})\Big] u_n^{\rm HF}, \cL u_n^{\rm HF}\right>_\HL\\
     &\quad-\sum_{n=1}^q \left<\cL u_n^{\rm HF}, (-V+W_{1,\gamma_*^{\rm HF}})\cL u_n^{\rm HF}\right>_{\HL}
\end{align*}
Note that
\begin{align*}
\MoveEqLeft \Re \left<\Big[\cL, (-V+W_{1,\gamma_*^{\rm HF}})\Big] u_n^{\rm HF}, \cL u_n^{\rm HF}\right>_\HL\\
&=\Re  \left<\Big[\cL,\Big[\cL, (-V+W_{1,\gamma_*^{\rm HF}})\Big] \Big] u_n^{\rm HF},  u_n^{\rm HF}\right>_\HL+\Re \left<\Big[\cL, (-V+W_{1,\gamma_*^{\rm HF}})\Big] \cL u_n^{\rm HF},  u_n^{\rm HF}\right>_\HL\\
&=\Re\left<\Big[\cL,\Big[\cL, (-V+W_{1,\gamma_*^{\rm HF}})\Big] \Big]u_n^{\rm HF},  u_n^{\rm HF}\right>_\HL-\Re \left<\Big[\cL, (-V+W_{1,\gamma_*^{\rm HF}})\Big]  u_n^{\rm HF}, \cL  u_n^{\rm HF}\right>_\HL.
\end{align*}
Then,
\begin{align*}
   \MoveEqLeft \Re \left<\Big[\cL, (-V+W_{1,\gamma_*^{\rm HF}})\Big] u_n^{\rm HF}, \cL u_n^{\rm HF}\right>_\HL\\
    &=\frac{1}{2}\Re\left<\Big[\cL,\Big[\cL, (-V+W_{1,\gamma_*^{\rm HF}})\Big]\Big]u_n^{\rm HF},  u_n^{\rm HF}\right>_\HL=\frac{1}{2}\left<\Big[\cL,\Big[\cL, (-V+W_{1,\gamma_*^{\rm HF}})\Big]\Big]u_n^{\rm HF},  u_n^{\rm HF}\right>_\HL.
\end{align*}
Thus,
\begin{align}\label{eq:7.7'}
   \MoveEqLeft  -\sum_{n=1}^q\Re \left<(-V+W_{1,\gamma_*^{\rm HF}}) u_n^{\rm HF}, \cL^2 u_n^{\rm HF}\right>_\HL\notag\\
     &=-\frac{1}{2}\sum_{n=1}^q\left<\Big[\cL,\Big[\cL, (-V+W_{1,\gamma_*^{\rm HF}})\Big] \Big]u_n^{\rm HF},  u_n^{\rm HF}\right>_\HL\notag\\
     &\quad-\sum_{n=1}^q \left<\cL u_n^{\rm HF}, (-V+W_{1,\gamma_*^{\rm HF}})\cL u_n^{\rm HF}\right>_{\HL}.
\end{align}
In addition, for any potential $\widetilde{V}$, we have formally
\begin{align}\label{eq:[L,[L,V]]}
    [\cL,[\cL, \widetilde{V}]]u&=-\Delta(\widetilde{V} u)-2\cL ( \widetilde{V} \cL u)+ \widetilde{V}(-\Delta u)\notag\\
    &=(-\Delta \widetilde{V})u-2\nabla \widetilde{V}\cdot \nabla u-2(\cL V)\cL u\notag\\
    &=(-\Delta \widetilde{V}) u+2i\pmb \sigma\cdot\Big((\nabla \widetilde{V}) \times \nabla\Big)u
\end{align}
where  we recall $\cL=-i\pmb \sigma\cdot \nabla$, the notation ``$\times$'' is the cross product and in the last equation we used the fact that for any vector $\pmb a,\pmb b\in \R^3$,
\begin{align*}
 (\pmb \sigma\cdot\pmb a)(\pmb \sigma \cdot \pmb b)=\pmb a\cdot \pmb b +i\pmb \sigma \cdot (\pmb a\times \pmb b ).
\end{align*}
Thus, 
\begin{align}\label{eq:8.6}
   \MoveEqLeft  -\sum_{n=1}^q\Re \left<(-V+W_{1,\gamma_*^{\rm HF}}) u_n^{\rm HF}, \cL^2 u_n^{\rm HF}\right>_\HL\notag\\
     &=\frac{1}{2}\sum_{n=1}^q\left<u_n^{\rm HF}, \Big[\Delta (-V+W_{1,\gamma_*^{\rm HF}})\Big]  u_n^{\rm HF}\right>_\HL\notag\\
     &\quad +\frac{1}{2}\sum_{n=1}^q\left<u_n^{\rm HF},\pmb\sigma \cdot \Big[ (-\nabla V+\nabla W_{1,\gamma_*^{\rm HF}})\times (-i\nabla)\Big]  u_n^{\rm HF}\right>_\HL\notag\\
     &\quad-\sum_{n=1}^q \left<\cL u_n^{\rm HF}, (-V+W_{1,\gamma_*^{\rm HF}})\cL u_n^{\rm HF}\right>_{\HL}.
\end{align}

\medskip

\noindent{\bf Term with $W_{2,\bullet}$.} Finally, we study terms associated with $W_{2,\bullet}$. We have
\begin{align*}
 \MoveEqLeft   \sum_{n=1}^q \Re\left<W_{2,\gamma_*^{\rm HF}}u_n^{\rm HF}, \cL^2 u_n^{\rm HF}\right>\\
    &=   \sum_{m,n=1}^q \Re\int_{\R^3} (u_n^{\rm HF})^*(x) u_m^{\rm HF}(x)  \left<W(x-\cdot) u_m^{\rm HF}, \cL^2 u_n^{\rm HF}\right>_{\HL}dx\\
    &=\sum_{m,n=1}^q \Re\int_{\R^3} (u_n^{\rm HF})^*(x) u_m^{\rm HF}(x)\left<W(x-\cdot) \cL u_m^{\rm HF}, \cL u_n^{\rm HF}\right>_{\HL}dx\\
    &\quad +\sum_{m,n=1}^q \Re\int_{\R^3} (u_n^{\rm HF})^*(x) u_m^{\rm HF}(x)\left<\Big[\cL, W(x-\cdot)\Big] u_m^{\rm HF}, \cL u_n^{\rm HF}\right>_{\HL} dx.
\end{align*}
Analogous to \eqref{eq:7.7'}, we have
\begin{align*}
  \MoveEqLeft  \sum_{m,n=1}^q \Re\int_{\R^3} (u_n^{\rm HF})^*(x) u_m^{\rm HF}(x)\left< \Big[\cL, W(x-\cdot)\Big] u_m^{\rm HF},\cL u_n^{\rm HF}\right>_{\HL} dx\\
  &= \frac{1}{2}\sum_{m,n=1}^q \int_{\R^3} (u_n^{\rm HF})^*(x) u_m^{\rm HF}(x)\left< \Big[\cL, \Big[\cL, W(x-\cdot)\Big]\Big] u_m^{\rm HF}, u_n^{\rm HF}\right>_{\HL} dx.
\end{align*}
Thus by \eqref{eq:[L,[L,V]]},
\begin{align}\label{eq:8.7}
     \MoveEqLeft   \sum_{n=1}^q \Re\left<W_{2,\gamma_*^{\rm HF}}u_n^{\rm HF}, \cL^2 u_n^{\rm HF}\right>\notag\\
     &=\sum_{m,n=1}^q \Re\int_{\R^3} (u_n^{\rm HF})^*(x) u_m^{\rm HF}(x)\left<W(x-\cdot) \cL u_m^{\rm HF}, \cL u_n^{\rm HF}\right>_{\HL}dx\notag\\
    &\quad +\frac{1}{2}\sum_{m,n=1}^q\int_{\R^3} (u_n^{\rm HF})^*(x) u_m^{\rm HF}(x)\left< \Big[\cL, \Big[\cL, W(x-\cdot)\Big]\Big] u_m^{\rm HF}, u_n^{\rm HF}\right>_{\HL} dx\notag\\
    &=\sum_{m,n=1}^q \int_{\R^3} (u_n^{\rm HF})^*(x) u_m^{\rm HF}(x)\left<W(x-\cdot) \cL u_m^{\rm HF}, \cL u_n^{\rm HF}\right>_{\HL}dx\notag\\
    &\quad -\frac{1}{2}\sum_{m,n=1}^q \int_{\R^3} (u_n^{\rm HF})^*(x) u_m^{\rm HF}(x)\left< u_m^{\rm HF}, \Big[\Delta_y W(x-\cdot)\Big] u_n^{\rm HF}\right>_{\HL} dx\notag\\
    &\quad- \frac{1}{2}\sum_{m,n=1}^q \int_{\R^3} (u_n^{\rm HF})^*(x) u_m^{\rm HF}(x)\left<  u_m^{\rm HF},\pmb \sigma\cdot \Big[(\nabla_y W(x-\cdot))\times (-i\nabla)\Big] u_n^{\rm HF}\right>_{\HL} dx.
\end{align}

\medskip

\noindent{\bf Conclusion.} From \eqref{eq:8.5}-\eqref{eq:8.7}, we conclude that
\begin{align*}
\MoveEqLeft   - \sum_{n=1}^q \lambda_{n}^{\rm HF}\left<\cL u_{n}^{\rm HF},\cL  u_n^{\rm HF}\right>_\HL\\
&=E_{\rm mv}+E_{\rm D}+E_{\rm so}-\sum_{n=1}^q\left<\cL u_n^{\rm HF}, (-V+W_{1,\gamma_*^{\rm HF}})\cL u_n^{\rm HF}\right>_{\HL}\\
     &\quad+ \sum_{m,n=1}^q\int_{\R^3}  (u_n^{\rm HF})^*(x) u_m^{\rm HF}(x) \left<\cL u_m^{\rm HF}, |x-\cdot|^{-1}\cL u_n^{\rm HF}\right>_{\HL} dx.
\end{align*}
Inserting this into the formula of $E^{(2)}_c$, we get
\begin{align*}
    4c^2 E^{(2)}_c= E_{\rm mv}+E_{\rm D}+E_{\rm so}.
\end{align*}
This ends the proof.

%%%%%%%%%%%%%%%%%%%%%%%%%%%%%%
 
\appendix

\section{Some technical estimates}\label{sec:A}
In this section, we list some basic estimates used in this paper taken from \cite{sere2023new,meng2024rigorous}. The difference is only because of the change of units for $Z$, $\alpha$ and $c$.
\begin{lemma}\label{lem:ope}
Let $\gamma\in X$. 
\begin{enumerate}
    \item 
    \begin{align}\label{eq:5.1-1}
    \|W_{\gamma}\|_{\mathcal{B}(\mathcal{H})}\leq \frac{\pi}{2}\|\gamma\|_X\leq \frac{\pi}{2 c}\|\gamma\|_{X_c}    
    \end{align}
    \item 
    \begin{align}\label{eq:5.1-2}
        \|W_{\gamma}u\|_{\mathcal{H}}\leq 2\|\gamma\|_{\mathfrak{S}_1}\|\nabla u\|_{\mathcal{H}}\leq \frac{2\|\gamma\|_{\mathfrak{S}_1}}{c}\||\cD|^{1/2}u\|_{\mathcal{H}}.
    \end{align}
    \item  Let $\gamma\in \Gamma_{q}$ and $\kappa_c<1$. Then
    \begin{align}\label{eq:5.1-3}
        (1-\kappa_c)^2|\cD|^2\leq |\mathcal{D}^c_{\gamma}|^2\leq  (1+\kappa_c)^{2}|\cD|^2.
    \end{align}
As a result,
\begin{align}\label{eq:D-D}
        (1-\kappa_c)|\cD|\leq |\mathcal{D}^c_{\gamma}|\leq  (1+\kappa_c)|\cD|.
    \end{align}
\item Let $\gamma\in \Gamma_{q}$, we have
\begin{equation}\label{eq:DP}
    \||\cD|^{1/2}P^{\pm}_{\gamma}u\|_{\mathcal{H}}\leq \frac{(1+\kappa_c)^{1/2}}{(1-\kappa_c)^{1/2}}\||\cD|^{1/2} u\|_{\mathcal{H}}.
\end{equation}
\item Let $\gamma\in \Gamma_{q}$ and $\max(q,Z)<\frac{2}{\pi/2+2/\pi}$, then
\begin{align}\label{eq:5.1-5}
    \inf|\sigma(\mathcal{D}^c_{\gamma})|\geq c^2\lambda_{0,c}(\alpha,c)=c^2(1-\max(\alpha_c q,Z_c)).
\end{align}
\item Let $h\in X^2$ and $\gamma\in \Gamma_{q}$, then 
\begin{align}\label{eq:W-cD}
    \|[W_h,\cD_{\gamma}]\|_{\cB(\cH)}\leq 16c(1+\kappa_c)\|h\|_{X^2}+c^2\|[W_h,\beta]\|_{\cB(\cH)}.
\end{align}
\end{enumerate}
\end{lemma}
\begin{proof}
Estimates \eqref{eq:5.1-1}-\eqref{eq:5.1-5} can be found in \cite[Lemma 2.6]{sere2023new}. Here they are a direct copy of \cite[Lemma A.1]{meng2024rigorous}. The last estimate \eqref{eq:W-cD} is a modification of \cite[Lemma 5.5]{meng2024rigorous}.
\end{proof}

{\bf Declarations.}\medskip

{\bf Competing Interests and Funding} The authors declare no conflict of interest.\medskip

{\bf Data Availability} My manuscript has no associated data.

{\bf Acknowledgments.}: Supports by the Deutsche
Forschungsgemeinschaft (DFG, German Research Foundation) through TRR 352 -- Project 470903074 and by ERC CoG RAMBAS – Project-Nr. 10104424 are acknowledged. \medskip

\medskip
\begin{refcontext}[sorting=nyt]
\printbibliography[heading=bibintoc, title={Bibliography}]
\end{refcontext}
\end{document}